\documentclass[aps, pra, showpacs, notitlepage, a4paper, onecolumn, superscriptaddress, nofootinbib, accepted=2025-06-10]{quantumarticle}
\pdfoutput=1
\usepackage[utf8]{inputenc}
\usepackage[tmargin=1in, bmargin=1in, lmargin=1in, rmargin=1in]{geometry}
\usepackage{amsmath, amssymb, amsthm}
\usepackage{graphicx}
\usepackage[svgnames]{xcolor}
\usepackage{enumitem}
\usepackage{titlesec}
\usepackage{datetime}
\usepackage{hyperref}
\usepackage{import}
\usepackage{mathtools}
\usepackage{mdframed}
\usepackage{tikz}
\usepackage{tikz-cd}
\usepackage{thmtools,thm-restate}
\usepackage{pifont}
\usepackage{blkarray}

\definecolor{crimson}{RGB}{186,0,44}
\definecolor{moss}{RGB}{0, 186, 111}
\newcommand{\pop}[1]{\textcolor{crimson}{#1}}

\newcommand{\pleq}{\preccurlyeq}

\hypersetup{
    colorlinks,
    linkcolor={crimson},
    citecolor={crimson},
    urlcolor={crimson}
}

\tikzset{
    every node/.style={font=\sffamily\small},
    main node/.style={thick,circle ,draw},
    visible node/.style={thick,rectangle ,draw}
}
\usepackage{qcircuit}

\makeatletter
\renewcommand\@endtheorem{\vvv@endmarker\endtrivlist\@endpefalse}
\newcommand\vvv@endmarker{%
  {\unskip\nobreak\hfil\penalty50
  \hskip2em\vadjust{}\nobreak\hfil\openbox
  \parfillskip=0pt \finalhyphendemerits=0 \par
  \penalty 10000 \parskip=0pt\noindent}\ignorespaces}
\makeatother

\theoremstyle{definition}
\newtheorem{definition}{Definition}[section]
\newtheorem{lemma}{Lemma}[section]
\newtheorem{theorem}{Theorem}[section]
\newtheorem{corollary}{Corollary}[theorem]
\newtheorem*{theorem*}{Theorem}
\newtheorem*{corollary*}{Corollary}
\newtheorem{remark}{Remark}[section]

\newtheorem{example}{Example}[section]

\usepackage{geometry}
\begin{document}

\title{Modular quantum signal processing in many variables}

\author{Zane M. Rossi}
\thanks{This author contributed equally. \pop{zmr@mit.edu/zmr@g.ecc.u-tokyo.ac.jp}}
\affiliation{%
Department of Physics, Massachusetts Institute of Technology, Cambridge, Massachusetts 02139, USA}
\author{Jack L. Ceroni}
\thanks{This author contributed equally. \pop{jack.ceroni@mail.utoronto.ca/jceroni@uchicago.edu}}
\affiliation{%
Department of Mathematics, University of Toronto, Toronto, Ontario M5S1A1, Canada}
\author{Isaac L. Chuang}
\affiliation{%
Department of Physics, Massachusetts Institute of Technology, Cambridge, Massachusetts 02139, USA}

\date{June 10, 2025}


\begin{abstract}
    \noindent Despite significant advances in quantum algorithms, quantum programs in practice are often expressed at the circuit level, forgoing helpful structural abstractions common to their classical counterparts.
    Consequently, as many quantum algorithms have been unified with the advent of quantum signal processing (QSP) and quantum singular value transformation (QSVT), an opportunity has appeared to cast these algorithms as modules that can be combined to constitute complex programs.
    Complicating this, however, is that while QSP/QSVT are often described by the polynomial transforms they apply to the singular values of large linear operators, and the algebraic manipulation of polynomials is simple, the QSP/QSVT protocols realizing analogous manipulations of their embedded polynomials are non-obvious.
    Here we provide a theory of modular multi-input-output QSP-based superoperators, the basic unit of which we call a \emph{gadget}, and show they can be snapped together with LEGO-like ease at the level of the functions they apply.
    To demonstrate this ease, we also provide a Python package for assembling gadgets and compiling them to circuits.
    Viewed alternately, gadgets both enable the efficient block encoding of large families of useful multivariable functions, and substantiate a functional-programming approach to quantum algorithm design in recasting QSP and QSVT as monadic types.
\end{abstract}


\maketitle

\vspace*{-2em}

\section{Introduction} \label{sec:introduction}

\noindent Quantum algorithms, persistently strange, remain difficult to design and interpret \cite{aaronson_15}; correspondingly, great effort has been spent to not only generate algorithms, but formalize the motifs of quantum advantage \cite{childs_10, montanaro_16, aaronson2022much}. Both of these desires have been partially addressed by the advent of quantum signal processing (QSP) \cite{ylc_fixed_point_14, lyc_equiangular_16, lc_ham_sim_17, lc_qubitization_19} and its lifted version quantum singular value transformation (QSVT) \cite{gslw_qsvt_19}. These algorithms allow one to modify the singular values of large linear operators by precisely tunable polynomial functions, unifying and simplifying most known quantum algorithms \cite{mrtc_unification_21}, with good numerical properties \cite{haah_product_decomp_19, chao_machine_prec_20, dmwl_efficient_qsp_phases_21, dlnw_infinite_qsp_22}, and deep, fruitful connections to well-studied matrix factorization techniques \cite{cs_qsvt_tang_tian_23}. 

In reducing diverse algorithmic problems to statements on the existence of classes of polynomial functions, QSP and QSVT encourage but do not substantiate interpreting quantum computations \emph{functionally}. That is, it becomes tempting to treat these algorithms purely in terms of the functions they apply, directly working within the natural ring over polynomials (generated by addition and multiplication), or the monoid associated with polynomial composition. Such \emph{function-first} operations have clear semantic interpretation, offer connection to a wealth of preëxisting mathematical literature, and connect neatly to existing classical formalisms for the design, verification, and optimization of programs \cite{selinger_qpl_04, selinger_higher_order_04, gay_qpls_06}.

Limited techniques for treating QSP and QSVT function-first have appeared, answering the following narrow question: given repeatable oracle access to a QSP protocol, and a description of another QSP protocol, can one instantiate the protocol achieving \emph{the composition of the functions achieved by each protocol} in black-box way? Independent lines of work \cite{rc_semantic_alg_23, mf_recursive_23} have answered this positively, with the former enumerating necessary and sufficient conditions under which such composition is possible, and the latter examining recursive composition converging to a useful class of functions. In both works, however, only a single oracular unitary process is considered, and the resulting computations are described only by compositions of polynomials in a single variable, limiting the variety of achievable behavior and concomitant algorithmic use. We also note that such limited methods for self-embedding simple quantum subroutines has a long history \cite{grover_05, ylc_fixed_point_14, lyc_optimal_pulses_14, jones_nested_not_13, hkj_comp_pulse_analysis_13}, with roots in composite pulse theory \cite{tkvn_matched_phase_grover_09, kv_passband_nesting_13} and classical signal processing \cite{kh_sharpening_77, saramaki_cascade_87}, albeit without clear functional interpretation. Ultimately, a mature, function-first theory, closed over operations natural to multivariable polynomials (e.g., ring operations, the composition monoid) has not been realized.

In moving to the multiple oracle setting, we are seeking a more ambitious construction than previous work; that is, we want to be able to blithely snap together basic modules, with structure similar to that of QSP, at the level of the functions they apply. We will show this is surprisingly possible, with a bit of additional work, through a careful synthesis of two previous techniques: multivariable QSP (M-QSP) \cite{rossi_m_qsp_22} and an intricate series of powerful results on fractional queries from Sheridan, Maslov, and Mosca \cite{smm_09} extended beautifully to QSVT by Gilyén, Su, Low, and Weibe \cite{gslw_qsvt_19}. In turn, these M-QSP-based examples will then be generalized to a more encompassing theory of \emph{gadgets}. Ultimately, we will be able to succinctly describe and subsume previous work on modular QSP-based algorithms \cite{rc_semantic_alg_23, mf_recursive_23}, the most obvious bridge between our work and previous work on the level of the aforementioned multivariable quantum signal processing (M-QSP) \cite{rossi_m_qsp_22}. M-QSP (Def.~\ref{def:m_qsp_func_prog}), despite having a well-defined series of conditions under which a transform is achievable, has proven unwieldy in practice, as multivariable analogues of the algebraic geometric results characterizing achievable polynomial transforms (chiefly the Fejér-Riesz Theorem \cite{polya_szego_analysis_98, mv_frt_04}) are unavoidably weaker in the multivariable setting. We will show that such barriers, however, do not rule out the existence of independently interesting subroutines strictly subsuming M-QSP, and which allow one to freely manipulate and combine polynomial functions. That is, to yield the same functional class as the one originally desired but disallowed in standard M-QSP, we show it is sufficient to generate algebraic structures over polynomials with rich closures; constructively achieving such closures by expanding the class of circuits considered is a central result of this paper.

The fundamental units of computation we construct are termed \emph{gadgets}, which take as input (possibly multiple) unknown unitary processes, and produce as output (possibly multiple) unitary processes, where outputs depend entirely on multivariable polynomials applied to parameters of the inputs. We highlight an important subclass of gadgets, \emph{atomic gadgets}, which are built directly using M-QSP, and subsequently combined with rich fractional query techniques \cite{smm_09, gslw_qsvt_19} to realize \emph{snappable gadgets}, which, having been suitably \emph{corrected}, can be freely combined. We provide a characterization of the \emph{syntax} for combining gadgets into composite gadgets (described by a two-level grammar, which under basic restrictions becomes context-free), as well as a description of valid \emph{semantic} manipulations for functions achieved by these gadgets (in the form of a monadic type over M-QSP protocols). In providing both a \emph{syntactic} and \emph{semantic} theory of gadgets, we provide a formal response to the temptation to view QSP- and QSVT-based modules purely functionally, and embed our constructions within established terminology from functional programming and category theory.

Unlike in the single-oracle setting, the conditions under which functional manipulations are possible within gadgets (for which the work of \cite{rc_semantic_alg_23, mf_recursive_23} consider a special case) do not depend solely on simple circuit symmetries; consequently, the bulk of appendices of this work are spent specifying detailed functional analytic arguments on the existence and performance of special \emph{correction protocols} by which the output of a given gadget can be \emph{corrected} such that it can be connected to the input of another gadget while preserving the desired functional manipulation. We show that such gadget connections are generally ill-conditioned without correction, and we prove that correction generically cannot be done exactly, only to arbitrary precision. In fact, it is this softened condition of approximate achieveability that enables our variety of functional transforms. We emphasize that unlike with other methods for generating composite superoperators, like linear combination of unitaries (LCU), we maintain both stringent space requirements and rapid convergence (polylogarithmic in inverse error) to intended functional transforms. Constructing these correction subroutines constitutes the core difficulty of this work---however, once constructed, these protocols allow one to immediately characterize the space and query complexity of gadgets, placing them in comparison with other block-encoding manipulation techniques \cite{gslw_qsvt_19, bck_ham_sim_15}.

The results of this work can be posed a few ways, in order of increasing abstraction. At a practical level we achieve highly space- and query-efficient block encodings of multivariable polynomials in commuting linear operators, showing that QSVT can retain key advantages over alternative methods like linear combination of unitaries (LCU) \cite{bck_ham_sim_15, bss_commuting_matrices_23} in the multivariable setting (this fact is detailed in Appx.~\ref{appx:performance_comparisons}). At a higher level, this work solidifies and expands the approach of \cite{rc_semantic_alg_23} in treating QSP and QSVT protocols in the language of functional programming through the composite objects of \emph{gadgets}. I.e., we frame our results as instantiating a \emph{monad} over QSP and QSVT \emph{types}, and situate this monad in quantum programming language theory, with possible compositions described by simple formal grammars. 

To better support our claims we provide a series of example constructions of achievable multivariable functions (see Fig.~\ref{fig:fig_permitted_functions}), as well as a Python package, \texttt{pyqsp}, for automatically assembling gadgets and compiling their circuits located at \href{https://github.com/ichuang/pyqsp/tree/beta}{\texttt{[https://github.com/ichuang/pyqsp/tree/beta]}}. Note that this package (in the \texttt{beta} branch) is compatible with phases computed by the (recently improved) \texttt{main} branch, but these branches are currently un-merged given competing conventions used (for legacy reasons) in some sub-modules of \texttt{main}. The current contact for and maintainer of this repository is the first listed author.

Our presentation is divided in two components: the main body Secs.\ref{sec:introduction}-\ref{sec:discussion_conclusion}, and appendices \ref{appx:main_proofs}-\ref{appx:functional_programming}. In the main body, we provide the construction for \emph{gadgets} (Sec.~\ref{sec:constructing_gadgets}) (of which there are sub-types, Appx.~\ref{sec:gadget_types}), define and discuss the connection of these gadgets into composite gadgets (Sec.~\ref{sec:gadget_composition}), and provide a series of linked concrete examples of useful functions achieved by assembling gadgets (Sec.~\ref{sec:examples}), with detailed cost analysis (Appx.~\ref{sec:gadget_cost}). Secondly, we provide a series of appendices which cover proofs of theorems stated in the main body (Appx.~\ref{appx:main_proofs}), detailed analysis of the aforementioned correction protocols (Appxs.~\ref{appx:extraction} and \ref{appx:roots}), explanation of the non-obvious method for computing gadget cost (Appx.~\ref{sec:gadget_cost}), discussion of how gadgets instantiate monadic functions and relate to functional programming (gadget \emph{semantics}, Appxs.~\ref{appx:natural_transformations} and \ref{appx:qsp_qsvt_types}), and finally discussion of a formal language for gadgets constituted by an attribute grammar (gadget \emph{syntax}, Appx.~\ref{appx:formal_gadget_grammar}).


\section{Motivating and constructing gadgets}
\label{sec:constructing_gadgets}

\noindent We have claimed that gadgets, properly constructed, can enable the efficient achievement of multivariable polynomial transformations and a modular approach to assembling quantum algorithms. This is shown explicitly in this section, by building from M-QSP and its combination with fractional queries to create \emph{atomic gadgets}, a remarkably useful subclass of gadgets. We use and extend fractional query techniques \cite{smm_09, gslw_qsvt_19} to allow certain imperfections in atomic gadgets to be corrected, leading the fundamental unit of our construction, the \emph{snappable gadget}, which we suitably generalize beyond explicit reference to QSP.
For the moment, however, we recall that QSP circuits prescribe how to take simply parameterized unknown unitary inputs and produce as output unitaries with properties non-linearly dependent on, and carefully tuneable with respect to, the input unitaries \cite{gslw_qsvt_19}. Crucially, within QSP, input unitaries are highly constrained (in fact rotations about a fixed, known axis), while the output unitiaries are not; even worse, if one were to consider a QSP-like ansatz depending on multiple unknown unitaries as input, not only can the form of the output unitary depend \emph{on the function applied}, but also on the \emph{relation between each of the many unknown inputs}. At a high level it is this \emph{mismatch of basis} between input and output unitaries in QSP (and inherited in singular vector subspaces in QSVT) which prevents simple composition of protocols at the level of the functions they achieve.

The core result of this work (Thm.~\ref{thm:polynomial_gadget_equivalence}) lies in a terse equivalence between general compositions of multivariable polynomials and general compositions of gadgets (themselves based on QSP-like components) achieving these polynomials. Toward gathering the components to state this equivalence, we first build off of a QSP-related circuit to generate a powerful subclass of gadgets, followed by identifying key issues in their compositions, and special protocols to resolve these issues.

To start off, we recall the general definition of an \emph{M-QSP protocol} \cite{rossi_m_qsp_22}, of which a standard QSP protocol is a special case. For the moment we will not need to know any properties of this algorithm (covered in Appx.~\ref{appx:variants_qsp} for the curious) besides its simple circuit form, other than to note that the problems M-QSP protocols exhibit when we attempt to combine them as modules will showcase shortcomings eventually resolved by the definition of \emph{gadgets} and \emph{correction protocols} given later in this section. These problems also distinguish the general setting from the single variable setting, where the former is substantively more difficult to address.

\begin{definition}[M-QSP protocol] \label{def:m_qsp_func_prog}
    Let $\Phi \equiv \{\phi_0, \dots, \phi_n\} \in \mathbb{R}^{n + 1}$ and $s \in [t]^{n},\, t \in \mathbb{N}$, where $[t] = \{0, \dots, t - 1\}$. Then the $t$-variable M-QSP protocol specified by $(\Phi, s)$ generates the unitary circuit
    	\begin{equation}
            \label{eq:mqsp_protocol}
    		\Phi[U_0, \cdots, U_{t - 1}] \equiv e^{i\phi_0\sigma_z}\prod_{k = 1}^{n} U_{s_k} e^{i\phi_k\sigma_z},
    	\end{equation}
    where $U_0, \cdots, U_{t-1}$ are unitaries of the form $U_k \equiv e^{i\theta_k\sigma_x}, k \in [t]$ for unrelated $\theta_k$ (see Def.~\ref{def:embeddable}). Here $\sigma_z, \sigma_x$ are the standard Pauli matrices. We say an M-QSP protocol is \emph{antisymmetric} if both $(N \circ R)(\Phi) = \Phi$ and $R(s) = s$, where $R, N$ reverse and negate lists of scalars respectively, with equality taken elementwise (when $\Phi$ contains zeros and the oracles can commute past each other, this definition can be replaced one stating $\Phi[U_0, \cdots, U_{t - 1}] = \Phi[U_0, \cdots, U_{t - 1}]^\dagger$.). Often $s$ will be suppressed when denoting $(\Phi, s)$ as a superoperator (see Appx.~\ref{appx:variants_qsp}). An M-QSP protocol \emph{achieves} a function $f$, where $f \equiv \langle 0 | \Phi[U_0, \cdots, U_{t - 1}] |0\rangle$ is a function over scalars (namely those $x_k \equiv \cos(\theta_k)$) parameterizing the oracles.
\end{definition}

To understand the problem eventually solved by gadgets, we recall one of the insights of previous work on embedding QSP-based circuits within themselves \cite{rc_semantic_alg_23}. In this work, it was shown how to iteratively compose \textit{antisymmetric} single-variable QSP protocols to compose their achieved functions. Such compositions can be realized by successively using the QSP unitary produced by some protocol $\Phi_k$ as the signal operator of another QSP protocol $\Phi_{k + 1}$, continuing perhaps for some collection protocols $\{\Phi_0, \dots, \Phi_{n - 1}\}$. The admissibility of this recursion follows from the fact that for any antisymmetric protocol $\Phi$ and oracle unitary $e^{i \theta \sigma_x}$, the polynomial $P$ achieved by a single-variable QSP protocol $\Phi$ is invariant under \textit{twisting}. More specifically, for any unitary $U$ and any $\sigma_z$-rotation $e^{i \varphi \sigma_z}$, $\Phi[e^{i \varphi \sigma_z} U e^{-i \varphi \sigma_z}] = e^{i \varphi \sigma_z} \Phi[U] e^{-i \varphi \sigma_z}$. Thus for a sequence of antisymmetric QSP protocols $\{\Phi_0, \dots, \Phi_{n - 1}\}$ the $(k + 1)$-th protocol automatically treats the function achieved by the $k$-th protocol as a variable for the function achieved by $\Phi_{k + 1}$. The output of each protocol remains in some sense \emph{embeddable}, and the composability of functions follows. Below, we clarify this notion of embeddability (introducing a few named types).

\begin{definition}[Variations on embeddable unitaries] \label{def:embeddable}
   An $\text{SU}(2)$ unitary $U$ is \textit{embeddable} if it has the form $e^{i \theta \sigma_x}$ for some $\theta \in [0, \pi]$. A unitary is $\varepsilon$-embeddable if it is at most $\varepsilon$-far (in operator norm) from an embeddable unitary.
   An $\text{SU}(2)$ unitary $U$ is \textit{twisted embeddable} (\textit{half-twisted embeddable}) if it has the form $e^{i \varphi \sigma_z/2} e^{i \theta \sigma_x} e^{-i \varphi \sigma_z/2}$ for some $\theta \in [0, \pi]$ and $\varphi \in [-\pi, \pi]$ ($\varphi \in [-\pi/2, \pi/2]$). A unitary $U$ is $\varepsilon$-twisted embeddable ($\varepsilon$-half-twisted embeddable) if it is at most $\varepsilon$-far from a twisted embeddable (half-twisted embeddable) unitary over a specified domain. Given a twisted or half-twisted embeddable unitary of the form $e^{i \varphi \sigma_z/2} e^{i \theta \sigma_x} e^{-i \varphi \sigma_z/2}$, we refer to $e^{i \theta \sigma_x}$ as its \textit{embeddable component}.
\end{definition}

In what follows we show that maintaining such embeddability is vital for a variety of desired computational tasks and discuss some subtleties of the definitions contained in Def.~\ref{def:embeddable}. To begin, we highlight a pedagogical example: consider two antisymmetric, single-variable QSP protocols $\Phi_0$ and $\Phi_1$ and signal operators $e^{i \theta_0 \sigma_x}$, $e^{i \theta_1 \sigma_x}$. By Thm.~\ref{thm:antisymmetric}, $\Phi_0[e^{i \theta_0 \sigma_x}]$ and $\Phi_1[e^{i \theta_1 \sigma_x}]$ are twisted embeddable, with
	\begin{equation}
	    \label{eq:mqsp_example}
	    \Phi_0[e^{i \theta_0 \sigma_x}] = e^{i \varphi_0 \sigma_z/2} e^{i \theta_0' \sigma_x} e^{-i \varphi_0 \sigma_z/2} \quad \text{and} \quad \Phi_1[e^{i \theta_1 \sigma_x}] = e^{i \varphi_1 \sigma_z/2} e^{i \theta_1' \sigma_x} e^{-i \varphi_1 \sigma_z/2}.
	\end{equation}
Suppose $\Phi_0$ and $\Phi_1$ achieve functions $f_0(\cos{(\theta_0)})$ and $f_1(\cos{(\theta_1)})$ as the top-left elements of their unitaries when applied to the signals $\theta_0, \theta_1$ (i.e., by the above equation, we meant that $\theta_0^\prime = \arccos{(f_0(\cos(\theta_0)))}$ and $\theta_1^\prime = \arccos{(f_1(\cos(\theta_1)))}$ for brevity). Then given a two-variable antisymmetric M-QSP protocol $\Phi$ achieving the function $g(x_0, x_1)$, in the general case,
	\begin{equation}
	    \label{eq:mqsp_example_2}
	    \langle 0 | \Phi \left[ \Phi_0[e^{i \theta_0 \sigma_x}], \Phi_1[e^{i \theta_1 \sigma_x}] \right] |0\rangle \neq g \left( f_0(\cos(\theta_0)), f_1(\cos(\theta_1))\right).
	\end{equation}
This is entirely due to the conjugation of the embeddable components of $\Phi_0[e^{i \theta_0 \sigma_x}]$ and $\Phi_1[e^{i \theta_1 \sigma_x}]$ by \textit{different} $\sigma_z$-rotations. Because $\varphi_0 \neq \varphi_1$, it is no longer possible to factor the $\sigma_z$-rotations out of the overall M-QSP protocol, in the same way as the single variable case. In single variable protocols, accumulating $\sigma_z$-conjugations by successively composing antisymmetric QSP protocols had no effect on the function achieved by the QSP polynomial. Here, one finds
	\begin{equation}
	    \Phi[e^{i \varphi \sigma_z} e^{i \theta \sigma_x} e^{-i \varphi \sigma_z}] 
        = e^{i \varphi \sigma_z} \Phi[e^{i \theta \sigma_x}] e^{-i \varphi \sigma_z} 
        \Longrightarrow \langle 0 | \Phi[e^{i \varphi \sigma_z} e^{i \theta \sigma_x} e^{-i \varphi \sigma_z}] |0\rangle 
        = \langle 0 | \Phi[e^{i \theta \sigma_x} ] |0\rangle.
	\end{equation}
Therefore, in the multivariable case, we require a technique for suppressing the disagreeing $\sigma_z$-conjugations of Eq.~\eqref{eq:mqsp_example}, leaving the relevant embeddable components such that there is no interference with the outer M-QSP protocol. Such a correction would turn Eq.~\eqref{eq:mqsp_example_2} into an approximate equality. Consequently we identify a concrete question: is it possible to map a twisted embeddable unitary $U$ with unknown $\varphi$ and $\theta$ to its embeddable component? In other words, can we \emph{align} the axes of rotation for two (or more) unitary oracles in an efficient, black box way? We provide an affirmative answer to this question, up to slightly weakened conditions, and call this process of alignment \emph{correction}.

\begin{restatable}[Efficient correction of $\varepsilon$-twisted embeddable unitaries]{lemma}{correctingunitaries} \label{lemma:embeddable_correction}
    Let $U$ be a $\nu$-twisted embeddable unitary, so $U$ is $\nu$-close to $U^\prime = e^{i \varphi \sigma_z/2} e^{i \theta \sigma_x} e^{- i \varphi \sigma_z/2}$. Suppose $\cos(\theta) \in [-1 + \delta, 1 - \delta]$. Let $\varepsilon > 0$. Then, given controlled access to $U$, there exists a quantum circuit using $\zeta = \mathcal{O}(\delta^{-1} \log(\varepsilon^{-2} \delta^{-1/2}))$ black-box calls to $U$ and a single ancilla qubit which yields a $(\zeta\nu + \varepsilon)$-embeddable unitary $V$, which is at most $(\zeta\nu + \varepsilon)$-far from $e^{i \theta \sigma_x}$ (the embeddable component of $U'$) with success probability at least $(1 - (\zeta\nu + \varepsilon))^2$.
    Alternatively, given access to two ancilla qubits and uncontrolled oracle access to $U$, it is possible to implement a circuit which achieves the same $(\zeta\nu + \varepsilon)$-approximation with the same asymptotic query complexity and success probability, on the domain $\cos(\theta_k) \in [\delta, 1 - \delta]$.
    Finally, if $U$ is $\nu$-half-twisted embeddable, and satisfies the same conditions, with the additional constraint that $\cos(\varphi) \in [\gamma, \sqrt{1 - \gamma^2}]$ (or $\varphi \in [-\pi/2 + 2\gamma, -2\gamma] \cup [2\gamma, \pi/2 - 2\gamma]$), there exists a quantum circuit using $\zeta' = \mathcal{O}(\delta^{-1} \gamma^{-2} \log^2(\gamma^{-4} \varepsilon^{-2} \delta^{-1/2}))$ black-box calls to $U$ and zero ancilla qubits which yields an $(\zeta' \nu + \varepsilon)$-embeddable unitary, at most $(\zeta' \nu + \varepsilon)$-far from the embeddable component $e^{i \theta \sigma_x}$, with unit success. In all cases, the circuits effectuating the corrections can be efficiently and constructively specified.
\end{restatable}

Lem.~\ref{lemma:embeddable_correction} (proven in Appx.~\ref{appx:main_proofs}) shows that given an approximately embeddable unitary, (1) the error accumulated during its correction is proportional to the error of the input, and (2) the constant of proportionality is linear in the length of the correction procedure. This scaling as not as poor as it first appears, mainly due to the ease with which $\nu$ can be made small and the limited number of correction procedures that will be used by any physically reasonable protocols. We call such protocols \emph{acceptable} if their length scales at worst inverse polynomially in the desired output error, and argue for the reasonableness of this class of protocols in the context of the efficiency of Lem.~\ref{lemma:embeddable_correction} in the expanded Rem.~\ref{rem:acceptable_efficiency}. The resource costs for the correction procedures under different access model assumptions is elaborated upon further in Appxs.~\ref{appx:performance_comparisons} and \ref{appx:gadget_compositions} (we also provide a high level summary in Table~\ref{tab:correction_comparison}, though proofs of all cited complexities are mostly relegated to the appendices mentioned above).

\begin{table}[htpb]
    \centering
    \begin{tabular}{l | l l l l l}
        Method/Access & Complexity \ \ \; & Anc. \ \ & $\cos(\theta)$-dom. \ \ \ & $\cos(\varphi)$-dom. \ \ \ \ & Succ. \\\hline
        Anc./controlled \ & $\widetilde{\mathcal{O}}(\delta^{-1} \log(\varepsilon^{-1}))$ \ \ & 1 & $[-1 + \delta, 1 - \delta]$ & $[-1, 1]$ & $1 - \varepsilon$
        \\
        No anc./un-controlled  & $\widetilde{\mathcal{O}}(\delta^{-1} \gamma^{-2} \log^2(\varepsilon^{-1}))$ \ \ & 0 & $[-1 + \delta, 1 - \delta]$ & $[\gamma, \sqrt{1 - \gamma^2}]$ \ \ & $1$
        \\ 
        Anc./un-controlled & $\widetilde{\mathcal{O}}(\delta^{-1} \log(\varepsilon^{-1}))$ & 2 & $[\delta, 1 - \delta]$ & $[-1, 1]$ & $1 - \varepsilon$
    \end{tabular}
    \caption{Comparison of the resource costs for three different correction procedures which apply under different assumptions about the oracle access model, using $\theta, \phi$ conventions as in Lem.~\ref{lemma:embeddable_correction}. These complexities are mainly analyzed in Appx.~\ref{appx:performance_comparisons}; the distinguishing factor is whether subroutines are provided as coherently controlled quantum processes (presupposing access to clean ancillary qubits) or black-box, after which one can either controllize them using additional space under $\theta$-restrictions, or perform correction in place (under $\phi$-restrictions according to a factor $\gamma$, introduced and discussed in Appx.~\ref{appx:root_without_ancilla}).}
    \label{tab:correction_comparison}
\end{table}

The correction procedure described enables us to take M-QSP protocols, apply corrective protocols to their output, and pass these outputs into other M-QSP protocols, such that the function achieved is the multivariate composition of the functions achieved by the individual protocols. The closure over unitary superoperators achieved by arbitrarily composing sequences of M-QSP unitaries and corrections is formalized by the notion of a \emph{gadget}. Below, we give two definitions for gadgets; the first, the \emph{atomic gadget} (Def.~\ref{def:qsp_gadget}), is built directly from M-QSP protocols, while the second, the \emph{gadget} (Def.~\ref{def:qsp_gadget_general}) captures the most general relevant definition for our purposes. Under special conditions (discussed in detail in Sec.~\ref{sec:gadget_types}), atomic gadgets form a strict subset of gadgets, and we will work with such atomic gadgets exclusively. In the following section, we show that such gadgets can be successively composed to yield composite gadgets describing complex, expressive quantum computations and achieving general and interpretable functional transforms. Moreover, gadget objects have a highly-interpretable form; we describe their permitted combinations described by attribute grammars (a \emph{syntax} for gadgets), and describe their functional action by monadic functions over QSP/QSVT types (a \emph{semantics} for gadgets).

\begin{definition}[Atomic $(a, b)$ gadget] \label{def:qsp_gadget}
    Let $(\Xi, S)$ be a tuple with $\Xi \equiv \{\Phi_0, \dots, \Phi_{b-1}\}$ where $\Phi_k \in \mathbb{R}^{r_k + 1}, r_k \in \mathbb{N}$ and $S \equiv \{s_0, \dots, s_{b-1}\}$ where $s_k \in [a]^{r_k}$ and $[a] = \{0, \dots, a - 1\}$. Then the atomic $(a, b)$ gadget labeled by $(\Xi, S)$ refers to the parameterization of $b$ M-QSP protocols using $a$ single-qubit oracles $U_k$ (by definition embeddable), where each of the $b$ output unitaries $\Phi_k[U_0, \dots, U_{a - 1}]\,, k \in [b]$ has parameterization $(\Phi_k, s_k)$ following Def.~\ref{def:m_qsp_func_prog}. A gadget is \emph{antisymmetric} if its constituting $(\Xi, S)$ generate only antisymmetric M-QSP protocols \cite{rc_semantic_alg_23}. Note that not all\footnote{This may seem like a poor choice of convention, but we leave the antisymmetry requirement out of our definition here because general atomic gadgets can still be useful when restricted to \emph{discrete} signals. In almost all cases however one should assume antisymmetry.} \emph{atomic gadgets} are \emph{gadgets} (Def.~\ref{def:qsp_gadget_general}) (see Sec.~\ref{sec:gadget_types}), though atomic gadgets considered in this work (e.g., antisymmetric gadgets) \emph{will} be.
\end{definition}

\begin{definition}[$(a, b)$ gadget] \label{def:qsp_gadget_general}
    An $(a, b)$ gadget $\mathfrak{G}$ is a superoperator which takes as input $a$ single-qubit embeddable oracles $U_k$ for $k \in [a]$, and outputs $b$ twisted-embeddable unitaries $U_j^\prime$ for $j \in [b]$, denoted by $U_j^\prime = \mathfrak{G}[U_0, \dots, U_{a - 1}]_j$. Note that \emph{antisymmetric} atomic $(a, b)$ gadgets (Def.~\ref{def:qsp_gadget}) are a strict subset of $(a, b)$ gadgets (Thm.~\ref{thm:antisymmetric}). However, there also exist other circuit-based methods for constructing gadgets from M-QSP (Def.~\ref{def:atypical_gadget}). A gadget $\mathfrak{G}$ is said to achieve $F \equiv \{f_0, \dots, f_{b - 1}\}$ where $f_j \equiv \langle 0| U'_j |0\rangle$ is the top-left matrix element of the $j$-th output unitary given by applying $\mathfrak{G}$ to a set of oracles. In the case that the input oracles are parameterized by variables $x_0, \dots, x_{a - 1}$, each output unitary $U_j^\prime$ and each $f_j$ can be treated as functions $U_j^\prime(x_0, \dots, x_{a - 1})$ and $f_j(x_0, \dots, x_{a - 1})$.
\end{definition}

\begin{remark}[Gadget terminology] \label{rem:gen_gadget_functions}
    Going forward, if an $(a, b)$ gadget $\mathfrak{G}$ is said to achieve output unitaries $U_j^\prime(x_0, \dots, x_{a - 1}), j \in [b]$ and functions $f_j(x_0, \dots, x_{b - 1}), j \in [b]$, it is meant that the gadget achieves these functions for input unitaries $U_k = e^{i \theta_k \sigma_x}, k \in [a]$, where each $x_k = \cos(\theta_k), k \in [a]$ is possibly restricted to a fixed domain specified along with the gadget.
\end{remark}

In many cases, a gadget will achieve its output unitaries $U_j'$ via some sequence of products and interspersed quantum (and possibly even classical) operations. The only criteria on the input oracles is that the gadget sees them as black-boxes, which can only be utilized through quantum circuit queries. When we refer to the \textit{cost} associated with a gadget $\mathfrak{G}$, we are referencing the number of black-box queries that must be made to the input oracles to achieve some particular output unitary, possibly at a specified precision over a specified parameter domain. Detailed discussion of cost and its associated objects for gadgets are restricted to Appx.~\ref{sec:gadget_cost}.

We now present the key theorem which describes how the correction protocol of Lem.~\ref{lemma:embeddable_correction} enables \emph{snappable gadgets} to be constructed from general gadgets. Snappable gadgets are engineered such that their outputs are embeddable -- not merely twisted-embeddable -- and may be passed as input to another gadget in a black-box way and while respecting achieved functions. Note that going forward we let $\widetilde{\mathcal{O}}$ indicate asymptotic complexity up to leading order/logarithmic factors in each independent variable.

\begin{restatable}[Efficient correction of unitaries produced by $(a, b)$ gadgets]{theorem}{correctinggadgets} \label{thm:qsp_correction}
    Let $\varepsilon, \delta > 0$ and $\mathfrak{G}$ an $(a, b)$ gadget whose output unitaries are guaranteed to be $\nu$-twisted-embeddable: they are $\nu$-close to unitaries $e^{i \varphi \sigma_z/2} e^{i \theta_k \sigma_x} e^{-i \varphi \sigma_z/2}$ whose embeddable components encode $\cos{\theta}_k, k \in [b]$ within $[-1+\delta, 1-\delta]$. Assuming an access model in which one can query the \emph{controlled} single-qubit unitaries achieved by the output legs of $\mathfrak{G}$, then given $k \in [b]$, there exists a computable $\nu' = \widetilde{\mathcal{O}}(\delta \varepsilon)$ and a quantum circuit using $\zeta = \widetilde{\mathcal{O}}(\delta^{-1}\text{log}{(\varepsilon^{-1}}))$ black-box calls to the controlled unitaries produced by executing $\mathfrak{G}$ on the input oracles, as well as one ancilla qubit, such that if $\nu \leq \nu'$ then the circuit achieves an $\varepsilon$-approximation of $U_k' = e^{i \theta_k \sigma_x}$ (the embeddable component of the $k$-th output unitary of $\mathfrak{G}$), with success probability at least $(1 - \varepsilon)^2$, for all $k$. Alternatively, given access to two ancilla qubits and access only to uncontrolled $U$ queries, the same outcome can be achieved with the same asymptotic complexity and success probability over the domain $\cos(\theta_k) \in [\delta, 1 - \delta]$. Finally, in the case that the output legs are all $\nu$-half twisted embeddable with $\cos(\varphi) \in [\gamma, \sqrt{1 - \gamma^2}]$, the same can be done with no extra space and unit success probability using $\zeta = \widetilde{\mathcal{O}}(\delta^{-1} \gamma^{-2} \log^2(\varepsilon^{-1}))$ black-box calls to $\mathfrak{G}$ and the guarantee that $\nu \leq \nu^{\prime\prime}$, for a computable $\nu^{\prime\prime} = \widetilde{\mathcal{O}}(\delta \gamma^2 \varepsilon)$. Each protocol described above can be constructed efficiently and obliviously to the output unitaries of the gadget.
\end{restatable}

\noindent Essentially this theorem (proven in Appx.~\ref{appx:main_proofs}) applies Lem.~\ref{lemma:embeddable_correction} to each output leg of a gadget, obtaining our desired endpoint for this section: \emph{snappable gadgets}. While the following definition is somewhat innocuous, in the following we show that such gadgets, when assembled, induce natural and simple-to-track algebraic manipulations of their achieved transforms.

\begin{definition}[Snappable gadget] \label{def:snappable_gadget}
    A gadget is said to be $\varepsilon$-$\delta$-\emph{snappable} if each of its output legs, over all unitary inputs with the form of $\sigma_x$-rotations by an angle $\theta$ such that $\delta < |\cos{(\theta)}| < 1 - \delta$, produces a unitary which is $\varepsilon$-close to a $\sigma_x$-rotation in norm.
\end{definition}


\section{Composing QSP gadgets}  \label{sec:gadget_composition}

\noindent The gadgets of Defs.~\ref{def:qsp_gadget} and \ref{def:qsp_gadget_general} and the correction procedure of Thm.~\ref{thm:qsp_correction} bring us to \emph{snappable gadgets}. Such gadgets can finally be freely assembled into composite gadgets to build complex, useful functional transforms; in this section we formally define this assembly, and enumerate its properties.
The desired endpoint of this section, the statement of Thm.~\ref{thm:polynomial_gadget_equivalence}, shows a neat equivalence between a series of (partially) composed polynomials achieved by individual gadgets, and a structured network of those gadgets. This section begins by focusing on the linkages of these networks, captured in Thm.~\ref{thm:full_gadget_composition}, such that Thm.~\ref{thm:polynomial_gadget_equivalence} can be used to snap together a series of example gadgets with LEGO-like ease in Sec.~\ref{sec:examples}.

As mentioned, we show in this section that gadgets permit a simple diagrammatic representation of their semantics, depicted in Fig.~\ref{fig:composing_qsp_gadgets}; in this way, complex functions can be built hierarchically out of simpler ones and visually reasoned about in an intuitive way. More specifically, we can think of so-called $(a, b)$ gadgets as boxes with $a$ \emph{input legs}, representing oracular unitary inputs, and $b$ \emph{output legs}, representing output unitaries (These boxes are shown in (c) of Fig.~\ref{fig:composing_qsp_gadgets} for $a = b = 2$). To flesh out this diagrammatic language, we present a theorem (Thm.~\ref{thm:full_gadget_composition}) enumerating the valid ways to link gadgets together, and translate the effect of such compositions into statements on algebraic manipulations over each gadget's achieved polynomial functions. Diagrammatically, linking gadgets is depicted by joining input and output legs. In this way, as exemplified in the \emph{Rosetta Stone} diagram of Fig.~\ref{fig:composing_qsp_gadgets}, one can freely reason about assemblages of gadgets in terms of any among (a) the circuits realizing them, (b) the functional manipulations they achieve, or (c) as linkages between boxes with semantically simple input and output legs. While not discussed in this section explicitly, the syntax of gadget assemblages follows a formal grammar (Appx.~\ref{appx:formal_gadget_grammar}), while the semantics of achievable functional manipulations is described by the instantiation of a monadic type (Appx.~\ref{appx:qsp_qsvt_types}). We finally note briefly that depicting quantum superoperators as such acyclic graphs is not entirely new, as previous works considering quantum combs (circuits with open slots) have expressed them in terms of simple \emph{causal networks} \cite{kretschmann_networks_05, chiribella_networks_08, chiribella_networks_09}; the difference between those works' general prescription and ours is that our networks depict semantic properties of interactions between achieved functional transforms, and we require careful constraints on allowed unitary inputs. For this reason, as well ease of description for higher-order operations over gadgets defined later, relying on causal networks rather than quantum combs to depict our superoperators leads to an expedient, clarified visual language.

\begin{figure}
    \centering
    \includegraphics[width=0.9\textwidth]{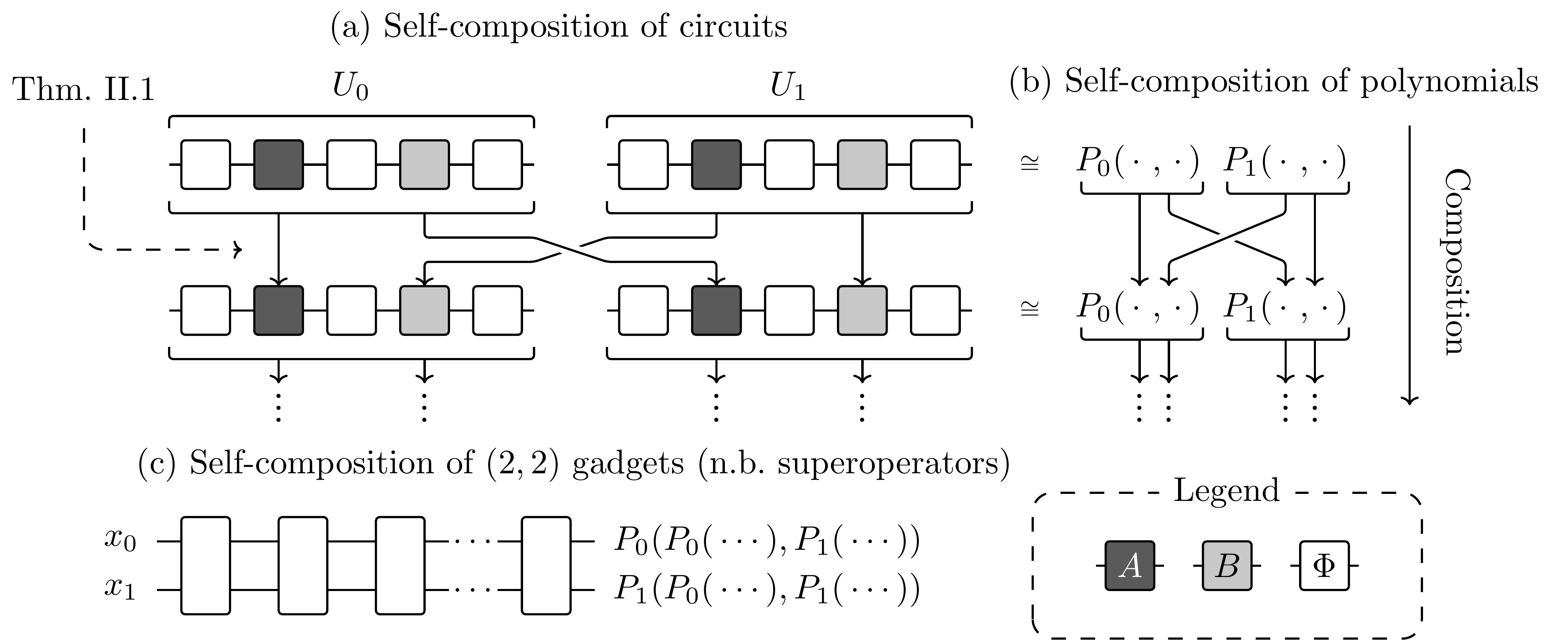}
    \caption{A gadget \emph{Rosetta Stone} giving three depictions of the same self-composition of a $(2,2)$ gadget (Def.~\ref{def:qsp_gadget}). Arrows indicate direct substitution, of (a) quantum circuits (time going left to right), (b) polynomials, and (c) gadget legs, where $A, B$ are slots (places where an oracle unitary can be substituted) and $\Phi$ are programmable, generally distinct $\sigma_z$ rotations. The (Laurent) polynomial transforms achieved by the two protocols, $P_0, P_1$, are real-valued with real arguments, and the correction protocol of Thm.~\ref{thm:qsp_correction} is implicitly interspersed between substitutions (represented explicitly by dots in Fig.~\ref{fig:gadget_interlink}). Note that the circuit (a) and functional (b) depictions are different from the \emph{gadget} (c) depiction given, e.g., in Figs.~\ref{fig:gadget_interlink} and \ref{fig:sum_gadget} and elsewhere; this difference is detailed in Rem.~\ref{rem:qsp_qsvt_pictures}.}
    \label{fig:composing_qsp_gadgets}
\end{figure}

We begin our description of gadget compositions with a definition.

\begin{definition}[Interlink for gadgets] \label{def:gadget_interlink}
    Let $\mathfrak{G}$ and $\mathfrak{G}'$ be $(a, b)$ and $(c, d)$ gadgets respectively. An \emph{interlink} between these gadgets specifies a way to validly connect them. Take $[b], [c]$ to be the length-$b$ and length-$c$ (zero-indexed) ordered lists of labels of the outputs of $\mathfrak{G}$ and the inputs of $\mathfrak{G}'$ respectively. An \emph{interlink} is a three element list $\mathfrak{I} \equiv (B, C, W)$ with the following prescription: $B$ is a sublist of $[b]$ of size $e \in \{0, 1, \cdots, \min{(b, c)}\}$, $C$ is a sublist of $[c]$ also of size $e$, and $W$ is a member of $S_{e}$ the permutation group over $e$ elements.
\end{definition}

\noindent We can use this definition to precisely state the following theorem on the general combinations of gadgets. This theorem shows that gadgets can be composed in a nearly unconstrained way, snapped together like LEGOs, with well-understood associated cost and functional action.

\begin{restatable}[Composing a gadget with an atomic gadget]{theorem}{completecompgadgets} 
\label{thm:full_gadget_composition}
    Let $\varepsilon, \delta > 0$, $\mathfrak{G}$ be an $(a, b)$ gadget and $(\Xi, S)$ an antisymmetric atomic $(c, d)$ gadget, where $\Xi \equiv \{\Phi_0, \dots, \Phi_{d - 1}\}$ and $S \equiv \{s_0, \dots, s_{d - 1})$ for $(\Xi, S)$. Suppose the gadget $\mathfrak{G}$ achieves 
    \begin{equation}
    \label{eq:F}
    F(x) \equiv \{f_0(x_0, \cdots, x_{a - 1}), f_1(x_0, \cdots, x_{a - 1}), \dots, f_{b - 1}(x_0, \cdots, x_{a - 1})\} \in (\mathbb{R}^{a} \rightarrow \mathbb{R}^{b})
    \end{equation}
    over $x \in [-1, 1]^{a}$, and the atomic gadget $(\Xi, S)$ achieves
    \begin{equation}
    \label{eq:G}
        G(y) \equiv \{g_0(y_0, \cdots, y_{c - 1}), g_1(y_0, \cdots, y_{c - 1}), \dots, g_{d - 1}(y_0, \cdots, y_{c - 1})\} \in (\mathbb{R}^c \rightarrow \mathbb{R}^d)
    \end{equation}
    over $y \in [-1, 1]^{c}$. Let $\mathfrak{I} = (B, C, W)$ be an interlink between these gadgets. Then, there exists a gadget $\mathfrak{G}^\prime$ which $\varepsilon$-approximately achieves
        \begin{equation}
        \label{eq:H}
            H(x, y') \equiv \bigcup_{k \in [d]} g_k
                \left(
                    \bigcup_{j \in B} f_{W(j)}(x_0, \dots, x_a) 
                    \cup 
                    \bigcup_{k \not\in C} y_k
                \right)
            \cup 
            \bigcup_{k \not\in B} f_k(x_0, \dots, x_a)
        \end{equation} 
    over $(F(x), y') \in \mathcal{D}$, where $y'$ is the subset of $y_k$ such that $k \notin C$ and $\mathcal{D}$ is a domain determined by the correction procedure utilized. The set union symbol is abused to mean concatenation of lists with respect to the pre-established order of the labels of the relevant input and output \emph{lists} (not sets), and where $W(j)$ is the result of the application of the specified permutation applied to the \emph{index} of the subset member $j$ in the ordered sublist $C$ of $[c]$. Moreover, $\mathfrak{G}'$ can be constructed efficiently, and uses only a description of $(\Xi, S)$ and a total of $\widetilde{\mathcal{O}}(d \lvert \Xi \rvert_{\infty}\,\zeta)$ black-box calls to the \emph{unitaries produced by running} $\mathfrak{G}$ to realize the function $H$.
    
    Note $\zeta$, the cost of correcting the gadget $\mathfrak{G}$, is precisely the cost to make $\mathfrak{G}$ \emph{snappable}, and this snappability enables the simple compositional form of Eq.~\ref{eq:H}.
    $\zeta$ is one among two choices presented in Thm~\ref{thm:qsp_correction}, with the variables $\varepsilon$ and $\delta$ carrying over. Generally, $\zeta = \widetilde{O}(\text{polylog}(\varepsilon^{-1}) \text{poly}(\delta^{-1}))$, with possible polynomial scaling in a parameter $\gamma$, given in Thm.~\ref{thm:qsp_correction}, as well. This choice also dictates whether $\mathcal{O}(1)$ or zero ancilla qubits are required to perform this composition, and whether the success probability of achieving each individual function is at least $(1 - \varepsilon)^2$, or $1$. $\lvert \Xi \rvert_{\infty}$ is the maximum length of lists within $\Xi$. For a proof of this result, see Appx.~\ref{appx:main_proofs}. 
\end{restatable}

\begin{remark}[Arbitrary gadget interlinks]
    While it is possible to consider interlinks between \emph{arbitrary} gadgets, rather than assuming \emph{a priori} that the second gadget is atomic, it isn't possible to make a general claim about the required complexity to implement such a composition. For a generic gadget, treated as a black box accepting and outputting unitaries, one cannot know the required query complexity of the input legs in order to achieve a given output, with some precision, without knowledge of exactly how the input legs are queried within the gadget: behaviour which can, in principle, vary dramatically between different gadgets. We do know this internal structure for the special case of atomic gadgets (they are collections of M-QSP circuits). Nevertheless, characterization of interlinks between gadgets and atomic gadgets allows for full characterizations of arbitrary ``gadget networks" of many interlinks along with associated corrections (as an atomic gadget which has been corrected may not longer strictly be an atomic gadget, due to the possible introduction of ancilla qubits). As a result, Thm.~\ref{thm:full_gadget_composition} is sufficiently powerful for all practical purposes within the scope of this paper.
\end{remark}

Via Thm.~\ref{thm:full_gadget_composition}, an interlink $\mathfrak{I}$ can be thought of as an operator over gadgets, mapping a pair of gadgets to a single gadget with a possibly new type
    \begin{equation}
        (a, b), (c, d) \rightarrow (a + c - e, b + d - e),
    \end{equation}
for any $e \in \{0, 1, \cdots, \min{(b, c)}\}$. 

\begin{remark}
\label{rem:interlink}
Given an interlink, $\mathfrak{I}$, a gadget $\mathfrak{G}$, and an antisymmetric atomic gadget $(\Xi, S)$, we refer to the gadget $\mathfrak{G}'$ resulting from the above operation as $\mathfrak{I}[\mathfrak{G}, (\Xi, S)]$. Moreover, we can think of an interlink as inducing an operation of pairs of functions themselves. Going forward, given an interlink $\mathfrak{I}$, as well as functions $F$ and $G$ of the form of Eq.~\eqref{eq:F} and Eq.~\eqref{eq:G}, we let $G \circ_{\mathfrak{I}} F$ denote the function $H$ of Eq.~\eqref{eq:H}. 
\end{remark}

Graphically an interlink is simple to describe, as is shown in Fig.~\ref{fig:gadget_interlink}, which also depicts the general action described in Thm.~\ref{thm:full_gadget_composition}.
    \begin{figure}
        \centering
        \includegraphics[width=0.8\textwidth]{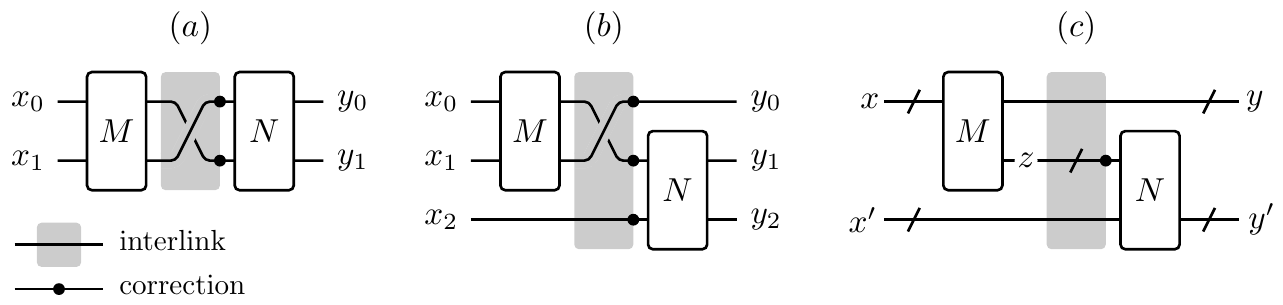}
        \caption{A depiction of gadgets coupled by interlinks (Def.~\ref{def:gadget_interlink}) both in example (a-b), and generally (c). Given two gadgets, an interlink $(B, C, W)$ specifies a subset of output legs $B$, input legs $C$, and a permutation $W$, joining gadgets according to Thm.~\ref{thm:full_gadget_composition}. Two possible interlinks are shown explicitly in (a-b) for $(2,2)$-gadgets. In (c), possibly many legs are suppressed in dashed legs, and the sets $B, C$ in Thm.~\ref{thm:full_gadget_composition} are precisely those composing the composite leg $z$; in this way (c) captures the most general coupling of two gadgets, per Thm.~\ref{thm:full_gadget_composition}. In all cases black dots indicate a correction protocol (\ref{thm:qsp_correction}) necessary to properly couple legs.}
        \label{fig:gadget_interlink}
    \end{figure}
Beyond composition, one can define a variety of operators over gadgets; we note that, as these gadgets are themselves superoperators, we are free to duplicate and elide outputs at the possible cost of additional queries. As discussed in Appx.~\ref{appx:main_proofs}, one can \emph{augment} an $(a, b)$ gadget to an $(a, b + c)$ gadget, as well as \emph{elide} an $(a, b)$ gadget to an $(a, b - c)$ gadget by ignoring outputs. Moreover, the input of an $(a, b)$ gadget can be \emph{pinned} to yield an $(a - c, b)$ gadget, or the output legs of an $(a, b)$ gadget can be \emph{permuted} to yield an $(a, b)$ gadget. Finally, note that the correction protocol itself can be thought of as an operator over gadgets, leaving the function achieved by the gadget approximately unchanged, but transforming its unitary output by a $\sigma_z$ conjugation. The form and cost of these basic operations is covered in Appx.~\ref{appx:main_proofs} and, once generally computed, these rules allow the algorithmist to reason about complex, multi-input/output functional operations agnostic to their realizing circuits, in a \emph{functional style}.

Before continuing, we note that defining the closure over arbitrarily interlinking finite-size gadgets is not a simple task. For instance, decomposing a given polynomial in a single-variable into an optimal tower of composed, lower degree polynomials (a strictly simpler problem) was once assumed to be a computationally hard problem \cite{bz_decomp_85, kl_poly_decomp_89, dickerson_thesis_89}, and the conditions under which such decomposition is efficient in the multivariable case have only recently been generally understood \cite{faugere_poly_decomp_09} and analyzed in the approximate setting \cite{cgjw_approx_decomp_99, gm_approx_decomp_07, dso_exact_approx_decomp_13}, with the generic problem known to be NP-hard \cite{dickerson_np_93}. Nevertheless such decompositions, if they do exist, are basically unique (a consequence of Ritt's theorem \cite{ritt_prime_poly_22}), and packages exist within most computer algebra systems for computing them. Consequently, while we show that we can achieve the full algebraic and monoidal operations natural to multivariable polynomials, our constructions only furnish competitive upper bounds on the required space, gate, and query complexity to achieve a given transform, with lower bounds depending on deep results in algorithmic complexity theory \cite{solomonoff_60, kolmogorov_65, chaitin_69} and polynomial decomposition theorems not exposited here. Despite this fact, we are still able to provide a compact \emph{equivalence theorem} for composite gadgets achieving polynomials which can be decomposed into a tower of lower-degree polynomials for which atomic gadgets are known.

\begin{restatable}[Efficient equivalence of polynomial compositions and gadget assemblages]{theorem}{existence} \label{thm:polynomial_gadget_equivalence}
    Let the expression $\mathcal{L}(x_1, \dots, x_n)$ denote the set of all polynomials achievable by atomic gadgets in the variables $x_1, \dots, x_n$. Suppose $P(x_1, \dots, x_n)$ is a polynomial of degree $D$ and can be split into a tower of $m = \mathcal{O}(\log(D))$ \emph{interlinked polynomials} (Rem.~\ref{rem:interlink}), 
    \begin{equation}
        P = P^{(m - 1)} \circ_{\mathfrak{I}_{m - 2}} \left( P^{(m - 2)} \circ_{\mathfrak{I}_{m - 3}} \circ \left( \cdots \left( P^{(1)} \circ_{\mathfrak{I}_{0}} P^{(0)} \right) \cdots \right) \right)
    \end{equation}
    such that $P_i^{(j)} \in \mathcal{L}(x_1, \dots, x_n)$ for all $i$ and $j$. Assume that the composition satisfies the domain condition (Rem.~\ref{rem:domain_nesting}) with $P_k$ supported on $\mathcal{D}_k$ such that each $\mathcal{D}_k$ is separated from singular points $\pm 1$ (and possibly $0$) by some length $\delta \in \mathcal{O}(1)$ for all $k$. Then, there exists an assemblage of $m$ atomic, snappable gadgets which $\varepsilon$-approximates the polynomial $P$ on the domain yielded from the $\mathcal{D}_k$, with cost $\widetilde{O}( \text{poly}(D)\,\text{polylog}(\varepsilon))$. 
\end{restatable}

For a proof of this theorem, see Appx.~\ref{appx:main_proofs}.
We remark that the discussion of query and gate complexity in Thm.~\ref{thm:full_gadget_composition} and thus Thm.~\ref{thm:polynomial_gadget_equivalence}, while simple, does not immediately extend to composite gadgets with complex internal structure. For this purpose we introduce a variety of independently interesting mathematical objects which allow for expedient calculation of the cost of implementing a gadget which achieves a desired function up to a desired precision over a desired range of inputs. The basic object for such costs is a \emph{cost matrix} (Def.~\ref{def:cost_matrix} in Appx.~\ref{sec:gadget_cost}), which permits one to determine the query complexity of a gadget with respect to a given input leg for a given output leg, up to a specified precision over a specified range of inputs. With a few simple abstractions, computing composite gadget costs is thus reduced to a lightly augmented form of matrix multiplication over sub-gadget cost matrices. Details of this calculation are presented in  Appx.~\ref{sec:gadget_cost}.

\section{Examples} \label{sec:examples}

\noindent The results of Sec.~\ref{sec:gadget_composition} together project a new path for building quantum subroutines: namely, the equivalence of Thm.~\ref{thm:polynomial_gadget_equivalence} ensures that we can think purely in terms of the polynomials achieved by a series of gadgets when snapping them together visually as modular pieces, while this LEGO-like construction is governed behind the scenes in the circuit picture by the simple rules of Thm.~\ref{thm:full_gadget_composition}. 
To highlight these principles together in action, we thus consider a series of linked examples making use of minimally complex pieces. We provide an overview of these examples in Fig.~\ref{fig:fig_permitted_functions}, grouping them in terms of the natural algebraic structures they constitute; in turn, these examples are exposited in Sec.~\ref{sec:arithmetic}, which covers common basic arithmetic, and Sec.~\ref{sec:familiar_functions}, which combines these pieces to achieve a variety of familiar and utilitarian functional classes.

\begin{figure}
    \centering
    \includegraphics[width=0.85\textwidth]{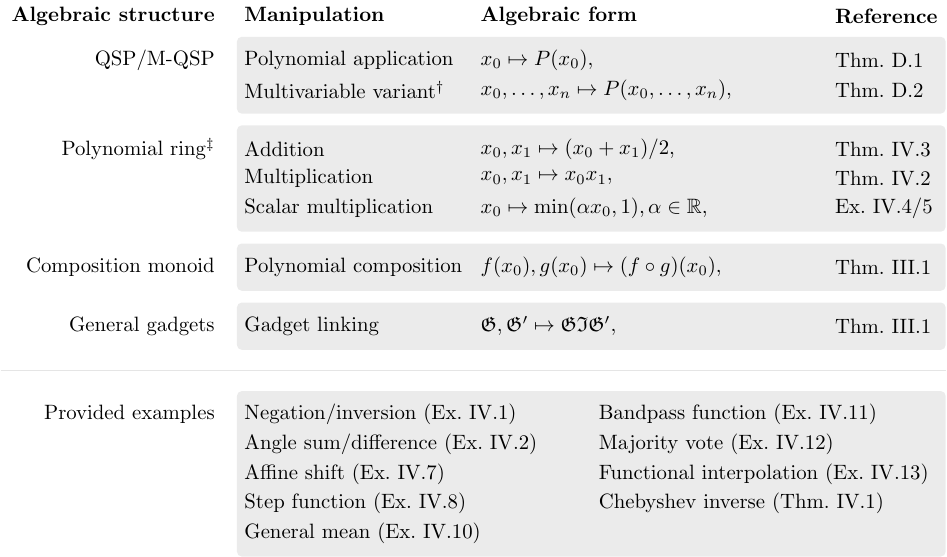}
    \caption{A summary of functional transforms achieved in this work, grouped (above the line) by their constitution of common mathematical structures over polynomials. Addition, polynomial multiplication, and scalar multiplication form a polynomial ring, where ($\ddag$) indicates we consider (possibly clipped to have a maximum norm of one) polynomials over $x \in [-1, 1]$ of definite parity taking real-values; these properties are preserved under the specified operations. We also achieve a monoid generated by single-variable polynomial composition, as well as its generalized form capturing tuples of multivariable polynomials, described by the grammar and language of gadgets (Appx.~\ref{appx:functional_programming}). We stress that in the general case we cannot achieve these operations exactly, only to arbitrary precision and with an associated cost. These operations, together with the ability to generate arbitrary single-variable bounded, definite parity, real polynomials using QSP, and a special subset ($\dagger$) of such polynomials with M-QSP \cite{rossi_m_qsp_22}, permit all highlighted examples of Sec.~\ref{sec:examples} (enumerated below the line).}
    \label{fig:fig_permitted_functions}
\end{figure}


\subsection{Basic arithmetic with gadgets}
\label{sec:arithmetic}

\noindent One of the most immediately useful applications of the protocols derived in the previous section is the ability to perform basic arithmetic and interpolation, coherently, between two polynomial functions, each achieved by an independently specified QSP protocol.

\begin{example}[Inversion and negation] \label{ex:simple_gadgets}
	We start with two simple $(1,1)$ gadgets with utility in the creation of composite functions. Each take as input some $x_0 \in [-1,1]$ as usual. The first, termed \emph{inversion} pre-applies an $iX$ (special unitary) to the oracle, producing $\sqrt{1 - x_0^2}$ which `logically inverts' the input on $\{0, 1\}$. The second, termed \emph{negation} conjugates the oracle by the (special unitary) $iY$, and achieves $-x_0$. Note that these are not \emph{atomic gadgets} per Def.~\ref{def:qsp_gadget}, as their form is outside that of M-QSP protocols, but that they produce embeddable output (Def.~\ref{def:embeddable}).
\end{example}

\begin{example}[Angle sum and difference] \label{ex:angle_gadget}
    To illustrate a useful but non-obviously antisymmetric gadget, we can consider the $(2, 1)$ gadget resulting from multiplying two oracles in different variables with QSP phases all zero, e.g., $U_0 U_1$, which achieves $x_0 x_1 - \sqrt{1 - x_0^2}\sqrt{1 - x_1^2}$. This corresponds to adding the angles associated with oracle unitaries. This can be changed to an angle difference by conjugating one oracle, e.g., $U_1 \mapsto e^{i\pi/4\sigma_z} U_1 e^{-i\pi/4\sigma_z}$.
\end{example}

\begin{example}[Multiplication] \label{ex:product_gadget}
    Let $\mathfrak{G}_{\text{mult}} \equiv (\Xi, S)$ be the $(2, 1)$ atomic gadget with 
    \begin{equation}
        \Xi \equiv \{\Phi\} = \{\{-\pi/4, \pi/4, -\pi/4, \pi/4\}\} \ \text{ and } \ S = \{\{0, 1, 0\}\}.   
    \end{equation}
    This gadget achieves $f(x_0, x_1) = T_{2}(x_0) x_1$, which can be linked to gadget outputs to multiply functions (up to pre-application of $T_2(x) = 2x^2 - 1$, the second Chebyshev polynomial of the first kind, to the first argument). 
\end{example}

\begin{example}[Sub-normalization] \label{ex:subnorm_gadget}
Let $(\Xi, S)$ be the $(2, 1)$ atomic gadget of Ex.~\ref{ex:product_gadget}. Let $\mathfrak{G}_{\text{subnorm}}(a)$ be the $(1, 1)$ gadget obtained from pinning (Def.~\ref{def:aux_gadget_operations}) the index-$0$ input leg to the unitary $U_0 = e^{i \arccos(x_0) \sigma_x}$ where $x_0 = \sqrt{(1 + a)/2}$, for $a \in [-1, 1]$. Then $\mathfrak{G}_{\text{subnorm}}(a)$ obtains the function $f(x) = T_2(\sqrt{(1 + a)/2}) x = ax$, for all $x \in [-1, 1]$. 
\end{example}

\begin{example}[Scaling] \label{ex:scaling_gadget}
	Take the $(1, 1)$ atypical gadget constructed in Ex.~\ref{ex:atypical_gadget}, which achieves an arbitrary bounded real polynomial $P$ of definite parity. Consider the polynomial approximation of $P$ from Lemma 30 of \cite{gslw_qsvt_19}, which $\varepsilon$-approximates the linear function $a x$ for $a > 1$ on the interval $x \in [(-1 + \delta)/a, (1 - \delta)/a]$. Then the QSP phases $\Phi$ of the protocol achieving $P$ can be inserted into the prescription of Ex.~\ref{ex:atypical_gadget} to yield a $(1, 1)$ gadget which multiplies its input by a scalar greater than $1$ with a clipped output: $\mathfrak{G}_{\text{scale}}(a)$. Moreover, the cost of this gadget is $\mathcal{O}([a/\delta]\log{[a/\varepsilon]})$ queries to $x_0$, by \cite{gslw_qsvt_19}, and requires no additional space.
\end{example}

\begin{example}[Addition] \label{ex:sum_gadget}
    Let $\mathfrak{G}_{\text{add}} \equiv (\Xi, S)$ be the $(2, 1)$ atomic gadget with
    \begin{equation}
    \Xi \equiv \{\Phi\} = \{ 0, \pi/4, 0, -\pi/4, 0 \} \ \text{ and } \ S = \{0, 1, 1, 0\}
    \end{equation}
    This protocol will achieve the function $f(x_0, x_1) = T_2(x_0) T_2(x_1)$ for $x_0, x_1 \in [-1, 1]$.
    Let $U_0$ and $U_1$ be embeddable, with $U_0 = e^{i \arccos(x_0) \sigma_x}$ and $U_1 = e^{i \arccos(x_1) \sigma_x}$. The products
        \begin{align}
         V_0 &= U_0 U_1 = e^{i (\arccos(x_0) + \arccos(x_1)) \sigma_x},
         \\
         V_1 &= U_0 e^{i \pi \sigma_z/2} U_1 e^{-i \sigma_x \pi/2} = U_0 U_1^{\dagger} = e^{i (\arccos(x_0) - \arccos(x_1)) \sigma_x},
        \end{align}
    define simple $(2, 1)$ gadgets. We can achieve a $(2, 1)$ gadget by first duplicating the input legs corresponding to $U_0$ and $U_1$, passing a pair $(U_0, U_1)$ into each of the $V_0$ and $V_1$ gadgets, then passing these two output legs into $(\Xi, S)$. However, note that $V_0$ and $V_1$, while being $\sigma_x$-rotations, can generally have their rotation angle in $[-\pi, \pi]$, rather than $[0, \pi]$. Thus, either $V_0$ is embeddable \emph{or} $\sigma_z V_0 \sigma_z$ is embeddable, with the same holding true for $V_1$. However, as can be easily checked, this $\sigma_z$-ambiguity will not matter: these extra $\sigma_z$-rotations (when absorbed into the phase sequence $\Phi$ as extra $\pi/2$-rotations) will have no effect on the output polynomial: this gadget will always achieve
        \begin{align}
           f(x_0, x_1) &= T_2 \left( \cos(\arccos(x_0) + \arccos(x_1)) \right) T_2 \left( \cos(\arccos(x_0) - \arccos(x_1)) \right)
           \\[0.7ex] &= \cos\left(2 (\arccos(x_0) + \arccos(x_1)) \right) \cos\left(2 (\arccos(x_0) - \arccos(x_1)) \right)
           \\ &= \frac{\cos(4 \arccos(x_0)) + \cos(4 \arccos(x_1))}{2} = \frac{T_4(x_0) + T_4(x_1)}{2}.
        \end{align}
    See Fig.~\ref{fig:sum_gadget} for a graphical depiction of the interlinks constituting this gadget.
\end{example}

\begin{example}[Constant shift] \label{ex:shift_gadget}
    It is possible to pin the addition gadget to perform a ``constant shift" (up to the fourth Chebyshev polynomial): $x \mapsto \frac{1}{2} T_4(x) + b$ for $b \in [-1/2, 1/2]$. This gadget, denoted $\mathfrak{G}_{\text{shift}}(b)$, is achieved by taking the gadget of Ex.~\ref{ex:sum_gadget} and pinning its second input (Def.~\ref{def:aux_gadget_operations}) to $e^{i \arccos(x_1) \sigma_x}$, with $T_4(x_1) = 2b$. Note it is possible to combine the sub-normalization and constant shift protocols into an \emph{affine shift} of the form $x \mapsto \frac{1}{2} T_4(ax) + b$.
\end{example}

The above examples, implementing basic arithmetic, have the constraint that they apply a Chebyshev polynomial to at least one of the inputs. As a final point, note that all of the provided examples can be rid of the inclusion of Chebyshev polynomials (approximately), via composition with the inverse Chebyshev protocol of Thm.~\ref{thm:inv_cheb}, stated and proved below.

\begin{restatable}[Inverses of Chebyshev polynomials]{theorem}{inversecheb} \label{thm:inv_cheb}
Given $0 < \delta \leq 1$ and $0 < \varepsilon \leq 1/2$, there exists $(1, 1)$ gadgets of depth and cost $\mathcal{O}(\delta^{-1} \log(\varepsilon^{-2} \delta^{-1/4}))$ which $\varepsilon$-approximately achieve functions $T_{2^n}^{-1}(x)$ (right-inverses) for $x \in [-1 + \delta, 1 - \delta]$, using a single qubit (no ancillae).
\end{restatable}

\noindent
For a proof of this result, see Appx.~\ref{appx:gadget_compositions}. Using this protocol (in particular, by interlinking the inverse Chebyshev gadget with the arithmetic gadgets), it is possible to realize the following transformations of the functions achieved above, approximately, on restricted domains:
\begin{align}
    \label{eq:primitives}
    f(x_0, x_1) = T_2(x_0) x_1 &\Longrightarrow f'(x_0, x_1) = x_0 x_1, \\[0.5ex]
    \label{eq:primitives2}
    f(x_0, x_1) = \frac{T_4(x_0) + T_4(x_1)}{2} &\Longrightarrow f'(x_0, x_1) = \frac{x_0 + x_1}{2}.
\end{align}
This, of course, will come at the expense of added circuit depth. To this end, we have the following theorems.

\begin{restatable}[Arbitrary approximate multiplication]{theorem}{arbitrarymultiplication}
\label{thm:provable_mult}
    Given $0 < \delta \leq 1$ an $0 < \varepsilon \leq 1/2$, there exists a $(2, 1)$ gadget of depth/cost $\mathcal{O}(\delta^{-1} \log(\varepsilon^{-2} \delta^{-1/4}))$ which $\varepsilon$-approximately achieves the function $f(x_0, x_1) = x_0 x_1$ for all $x_0 \in [-1 + \delta, 1 - \delta]$ and $x_1 \in [-1, 1]$ using no ancilla qubits.
\end{restatable}

\begin{restatable}[Arbitrary approximate addition]{theorem}{arbitraryaddition}
\label{thm:provable_addition}
    Given $0 < \delta \leq 1$ and $0 < \varepsilon \leq 1/2$, there exists a $(2, 1)$ gadget of depth/cost $\mathcal{O}(\delta^{-1} \log(\varepsilon^{-2} \delta^{-1/4}))$ which $\varepsilon$-approximately achieves the function $f(x_0, x_1) = (x_0 + x_1)/2$ for all $x_0, x_1 \in [-1 + \delta, 1 - \delta]$ using no ancilla qubits.
\end{restatable}

We refer to Appx.~\ref{appx:gadget_compositions} for proofs of these results, and more details about the cost of performing these transformations. With these basic operations instantiated, it is now easy to see how Thm.~\ref{thm:polynomial_gadget_equivalence} immediately implies that all multivariable polynomials up to bound and parity constraints are approximately achieveable by \emph{some} composite gadget mirroring the structure of the algebraic manipulations achieved. While the question of the \emph{most efficient} gadget realization is more difficult, we can freely use the gadgets above sparingly in the next section to instantiate a series of familiar and useful functions.

\begin{figure}
    \centering
    \includegraphics[width=0.75\textwidth]{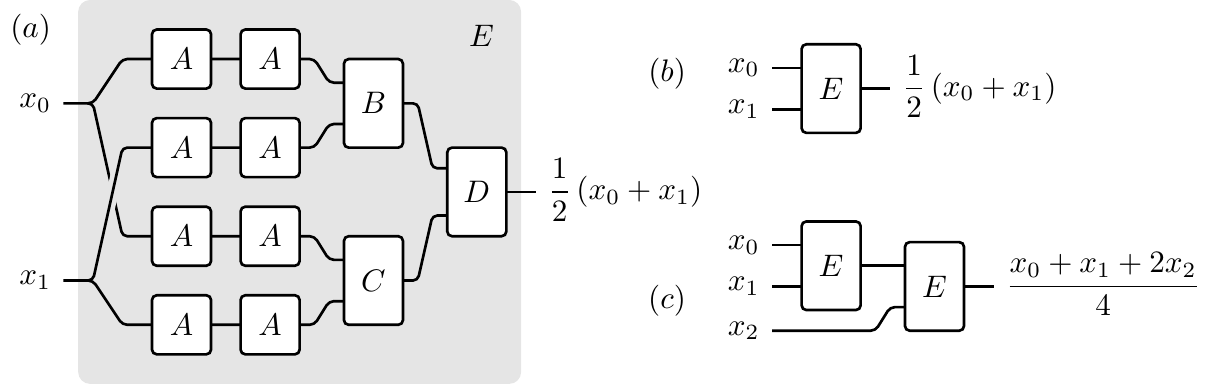}
    \caption{Component breakdown of the sum gadget (Ex.~\ref{ex:sum_gadget}). The sum gadget (a) is achieved by the composition of multiple $(1, 1)$ and $(2, 1)$ gadgets as shown, specifically $A$ (the square root gadget from Thm.~\ref{thm:inv_cheb} with $T_2$), $B$ (angle sum gadget, Ex.~\ref{ex:angle_gadget}), $C$ (angle subtraction gadget, Ex.~\ref{ex:angle_gadget}), and $D$ (product gadget, Ex.~\ref{ex:product_gadget}). The achieved composite operation $(b)$ can be treated as a $(2, 1)$ gadget labeled by $E$, whose partial composition with itself $(c)$ has the expected semantic behavior.}
    \label{fig:sum_gadget}
\end{figure}


\subsection{Familiar multivariable functions from gadgets} \label{sec:familiar_functions}

\noindent Having highlighted basic examples of arithmetic operations with respect to two variables, we leverage these operations to construct more complicated functional transformations. We begin with a short remark.

\begin{remark}[Algebra of gadget operations] \label{rem:multivar_block_encodings}
Given the operations of addition and multiplication presented in Eq.~\eqref{eq:primitives}, Eq.~\eqref{eq:primitives2}, as well as the scaling operations of Ex.~\ref{ex:subnorm_gadget} and Ex.~\ref{ex:scaling_gadget}, and the general ability to compose gadgets arbitrarily, it is possible to achieve approximations of arbitrary multivariate polynomials on mildly restricted domains. This is a somewhat surprising result in and of itself: sufficiently high-depth M-QSP protocols (up to the possible use of ancilla qubits during correction) can approximately achieve all multivariate polynomials, modulo certain constraints imposed in the examples of Sec.~\ref{sec:arithmetic}. In general, these naïve, ``ground-up" constructions of polynomials will have poor scaling, due to the need to perform the highly non-linear inverse Chebyshev and corrective protocols. Nevertheless, there are also other examples of short protocols which achieve useful transformations \textit{without} the added depth of performing square roots or many nested correction protocols. This fact is one of the fundamental takeaways of the following results we wish to emphasize: while it is \textit{possible} to build all transformations from the repeated composition of only a few low-degree gadgets, in order to achieve experimentally reasonable protocols, it is much more advantageous to use such basic gadgets sparingly, relying on functions efficiently achieved by standard M-QSP protocols where possible.
\end{remark}

To begin our discussion of more advanced examples, we present a brief discussion of techniques for constructively achieving step and bandpass functions through gadget composition.

\begin{restatable}[Basic step function of \cite{mf_recursive_23}]{example}{basic_stepfn} 
\label{ex:step_fun_gadget}
    The protocol introduced in~\cite{mf_recursive_23}, refered to as \emph{r-QSP} is an example of an $n$-fold composition of $(1, 1)$ atomic gadgets. In particular, the authors (implicitly) construct a family of $(1, 1)$ atomic gadgets $\mathfrak{G}^{(\ell)}_{\text{step}} \equiv (\Xi^{(\ell)}, S^{(\ell)}) = (\Phi_{\text{step}}^{(\ell)}, \{0, \dots, 0\})$ with $|\Phi_{\text{step}}^{(\ell)}| = 2\ell + 1$ such that the $n_{\ell}$-fold full composition (Thm.~\ref{thm:full_gadget_composition}) $\mathfrak{G}^{(\ell)}_{\text{step}} \circ \cdots \circ \mathfrak{G}^{(\ell)}_{\text{step}}$ achieves an $\varepsilon$-approximation of the function $f(x) = \text{sign}(x)$ on the interval $[-1, -\delta] \cup [\delta, 1]$ when the length of the resulting gadget, $\zeta = (2\ell + 1)^{n_{\ell}}$, satisfies
    \begin{equation}
        \zeta = \mathcal{O}\left( \delta^{-2(\nu_{\ell} + 1)} \log^{1 + \nu_{\ell}} \left( \varepsilon^{-1} \right) \right), \ \ \ \nu_{\ell} = \frac{\log(2\ell + 1)}{\log(\ell + 1)} - 1.
    \end{equation}
\end{restatable}

Making use of the basic step function construction of~\cite{mf_recursive_23}, as well as the arithmetic operations of the previous section, it is possible to construct a simple, infinite family of bandpass filters.

\begin{restatable}[Simple bandpass family]{example}{bandpass}
Given $a \in [\delta, 1/\sqrt{2}]$, there exists a family of $(1, 1)$ gadgets $\mathfrak{G}_{\text{bandpass}}^{(\ell)}(a)$, such that $\mathfrak{G}^{(\ell)}_{\text{bandpass}}(a)$ $\varepsilon$-approximately achieves a function $f$ satisfying $|f| \leq 1$ as well as
\begin{equation}
    f(x) \in \begin{cases} 
    [1 - \varepsilon, 1] & \text{if} \ x \in [-a + \delta, a - \delta] \\
    [-1, -1 + \varepsilon] & \text{if} \ x \in \left[-\frac{1}{\sqrt{2}}, -a - \delta\right] \cup \left[a + \delta, \frac{1}{\sqrt{2}}\right],
    \end{cases}
\end{equation}
where the total depth of the gadget, $\zeta'$, satisfies
\begin{equation}
    \label{eq:bandpass_scaling}
    \zeta' = \mathcal{O}\left( (2a\delta)^{-4(\nu_{\ell} + 1)} \log^{1 + \nu_{\ell}} \left( \varepsilon^{-1} \right) \right).
\end{equation}
Moreover, given the phases of $\Phi^{(\ell)}_{\text{step}}$, a sequence of Ex.~\ref{ex:step_fun_gadget}, the gadget $\mathfrak{G}^{(\ell)}_{\text{bandpass}}(a)$ can be described analytically, in terms of these phases.
\end{restatable}

For a proof of this result, see Appx.~\ref{appx:gadget_compositions}. This example provide good illustration of the fundamental tension between the generality of the classes of functions that can be achieved via composition of gadgets, and the asymptotic scaling of resource requirements. Here, we have shown that it is possible to achieve a \emph{restricted} class of bandpass filters via the gadget formalism. It is, generally speaking, possible to drop the restrictions of this example (in particular, the constraint that the function approximately achieves only the values $\pm 1$), at the expense of more queries/corrective protocol applications (see Ex.~\ref{ex:bandpass_gadget}).

There are many other examples of low-depth protocols which can be achieved via gadget composition, and the simple arithmetic operations discussed previously. We highlight some of these examples below.

\begin{restatable}[$2^n$ mean]{example}{meangadget} \label{ex:mean_gadget}
	Let $\mathfrak{G}$ be the $(2, 1)$ gadget achieving $(x_0 + x_1)/2$. Then one can build a composite $(2^n, 1)$ gadget $\mathfrak{G}^\prime$, as well as the correction protocol, achieving $2^{-n}(x_0 + x_1 + \dots + x_{2^n - 1})$ using $(2^n - 1)$ instances of $\mathfrak{G}$ and successively pooling pairs of variables. E.g., $(x_0 + x_1)/2$ and $(x_2 + x_3)/2$ are taken to their average $(x_0 + x_1 + x_2 + x_3)/4$.
\end{restatable} 

\noindent The cost of implementing gadgets will depend on the range of input parameters for which it must output an $\varepsilon$-approximation of the desired function. Moreover, the cost will depend on whether the ancilla or ancilla-free gadget correction protocol is used, which will vary on a case-by-case basis. As an illustrative example, we provide a detailed discussion of the complexity of implementing the $2^n$-gadget, and discuss how its cost differs from the cost LCU, in Appx.~\ref{appx:gadget_compositions}. In addition to this example, we make note of some other gadgets, which are, in principle, possible to achieve. The complexity analysis of each, and how it compares to other techniques such as LCU, carries forward in the same manner of casework as is presented in the example of Appx.~\ref{appx:gadget_compositions}.

\begin{example}[Arbitrary one-dimensional bandpass functions] \label{ex:bandpass_gadget}
	Consider $\mathfrak{G}, \mathfrak{G}^\prime$ two $(1,1)$ gadgets approximately achieving sign functions $\Theta(x_0 - a_0), \Theta(x_0 - a_1)$ (Ex.~\ref{ex:step_fun_gadget}), where $\Theta(x_0)$ takes the value $-1$ for $x < 0$ and $1$ for $x > 0$ on the interval $[-1, 1]$. Without loss of generality assume $a_0 < a_1$; then the averaging gadget applied to $\mathfrak{G}$ and $\mathfrak{G}^{\prime\prime}$ (the modified gadget achieving $-\Theta(a_1)$, enabled by Ex.~\ref{ex:simple_gadgets}) achieves an approximation to the following ideal bandpass function: $\Theta^\prime(x_0) \equiv [\Theta(x_0 - a_0) - \Theta(x_0 - a_1)]/2$.which for $x_0 < a_0$ achieves $0$, for $a_0 < x < a_1$ achieves $1$, and for $x > a_1$ again achieves zero.
\end{example}

 \begin{example}[Majority vote] \label{ex:maj_vote}
	Majority vote among $2^n$ elements follows directly from the previous examples; let $\mathfrak{G}$ the $(2^n, 1)$ gadget of Ex.~\ref{ex:mean_gadget} and $\mathfrak{G}^\prime$ the $(1, 1)$ gadget thresholding at $x_0 = 1/2$ (Ex.~\ref{ex:bandpass_gadget}). Then the simple composition $\mathfrak{G}^\prime \circ \mathfrak{G}$ is a $(2^n, 1)$ gadget satisfying the desired properties.
\end{example}

\begin{example}[Subnormalized functional interpolation] \label{ex:subnorm_interpolation_gadget}
	There exists a $(3, 1)$ gadget that approximately achieves the function $f(x_0, x_1, x_2) = [x_0 x_2 + x_1 \sqrt{1 - x_2^2}]/2$. In other words, depending on $x_2$, the gadget smoothly interpolates between the sub-normalized variables $x_0$ and $x_1$. If $x_0, x_1$ result from previous gadgets, this realizes a (sub-normalized) interpolation between two functions. Construction follows from previous gadgets. Namely the products $x_0 x_2$ and $x_1 \sqrt{1 - x_2^2}$ are realized by Ex.~\ref{ex:product_gadget}, with $\sqrt{1 - x^2}$ achieved from $x_2$ by Ex.~\ref{ex:simple_gadgets}. Finally, both results are summed by Ex.~\ref{ex:sum_gadget}.
\end{example}

The cost analysis of these gadgets carries forward similarly to the example highlighted in Appx.~\ref{appx:gadget_compositions}, depending on which corrective protocol is chosen.

\section{Discussion and conclusion} \label{sec:discussion_conclusion}

\noindent At a practical level this work gives a modular construction, rooted in QSP and QSVT, for achieving highly expressive multivariable block encodings for commuting linear operators. This construction relies on special \emph{corrective} subroutines which restore the \emph{embeddability} (Def.~\ref{def:embeddable}) of QSP- and QSVT-like protocols in a black-box way. Recovering embeddability allows assemblages of such protocols, which can take many unitary oracles as input and produce many unitary outputs, to be simply wired together \emph{at the level of the functions they apply} to matrix-invariant quantities; that this is a true modular process is demonstrated in the statements of the work's main theorems: Thms.~\ref{thm:full_gadget_composition} and \ref{thm:polynomial_gadget_equivalence}, which indicate how to intersperse independently defined correction protocols (Thm.~\ref{thm:qsp_correction}) between modules. Such modules are named \emph{gadgets} (Def.~\ref{def:qsp_gadget_general}), and their valid combinations are shown to follow a simple grammar, and generate both the natural commutative ring and composition monoid over multivariable polynomials. In the terminology of programming language theory, the modularity of gadgets is formalized in identifying them as \emph{monadic types}, with the correction protocol furnishing a previously missing component of the monadic function \emph{bind} (Appx.~\ref{appx:functional_programming}).

We can phrase our results at three level, in order of increasing burden of proof and abstractness. \textbf{(a)} We provide improved, constructive upper bounds on the resource requirements (space, gate, and query complexity) of multivariable polynomial block encodings, and make comparisons to previously known methods (see Appx.~\ref{appx:performance_comparisons}), with inspiration from older results in modular classical filter design \cite{kh_sharpening_77, saramaki_cascade_87} and composite pulse sequences \cite{kv_passband_nesting_13, jones_nested_not_13, lyc_optimal_pulses_14}. \textbf{(b)} We provide methods for interpretable constructions of coherent (quantum) control flow, with no obvious incoherent counterpart. Specifically, we work in a coherent-access model, analogous to that of recent work on the benefit of measurement-free oracular access \cite{acq_qualm_22, huang_exp_22}, and show that the action of QSVT-like protocols can be coherently conditioned on or coupled with the action of other QSVT-like protocols, \emph{simultaneously within invariant subspaces} related to the singular values and vectors of (possibly many) block-encoded operators. \textbf{(c)} We apply methods in functional programming, type theory, and category theory to re-situate QSP and QSVT as quantum \emph{types}, where our results instantiate a \emph{monadic type} over real-valued, real-argument multivariable polynomials (see Appx.~\ref{appx:functional_programming}). Our methods thus enable high-level quantum algorithm design, generalizing the work of \cite{rc_semantic_alg_23} and providing a simpler path for bootstrapping M-QSP \cite{rossi_m_qsp_22} to achieve wider and less constrained multivariable functional transforms.

Our construction of gadgets couples to the design of quantum algorithms on two levels, inspired by a natural division in formal and natural languages: semantics and syntax \cite{knuth_attribute_grammars_68, selinger_qpl_04, selinger_higher_order_04}. On the first, gadgets specify which polynomial transforms can be treated functionally, serving as basic units of computation; formally this is summarized in the instantiation of a monad, and characterizes the basic \emph{semantics} of the resulting system. At the second level, we define a formal grammar enumerating valid ways to connect gadgets, which characterizes \emph{syntax} of our system. While the ultimate relation between our construction's syntax and semantics is more complicated than this simple division implies (depending on various hybrid structures such as attribute grammars \cite{knuth_attribute_grammars_68}) these two levels help to situate this work in classical programming language theory, where division into semantics and syntax is considered a first step along the path of parsing, optimizing over, and automatically generating programs for concrete tasks under concrete constraints.

To restate a crucial point, this work considers gadgets comprising QSP/QSVT protocols that can be linked, following intermediary processing, to other gadgets \emph{such that} their embedded functional transforms are combined in a \emph{semantically clear way}. This intermediary processing is generally necessary, non-trivial, and strictly subsumes the single-variable case \cite{rc_semantic_alg_23}, which itself subsumed various techniques in the recursive application of quantum algorithms \cite{lyc_optimal_pulses_14, jones_nested_not_13, kv_passband_nesting_13}. The key benefit of our method is that each gadget inherits the approximative efficiency and ancilla-efficient character of QSP, meaning that the complexity of achieving a desired polynomial transform is no longer coupled directly to term number, polynomial degree, or polynomial norm, as in LCU or standard QSVT \cite{motlagh2023generalized}, but instead sophisticated results in the theory of multivariable polynomial decomposition \cite{bz_decomp_85, kl_poly_decomp_89, dickerson_np_93, faugere_poly_decomp_09}, and multivariable quantum signal processing (discussed in detail in Appx.~\ref{appx:performance_comparisons}). In turn, this allows new questions in the algorithmic complexity (or Solomonoff–Kolmogorov–Chaitin complexity) \cite{solomonoff_60, kolmogorov_65, chaitin_69} of certain approximate functional transforms under the composition of gadgets, for which our results provide an upper bound. The benefit of this approach to quantum algorithm design is multiform: numerical (given our ability to re-apply pre-computed gadgets elsewhere), experimental (given the reduced number of distinct QSP phases to compute), conceptual (given the ability to reason at the level of abstracted, useful quantum subroutines), and finally pedagogical (given the treatment of disparate quantum algorithms through a single lens of block-encoding manipulations).  

We stress, however, that there remain significant open questions on the optimality of method offered here, especially in light of the dependence of a gadget's complexity on sophisticated results in the theory of polynomial decompositions \cite{bz_decomp_85, kl_poly_decomp_89, dickerson_np_93, faugere_poly_decomp_09}. Consequently existing methods for multivariable block encoding, such as LCU, continue to offer diverse, context-specific utility; in fact, as we showed, LCU offers one method for instantiating gadgets and thus provides an upper bound on the complexity of any given multivariable transformation---critically, for a wide class of functions, summarized in Fig.~\ref{fig:fig_permitted_functions}, we show that we can improve on these LCU-provided bounds arbitrarily, recovering the constant-space and improved infinity-norm dependence of QSP and QSVT. Investigating richer classes of target functions for which robust separations in space and query complexity can be shown between gadgets and LCU is a promising avenue for future work.

Whether viewed in terms of highly-efficient block encodings, or as the instantiation of a monadic type \cite{wadler_comprehending_90, wadler_essence_92, wadler_monads_95} over higher-order quantum processes composable in accordance to a simple grammar, this work instantiates new and expressive QSP/QSVT-based quantum algorithms, with function-first interpretability and a declarative style. To support this claim, we provide multiple examples and an associated code repository allowing for the automated construction and analysis of gadgets (located at \texttt{\href{https://github.com/ichuang/pyqsp/tree/beta}{https://github.com/ichuang/pyqsp/tree/beta}}). In that repository we provide block encodings of algorithmically useful multivariable polynomials with improved query and space complexity over alternative methods for non-trivial classes of target functions, contained in an extensive test suite. We hope that these concrete numerical tools lower the barrier toward further experimentation---e.g., we noted in this work that the ability to semantically combine QSP/QSVT as modules in a black-box way relied strongly on a restricted circuit form, leaving the question of the existence of larger families of (less) structured parameterized quantum circuits for which functional programming abstractions can still be built.

\section{Acknowledgements}

\noindent  Z.R. was supported in part by the NSF EPiQC program. J.C. thanks Prof. Nathan Wiebe for valuable conversations, and for catalyzing and funding an internship during the summer of 2023 at MIT, where the majority of the work on this project was conducted. J.C. also thanks the MIT Research Laboratory of Electronics (RLE) for their hospitality during his time at MIT. I.C. was supported in part by the U.S. DoE, Office of Science, National Quantum Information Science Research Centers, and Co-design Center for Quantum Advantage (C2QA) under contract number DE-SC0012704. 


\bibliography{main}

\begin{thebibliography}{77}
\providecommand{\natexlab}[1]{#1}
\providecommand{\url}[1]{\texttt{#1}}
\expandafter\ifx\csname urlstyle\endcsname\relax
  \providecommand{\doi}[1]{doi: #1}\else
  \providecommand{\doi}{doi: \begingroup \urlstyle{rm}\Url}\fi

\bibitem[Aaronson(2015)]{aaronson_15}
Scott Aaronson.
\newblock Read the fine print.
\newblock \emph{Nat. Phys.}, 11\penalty0 (4):\penalty0 291--293, 2015.
\newblock URL \url{https://doi.org/10.1038/nphys3272}.

\bibitem[Aaronson(2022)]{aaronson2022much}
Scott Aaronson.
\newblock How much structure is needed for huge quantum speedups?
\newblock \emph{arXiv preprint, arXiv:2209.06930}, 2022.
\newblock URL \url{https://doi.org/10.48550/arXiv.2209.06930}.

\bibitem[Aharonov et~al.(2022)Aharonov, Cotler, and Qi]{acq_qualm_22}
Dorit Aharonov, Jordan Cotler, and Xiao-Liang Qi.
\newblock Quantum algorithmic measurement.
\newblock \emph{Nat. Commun}, 13\penalty0 (1), 2022.
\newblock \doi{10.1038/s41467-021-27922-0}.
\newblock URL \url{https://doi.org/10.1038%2Fs41467-021-27922-0}.

\bibitem[Barenco et~al.(1995)Barenco, Bennett, Cleve, DiVincenzo, Margolus, Shor, Sleator, Smolin, and Weinfurter]{bbcdmsssw_elem_gates_95}
Adriano Barenco, Charles~H. Bennett, Richard Cleve, David~P. DiVincenzo, Norman Margolus, Peter Shor, Tycho Sleator, John~A. Smolin, and Harald Weinfurter.
\newblock Elementary gates for quantum computation.
\newblock \emph{Phys. Rev. A}, 52:\penalty0 3457--3467, Nov 1995.
\newblock \doi{10.1103/PhysRevA.52.3457}.
\newblock URL \url{https://doi.org/10.1103/PhysRevA.52.3457}.

\bibitem[Barton and Zippel(1985)]{bz_decomp_85}
David~R. Barton and Richard Zippel.
\newblock Polynomial decomposition algorithms.
\newblock \emph{J. Symb. Comput.}, 1\penalty0 (2):\penalty0 159--168, 1985.
\newblock URL \url{https://doi.org/10.1016/S0747-7171(85)80012-2}.

\bibitem[Berry et~al.(2015)Berry, Childs, and Kothari]{bck_ham_sim_15}
Dominic~W. Berry, Andrew~M. Childs, and Robin Kothari.
\newblock Hamiltonian simulation with nearly optimal dependence on all parameters.
\newblock In \emph{2015 IEEE 56th Annual Symposium on Foundations of Computer Science}, pages 792--809, 2015.
\newblock \doi{10.1109/FOCS.2015.54}.
\newblock URL \url{https://doi.org/10.1109/FOCS.2015.54}.

\bibitem[Borns-Weil et~al.(2023)Borns-Weil, Saffat, and Stier]{bss_commuting_matrices_23}
Yonah Borns-Weil, Tahsin Saffat, and Zachary Stier.
\newblock A quantum algorithm for functions of multiple commuting {H}ermitian matrices.
\newblock \emph{arXiv preprint, arXiv:2302.11139}, 2023.
\newblock URL \url{http://dx.doi.org/10.48550/arXiv.2302.11139}.

\bibitem[Chaitin(1969)]{chaitin_69}
Gregory~J. Chaitin.
\newblock On the simplicity and speed of programs for computing infinite sets of natural numbers.
\newblock \emph{J. ACM}, 16\penalty0 (3):\penalty0 407–422, 1969.
\newblock ISSN 0004-5411.
\newblock \doi{10.1145/321526.321530}.
\newblock URL \url{https://doi.org/10.1145/321526.321530}.

\bibitem[Chao et~al.(2020)Chao, Ding, Gilyen, Huang, and Szegedy]{chao_machine_prec_20}
Rui Chao, Dawei Ding, Andras Gilyen, Cupjin Huang, and Mario Szegedy.
\newblock Finding angles for quantum signal processing with machine precision.
\newblock \emph{arXiv preprint, arXiv:2003.02831}, 2020.
\newblock URL \url{https://doi.org/10.48550/arXiv.2003.02831}.

\bibitem[Chia et~al.(2020)Chia, Gily\'{e}n, Li, Lin, Tang, and Wang]{chia_20}
Nai-Hui Chia, Andr\'{a}s Gily\'{e}n, Tongyang Li, Han-Hsuan Lin, Ewin Tang, and Chunhao Wang.
\newblock Sampling-based sublinear low-rank matrix arithmetic framework for dequantizing quantum machine learning.
\newblock In \emph{Proceedings of the 52nd Annual ACM SIGACT Symposium on Theory of Computing}, STOC 2020, page 387–400, NY, USA, 2020. ACM.
\newblock \doi{10.1145/3357713.3384314}.
\newblock URL \url{https://doi.org/10.1145/3357713.3384314}.

\bibitem[Childs and van Dam(2010)]{childs_10}
Andrew~M. Childs and Wim van Dam.
\newblock Quantum algorithms for algebraic problems.
\newblock \emph{Rev. Mod. Phys.}, 82:\penalty0 1--52, 2010.
\newblock \doi{10.1103/RevModPhys.82.1}.
\newblock URL \url{https://doi.org/10.1103/RevModPhys.82.1}.

\bibitem[Chiribella et~al.(2008)Chiribella, D'Ariano, and Perinotti]{chiribella_networks_08}
G.~Chiribella, G.~M. D'Ariano, and P.~Perinotti.
\newblock Quantum circuit architecture.
\newblock \emph{Phys. Rev. Lett.}, 101:\penalty0 060401, 2008.
\newblock \doi{10.1103/PhysRevLett.101.060401}.
\newblock URL \url{https://doi.org/10.1103/PhysRevLett.101.060401}.

\bibitem[Chiribella et~al.(2009)Chiribella, D'Ariano, and Perinotti]{chiribella_networks_09}
Giulio Chiribella, Giacomo~Mauro D'Ariano, and Paolo Perinotti.
\newblock Theoretical framework for quantum networks.
\newblock \emph{Phys. Rev. A}, 80:\penalty0 022339, 2009.
\newblock \doi{10.1103/PhysRevA.80.022339}.
\newblock URL \url{https://doi.org/10.1103/PhysRevA.80.022339}.

\bibitem[Chomsky(1956)]{chomsky_hierarchy_56}
N.~Chomsky.
\newblock Three models for the description of language.
\newblock \emph{IRE Trans. Inf. Theory}, 2\penalty0 (3):\penalty0 113--124, 1956.
\newblock \doi{10.1109/TIT.1956.1056813}.
\newblock URL \url{https://doi.org/10.1109/TIT.1956.1056813}.

\bibitem[Claudon et~al.(2024)Claudon, Zylberman, Feniou, Debbasch, Peruzzo, and Piquemal]{czfdpp_polylog_cnot_24}
B.~Claudon, J.~Zylberman, C.~Feniou, F.~Debbasch, A.~Peruzzo, and J.~Piquemal.
\newblock Polylogarithmic-depth controlled-{NOT} gates without ancilla qubits.
\newblock \emph{Nat. Comm.}, 15, 2024.
\newblock URL \url{https://doi.org/10.1038/s41467-024-50065-x}.

\bibitem[Corless et~al.(1999)Corless, Giesbrecht, Jeffrey, and Watt]{cgjw_approx_decomp_99}
Robert~M. Corless, Mark~W. Giesbrecht, David~J. Jeffrey, and Stephen~M. Watt.
\newblock Approximate polynomial decomposition.
\newblock In \emph{Proceedings of the 1999 International Symposium on Symbolic and Algebraic Computation}, ISSAC '99, pages 213--219, New York, NY, USA, 1999. Association for Computing Machinery.
\newblock \doi{10.1145/309831.309939}.
\newblock URL \url{https://doi.org/10.1145/309831.309939}.

\bibitem[Demirtas et~al.(2013)Demirtas, Su, and Oppenheim]{dso_exact_approx_decomp_13}
Sefa Demirtas, Guolong Su, and Alan~V. Oppenheim.
\newblock Exact and approximate polynomial decomposition methods for signal processing applications.
\newblock In \emph{2013 IEEE International Conference on Acoustics, Speech and Signal Processing}, pages 5373--5377, 2013.
\newblock \doi{10.1109/ICASSP.2013.6638689}.
\newblock URL \url{https://doi.org/10.1109/ICASSP.2013.6638689}.

\bibitem[Dickerson(1989)]{dickerson_thesis_89}
Matthew~T Dickerson.
\newblock \emph{The functional decomposition of polynomials}.
\newblock PhD thesis, Cornell University, 1989.
\newblock URL \url{https://dl.acm.org/doi/10.5555/866387}.

\bibitem[Dickerson(1993)]{dickerson_np_93}
Matthew~T Dickerson.
\newblock General polynomial decomposition and the s-1-decomposition are {NP}-hard.
\newblock \emph{Int. J. Found. Comput. Sci.}, 4\penalty0 (02):\penalty0 147--156, 1993.
\newblock \doi{10.1142/S0129054193000109}.
\newblock URL \url{https://doi.org/10.1142/S0129054193000109}.

\bibitem[Dong et~al.(2021)Dong, Meng, Whaley, and Lin]{dmwl_efficient_qsp_phases_21}
Yulong Dong, Xiang Meng, K.~Birgitta Whaley, and Lin Lin.
\newblock Efficient phase-factor evaluation in quantum signal processing.
\newblock \emph{Phys. Rev. A}, 103\penalty0 (4), 2021.
\newblock ISSN 2469-9934.
\newblock \doi{10.1103/physreva.103.042419}.
\newblock URL \url{https://doi.org/10.1103/PhysRevA.103.042419}.

\bibitem[Dong et~al.(2024)Dong, Lin, Ni, and Wang]{dlnw_infinite_qsp_22}
Yulong Dong, Lin Lin, Hongkang Ni, and Jiasu Wang.
\newblock Infinite quantum signal processing.
\newblock \emph{{Quantum}}, 8:\penalty0 1558, December 2024.
\newblock ISSN 2521-327X.
\newblock \doi{10.22331/q-2024-12-10-1558}.
\newblock URL \url{https://doi.org/10.22331/q-2024-12-10-1558}.

\bibitem[Faug{\`e}re and Perret(2009)]{faugere_poly_decomp_09}
Jean-Charles Faug{\`e}re and Ludovic Perret.
\newblock An efficient algorithm for decomposing multivariate polynomials and its applications to cryptography.
\newblock \emph{J. Symb. Comput.}, 44\penalty0 (12):\penalty0 1676--1689, 2009.
\newblock URL \url{http://dx.doi.org/10.1016/j.jsc.2008.02.005}.

\bibitem[Gay(2006)]{gay_qpls_06}
Simon~J. Gay.
\newblock Quantum programming languages: survey and bibliography.
\newblock \emph{Math. Struct.}, 16\penalty0 (4):\penalty0 581–600, 2006.
\newblock ISSN 0960-1295.
\newblock \doi{10.1017/S0960129506005378}.
\newblock URL \url{https://doi.org/10.1017/S0960129506005378}.

\bibitem[Geronimo and Woerdeman(2004)]{mv_frt_04}
J.~Geronimo and Hugo Woerdeman.
\newblock Positive extensions, {F}ejér-{R}iesz factorization and autoregressive filters in two variables.
\newblock \emph{Ann. Math.}, 160:\penalty0 839--906, Nov 2004.
\newblock \doi{10.4007/annals.2004.160.839}.
\newblock URL \url{https://doi.org/10.4007/annals.2004.160.839}.

\bibitem[Giesbrecht and May(2007)]{gm_approx_decomp_07}
Mark Giesbrecht and John May.
\newblock New algorithms for exact and approximate polynomial decomposition.
\newblock In \emph{Symbolic-Numeric Computation}, pages 99--112. Springer, 2007.
\newblock URL \url{https://doi.org/10.1007/978-3-7643-7984-1_7}.

\bibitem[Gilyén et~al.(2019)Gilyén, Su, Low, and Wiebe]{gslw_qsvt_19}
András Gilyén, Yuan Su, Guang~Hao Low, and Nathan Wiebe.
\newblock Quantum singular value transformation and beyond: exponential improvements for quantum matrix arithmetics.
\newblock \emph{Proceedings of the 51st Annual ACM SIGACT Symposium on Theory of Computing}, 2019.
\newblock \doi{10.1145/3313276.3316366}.
\newblock URL \url{http://dx.doi.org/10.1145/3313276.3316366}.

\bibitem[Grattage(2005)]{ag_functional_qpl_05}
Jonathan Grattage.
\newblock A functional quantum programming language.
\newblock In \emph{Proceedings of the 20th Annual IEEE Symposium on Logic in Computer Science}, LICS '05, page 249–258, USA, 2005. IEEE Computer Society.
\newblock ISBN 0769522661.
\newblock \doi{10.1109/LICS.2005.1}.
\newblock URL \url{https://doi.org/10.1109/LICS.2005.1}.

\bibitem[Grover(2005)]{grover_05}
Lov~K Grover.
\newblock Fixed-point quantum search.
\newblock \emph{Phys. Rev. Lett.}, 95\penalty0 (15):\penalty0 150501, 2005.
\newblock \doi{10.1103/PhysRevLett.95.150501}.
\newblock URL \url{https://doi.org/10.1103/PhysRevLett.95.150501}.

\bibitem[Haah(2019)]{haah_product_decomp_19}
Jeongwan Haah.
\newblock Product decomposition of periodic functions in quantum signal processing.
\newblock \emph{Quantum}, 3:\penalty0 190, Oct 2019.
\newblock ISSN 2521-327X.
\newblock \doi{10.22331/q-2019-10-07-190}.
\newblock URL \url{http://dx.doi.org/10.22331/q-2019-10-07-190}.

\bibitem[Huang et~al.(2022)Huang, Broughton, Cotler, Chen, Li, Mohseni, Neven, Babbush, Kueng, Preskill, et~al.]{huang_exp_22}
Hsin-Yuan Huang, Michael Broughton, Jordan Cotler, Sitan Chen, Jerry Li, Masoud Mohseni, Hartmut Neven, Ryan Babbush, Richard Kueng, John Preskill, et~al.
\newblock Quantum advantage in learning from experiments.
\newblock \emph{Science}, 376\penalty0 (6598):\penalty0 1182--1186, 2022.
\newblock URL \url{https://doi.org/10.1126/science.abn7293}.

\bibitem[Husain et~al.(2013)Husain, Kawamura, and Jones]{hkj_comp_pulse_analysis_13}
Sami Husain, Minaru Kawamura, and Jonathan~A Jones.
\newblock Further analysis of some symmetric and antisymmetric composite pulses for tackling pulse strength errors.
\newblock \emph{J. Magn. Reson.}, 230:\penalty0 145--154, 2013.
\newblock URL \url{https://doi.org/10.1016/j.jmr.2013.02.007}.

\bibitem[Jones(2013)]{jones_nested_not_13}
Jonathan~A Jones.
\newblock Nested composite {NOT} gates for quantum computation.
\newblock \emph{Phys. Lett. A}, 377\penalty0 (40):\penalty0 2860--2862, 2013.
\newblock URL \url{https://doi.org/10.1016/j.physleta.2013.08.040}.

\bibitem[Jordan(1875)]{jordan_75}
Camille Jordan.
\newblock Essai sur la g{\'e}om{\'e}trie {\`a} $ n $ dimensions.
\newblock \emph{Bulletin de la Soci{\'e}t{\'e} math{\'e}matique de France}, 3:\penalty0 103--174, 1875.

\bibitem[Kaiser and Hamming(1977)]{kh_sharpening_77}
J~Kaiser and R~Hamming.
\newblock Sharpening the response of a symmetric nonrecursive filter by multiple use of the same filter.
\newblock \emph{IEEE Transactions on Acoustics, Speech, and Signal Processing}, 25\penalty0 (5):\penalty0 415--422, 1977.
\newblock URL \url{https://doi.org/10.1109/TASSP.1977.1162980}.

\bibitem[Knuth(1968)]{knuth_attribute_grammars_68}
Donald~E Knuth.
\newblock Semantics of context-free languages.
\newblock \emph{Math. Syst. Theory}, 2\penalty0 (2):\penalty0 127--145, 1968.
\newblock URL \url{https://doi.org/10.1007/BF01702865}.

\bibitem[Kolmogorov(1965)]{kolmogorov_65}
A.~N. Kolmogorov.
\newblock Three approaches to the definition of “the quantity of information”.
\newblock \emph{Problemy peredachi informatsii}, 1\penalty0 (1):\penalty0 3--11, 1965.
\newblock URL \url{https://doi.org/10.1080/00207166808803030}.

\bibitem[Kozen and Landau(1989)]{kl_poly_decomp_89}
Dexter Kozen and Susan Landau.
\newblock Polynomial decomposition algorithms.
\newblock \emph{J. Symb. Comput.}, 7\penalty0 (5):\penalty0 445--456, 1989.
\newblock URL \url{https://doi.org/10.1016/S0747-7171(89)80027-6}.

\bibitem[Kretschmann and Werner(2005)]{kretschmann_networks_05}
Dennis Kretschmann and Reinhard~F Werner.
\newblock Quantum channels with memory.
\newblock \emph{PRA}, 72\penalty0 (6):\penalty0 062323, 2005.
\newblock URL \url{https://doi.org/10.1103/PhysRevA.72.062323}.

\bibitem[Kyoseva and Vitanov(2013)]{kv_passband_nesting_13}
Elica Kyoseva and Nikolay~V Vitanov.
\newblock Arbitrarily accurate passband composite pulses for dynamical suppression of amplitude noise.
\newblock \emph{Phys. Rev. A}, 88\penalty0 (6):\penalty0 063410, 2013.
\newblock URL \url{https://doi.org/10.1103/PhysRevA.88.063410}.

\bibitem[Lam et~al.(2006)Lam, Sethi, Ullman, and Aho]{lsua_compilers_06}
Monica Lam, Ravi Sethi, Jeffrey~D Ullman, and Alfred Aho.
\newblock Compilers: principles, techniques, and tools.
\newblock \emph{Pearson Education}, 2006.

\bibitem[Lambek and Scott(1988)]{lambek_higher_order_88}
Joachim Lambek and Philip~J Scott.
\newblock \emph{Introduction to higher-order categorical logic}, volume~7.
\newblock Cambridge University Press, 1988.

\bibitem[Lane(1978)]{cat_theory_78}
Saunders~Mac Lane.
\newblock \emph{Categories for the Working Mathematician}.
\newblock Springer New York, NY, 2nd edition, 1978.
\newblock URL \url{https://doi.org/10.1007/978-1-4757-4721-8}.

\bibitem[Low and Chuang(2017)]{lc_ham_sim_17}
G.~H. Low and I.~L. Chuang.
\newblock Optimal {H}amiltonian simulation by quantum signal processing.
\newblock \emph{Phys. Rev. Lett.}, 118:\penalty0 010501, 2017.
\newblock \doi{10.1103/PhysRevLett.118.010501}.
\newblock URL \url{https://doi.org/10.1103/PhysRevLett.118.010501}.

\bibitem[Low and Chuang(2019)]{lc_qubitization_19}
G.~H. Low and I.~L. Chuang.
\newblock Hamiltonian simulation by qubitization.
\newblock \emph{Quantum}, 3:\penalty0 163, 2019.
\newblock \doi{10.22331/q-2019-07-12-163}.
\newblock URL \url{http://dx.doi.org/10.22331/q-2019-07-12-163}.

\bibitem[Low et~al.(2016)Low, Yoder, and Chuang]{lyc_equiangular_16}
G.~H. Low, T.~J. Yoder, and I.~L. Chuang.
\newblock Methodology of resonant equiangular composite quantum gates.
\newblock \emph{Phys. Rev. X}, 6:\penalty0 041067, 2016.
\newblock \doi{10.1103/PhysRevX.6.041067}.
\newblock URL \url{https://doi.org/10.1103/PhysRevX.6.041067}.

\bibitem[Low et~al.(2014)Low, Yoder, and Chuang]{lyc_optimal_pulses_14}
Guang~Hao Low, Theodore~J Yoder, and Isaac~L Chuang.
\newblock Optimal arbitrarily accurate composite pulse sequences.
\newblock \emph{Phys. Rev. A}, 89\penalty0 (2):\penalty0 022341, 2014.
\newblock URL \url{https://doi.org/10.1103/PhysRevA.89.022341}.

\bibitem[Martyn et~al.(2021)Martyn, Rossi, Tan, and Chuang]{mrtc_unification_21}
John~M Martyn, Zane~M Rossi, Andrew~K Tan, and Isaac~L Chuang.
\newblock Grand unification of quantum algorithms.
\newblock \emph{PRX Quantum}, 2\penalty0 (4):\penalty0 040203, 2021.
\newblock URL \url{https://doi.org/10.1103/PRXQuantum.2.040203}.

\bibitem[Mizuta and Fujii(2024)]{mf_recursive_23}
Kaoru Mizuta and Keisuke Fujii.
\newblock Recursive quantum eigenvalue and singular-value transformation: Analytic construction of matrix sign function by {N}ewton iteration.
\newblock \emph{Phys. Rev. Res.}, 6:\penalty0 L012007, 2024.
\newblock \doi{10.1103/PhysRevResearch.6.L012007}.
\newblock URL \url{https://doi.org/10.1103/PhysRevResearch.6.L012007}.

\bibitem[Moggi(1988)]{moggi_lambda_calc_88}
Eugenio Moggi.
\newblock \emph{Computational lambda-calculus and monads}.
\newblock University of Edinburgh, 1988.

\bibitem[Moggi(1989)]{moggi_abstract_89}
Eugenio Moggi.
\newblock \emph{An abstract view of programming languages}.
\newblock University of Edinburgh, 1989.

\bibitem[Montanaro(2016)]{montanaro_16}
Ashley Montanaro.
\newblock Quantum algorithms: an overview.
\newblock \emph{Npj Quantum Inf.}, 2\penalty0 (1):\penalty0 1--8, 2016.
\newblock URL \url{https://doi.org/10.1038/npjqi.2015.23}.

\bibitem[Motlagh and Wiebe(2024)]{motlagh2023generalized}
Danial Motlagh and Nathan Wiebe.
\newblock Generalized quantum signal processing.
\newblock \emph{PRX Quantum}, 5:\penalty0 020368, 2024.
\newblock \doi{10.1103/PRXQuantum.5.020368}.
\newblock URL \url{https://doi.org/10.1103/PRXQuantum.5.020368}.

\bibitem[Nielsen and Chuang(2011)]{nc_textbook_11}
Michael~A. Nielsen and Isaac~L. Chuang.
\newblock \emph{Quantum {C}omputation and {Q}uantum {I}nformation}.
\newblock Cambridge University Press, USA, 10th edition, 2011.
\newblock ISBN 1107002176.
\newblock URL \url{https://doi.org/10.1017/CBO9780511976667}.

\bibitem[Odake et~al.(2024)Odake, Kristj\'ansson, Soeda, and Murao]{oksm_higher_order_23}
Tatsuki Odake, Hl\'er Kristj\'ansson, Akihito Soeda, and Mio Murao.
\newblock Higher-order quantum transformations of {H}amiltonian dynamics.
\newblock \emph{Phys. Rev. Res.}, 6:\penalty0 L012063, 2024.
\newblock \doi{10.1103/PhysRevResearch.6.L012063}.
\newblock URL \url{https://doi.org/10.1103/PhysRevResearch.6.L012063}.

\bibitem[Pagani et~al.(2014)Pagani, Selinger, and Valiron]{pagani_higher_order_semantics_14}
Michele Pagani, Peter Selinger, and Beno{\^\i}t Valiron.
\newblock Applying quantitative semantics to higher-order quantum computing.
\newblock In \emph{Proceedings of the 41st ACM SIGPLAN-SIGACT Symposium on Principles of Programming Languages}, pages 647--658, 2014.
\newblock URL \url{https://doi.org/10.1145/2578855.2535879}.

\bibitem[Polya and Szegö(1998)]{polya_szego_analysis_98}
George Polya and Gabor Szegö.
\newblock \emph{Problems and Theorems in Analysis II}.
\newblock Springer, Berlin, 1998.
\newblock URL \url{https://doi.org/10.1007/978-3-642-61905-2}.

\bibitem[Regev(2006)]{regev_06}
Oded Regev.
\newblock Fast amplification of {QMA} (lecture notes), 2006.
\newblock \url{https://cims.nyu.edu/~regev/teaching/quantum_fall_2005/ln/qma.pdf}.

\bibitem[Ritt(1922)]{ritt_prime_poly_22}
Joseph~Fels Ritt.
\newblock Prime and composite polynomials.
\newblock \emph{Trans. Amer. Math. Soc.}, 23\penalty0 (1):\penalty0 51--66, 1922.
\newblock URL \url{https://doi.org/10.1090/S0002-9947-1922-1501189-9}.

\bibitem[Rossi and Chuang(2022)]{rossi_m_qsp_22}
Zane~M. Rossi and Isaac~L. Chuang.
\newblock Multivariable quantum signal processing ({M-QSP}): prophecies of the two-headed oracle.
\newblock \emph{Quantum}, 6:\penalty0 811, Sep 2022.
\newblock \doi{10.22331/q-2022-09-20-811}.
\newblock URL \url{https://doi.org/10.22331\%2Fq-2022-09-20-811}.

\bibitem[Rossi and Chuang(2023)]{rc_semantic_alg_23}
Zane~M. Rossi and Isaac~L. Chuang.
\newblock {Semantic embedding for quantum algorithms}.
\newblock \emph{Journal of Mathematical Physics}, 64\penalty0 (12):\penalty0 122202, 12 2023.
\newblock ISSN 0022-2488.
\newblock \doi{10.1063/5.0160910}.
\newblock URL \url{https://doi.org/10.1063/5.0160910}.

\bibitem[Saramaki(1987)]{saramaki_cascade_87}
Tapio Saramaki.
\newblock Design of {FIR} filters as a tapped cascaded interconnection of identical subfilters.
\newblock \emph{IEEE Trans. Circuits Syst.}, 34\penalty0 (9):\penalty0 1011--1029, 1987.
\newblock URL \url{https://doi.org/10.1109/TCS.1987.1086263}.

\bibitem[Selinger(2004{\natexlab{a}})]{selinger_higher_order_04}
Peter Selinger.
\newblock Towards a semantics for higher-order quantum computation.
\newblock In \emph{Proceedings of the 2nd International Workshop on Quantum Programming Languages, TUCS General Publication}, volume~33, pages 127--143, 2004{\natexlab{a}}.

\bibitem[Selinger(2004{\natexlab{b}})]{selinger_qpl_04}
Peter Selinger.
\newblock Towards a quantum programming language.
\newblock \emph{Math. Struct.}, 14\penalty0 (4):\penalty0 527--586, 2004{\natexlab{b}}.
\newblock URL \url{https://doi.org/10.1017/s0960129504004256}.

\bibitem[Sheridan et~al.(2009)Sheridan, Maslov, and Mosca]{smm_09}
L~Sheridan, D~Maslov, and M~Mosca.
\newblock Approximating fractional time quantum evolution.
\newblock \emph{J. Phys. A Math. Theor.}, 42\penalty0 (18):\penalty0 185302, 2009.
\newblock \doi{10.1088/1751-8113/42/18/185302}.
\newblock URL \url{https://dx.doi.org/10.1088/1751-8113/42/18/185302}.

\bibitem[Solomonoff et~al.(1960)Solomonoff, of~Scientific~Research, and Company]{solomonoff_60}
R.J. Solomonoff, United States. Air Force.~Office of~Scientific~Research, and Zator Company.
\newblock \emph{A Preliminary Report on a General Theory of Inductive Inference}.
\newblock AFOSR TN-60-1459. United States Air Force, Office of Scientific Research, 1960.

\bibitem[Tan et~al.(2023)Tan, Liu, Tran, and Chuang]{tltc_alec_23}
Andrew~K Tan, Yuan Liu, Minh~C Tran, and Isaac~L Chuang.
\newblock Error correction of quantum algorithms: Arbitrarily accurate recovery of noisy quantum signal processing.
\newblock \emph{arXiv preprint, arXiv:2301.08542}, 2023.
\newblock URL \url{https://doi.org/10.48550/arXiv.2301.08542}.

\bibitem[Tang and Tian(2023)]{cs_qsvt_tang_tian_23}
Ewin Tang and Kevin Tian.
\newblock A {CS} guide to the quantum singular value transformation.
\newblock \emph{arXiv preprint, arXiv:2302.14324}, 2023.
\newblock \doi{10.48550/arxiv.2302.14324}.
\newblock URL \url{https://doi.org/10.48550/arXiv.2302.14324}.

\bibitem[Toyama et~al.(2009)Toyama, Kasai, van Dijk, and Nogami]{tkvn_matched_phase_grover_09}
F.~M. Toyama, S.~Kasai, W.~van Dijk, and Y.~Nogami.
\newblock Matched-multiphase {G}rover algorithm for a small number of marked states.
\newblock \emph{Phys. Rev. A}, 79:\penalty0 014301, Jan 2009.
\newblock \doi{10.1103/PhysRevA.79.014301}.
\newblock URL \url{https://doi.org/10.1103/PhysRevA.79.014301}.

\bibitem[van Wijngaarden(1965)]{w_grammar_65}
Adriaan van Wijngaarden.
\newblock Orthogonal design and description of a formal language.
\newblock \emph{Stichting Mathematisch Centrum. Rekenafdeling}, 1965.

\bibitem[Van~Wijngaarden et~al.(2012)Van~Wijngaarden, Mailloux, Peck, Koster, Lindsey, Sintzoff, Meertens, and Fisker]{w_grammar_12}
Adriaan Van~Wijngaarden, Barry~J Mailloux, John~EL Peck, Cornelis~HA Koster, Charles~Hodgson Lindsey, Michel Sintzoff, Lambert~GLT Meertens, and Richard~G Fisker.
\newblock \emph{Revised report on the algorithmic language ALGOL 68}.
\newblock Springer Science \& Business Media, 2012.

\bibitem[Wadler(1990)]{wadler_comprehending_90}
Philip Wadler.
\newblock Comprehending monads.
\newblock In \emph{Proceedings of the 1990 ACM Conference on LISP and Functional Programming}, pages 61--78, 1990.
\newblock URL \url{https://doi.org/10.1145/91556.91592}.

\bibitem[Wadler(1992)]{wadler_essence_92}
Philip Wadler.
\newblock The essence of functional programming.
\newblock In \emph{Proceedings of the 19th ACM SIGPLAN-SIGACT symposium on Principles of programming languages}, pages 1--14, 1992.
\newblock URL \url{https://doi.org/10.1145/143165.143169}.

\bibitem[Wadler(1995)]{wadler_monads_95}
Philip Wadler.
\newblock Monads for functional programming.
\newblock In \emph{Advanced Functional Programming: First International Spring School on Advanced Functional Programming Techniques B{\aa}stad, Sweden, May 24--30, 1995 Tutorial Text 1}, pages 24--52. Springer, 1995.
\newblock URL \url{https://doi.org/10.1007/978-3-662-02880-3_8}.

\bibitem[Wang et~al.(2022{\natexlab{a}})Wang, Dong, and Lin]{wdl_sym_qsp_22}
Jiasu Wang, Yulong Dong, and Lin Lin.
\newblock On the energy landscape of symmetric quantum signal processing.
\newblock \emph{Quantum}, 6:\penalty0 850, 2022{\natexlab{a}}.
\newblock \doi{10.22331/q-2022-11-03-850}.
\newblock URL \url{https://doi.org/10.22331/q-2022-11-03-850}.

\bibitem[Wang et~al.(2022{\natexlab{b}})Wang, Wang, Yu, and Zhang]{wang2022quantum}
Xin Wang, Youle Wang, Zhan Yu, and Lei Zhang.
\newblock Quantum phase processing: Transform and extract eigen-information of quantum systems.
\newblock \emph{arXiv preprint, arXiv:2209.14278}, 2022{\natexlab{b}}.
\newblock URL \url{https://doi.org/10.48550/arXiv.2209.14278}.

\bibitem[Yoder et~al.(2014)Yoder, Low, and Chuang]{ylc_fixed_point_14}
Theodore~J. Yoder, Guang~Hao Low, and Isaac~L. Chuang.
\newblock Fixed-point quantum search with an optimal number of queries.
\newblock \emph{Phys. Rev. Lett.}, 113\penalty0 (21):\penalty0 210501, 2014.
\newblock URL \url{https://doi.org/10.1103/PhysRevLett.113.210501}.

\bibitem[Yu et~al.(2022)Yu, Yao, Li, and Wang]{yu2022power}
Zhan Yu, Hongshun Yao, Mujin Li, and Xin Wang.
\newblock Power and limitations of single-qubit native quantum neural networks.
\newblock \emph{Advances in Neural Information Processing Systems}, 35:\penalty0 27810--27823, 2022.
\newblock URL \url{http://dx.doi.org/10.48550/arXiv.2205.07848}.

\end{thebibliography}


\newpage

\begin{appendix}

\centerline{\bf Supplemental Materials}

\vspace{2ex}

\noindent In what follows we provide a series of appendices. In turn, these detail proofs of theorems in the main body (Appx.~\ref{appx:main_proofs}), provide an overview of related prior techniques with performance comparisons (Appx.~\ref{appx:performance_comparisons}), provide detailed analysis of the non-obvious method for computing gadget cost (Appx.~\ref{sec:gadget_cost}), provide a review of the major results of QSP, M-QSP, and QSVT (Appx.~\ref{appx:variants_qsp}), provide analysis of \emph{correction protocols} (Appxs.~\ref{appx:extraction} and \ref{appx:roots}), cover some caveats in the theory of functional approximation by polynomials (Appx.~\ref{sec:poly_bnds}), compare certain concrete example gadgets for specific problems to their best known counterpart using other quantum algorithms (Appx.~\ref{appx:gadget_compositions}), discuss how gadgets instantiate monadic functions and relate to functional programming (gadget \emph{semantics}, Appxs.~\ref{appx:natural_transformations} and \ref{appx:qsp_qsvt_types}), and finally provide a formal language for gadgets constituted by an attribute grammar (gadget \emph{syntax}, Appx.~\ref{appx:formal_gadget_grammar}).

\section{Proofs of main results} \label{appx:main_proofs}

\noindent We begin by providing proofs of the main results of Sec.~\ref{sec:constructing_gadgets}. Note that $\widetilde{\mathcal{O}}$ denotes scaling to leading order. In addition, when using $\mathcal{O}$, we will often suppress scaling beyond leading and second-order (i.e. given leading-order polynomial scaling in parameter $p$, we keep $\log(p)$ terms and suppress terms of order $\log \log(p)$ and greater).

\correctingunitaries*

\begin{proof} 
    We will prove each of the complexities corresponding to the three different methods for performing the correction of an $\varepsilon$-twisted embeddable unitary. We will begin with the simplest case: the scenario where we are provided access to ancilla qubits as well as controlled oracle access to $U$.
    
    From Thm.~\ref{thm:extraction_superoperator}, we know that there exists a single-qubit protocol (effectuated by the superoperator $\mathcal{E}$) making $\zeta_1 = \mathcal{O}(\delta^{-1} \log(\varepsilon_1^{-2} \delta^{-1/2}))$ black-box calls to the unitary $U'$ which yields an $\varepsilon_1$-approximation of $e^{i \varphi \sigma_z}$, assuming $\cos(\theta) \in [-1 + \delta, 1 - \delta]$. We then have
    \begin{equation}
        || \mathcal{E}[U] - e^{i\varphi \sigma_z} || \leq || \mathcal{E}[U] - \mathcal{E}[U'] || + || \mathcal{E}[U'] - e^{i\varphi \sigma_z} || \leq ||\Phi[U] - \Phi[U']|| + \varepsilon_1,
    \end{equation}
    where $\Phi$ is a QSP protocol. Due to Lem.~\ref{lem:error}, we know that $||\Phi[U] - \Phi[U']|| \leq \zeta_1\nu$, implying that the overall error is $\zeta_1\nu + \varepsilon_1$. Due to Cor.~\ref{cor:ancilla_sqrt}, there exists a protocol $\mathcal{R}$ making a single black-box call to controlled-$\mathcal{E}[U]$ (which we have assumed we can access, as we have controlled-$U$ access) and using an extra qubit, such that $(\mathcal{R} \circ \mathcal{E})[U]$ will be an $(\zeta_1 \nu + \varepsilon_1)$-approximation of $e^{i \varphi/2} e^{i \varphi \sigma_z/2}$, with probability at least $1 - (2\zeta_1\nu + 2\varepsilon_1)$. It follows, in the case of success, that
    \begin{align}
        || (\mathcal{R} \circ \mathcal{E})[U]^{\dagger} U (\mathcal{R} \circ \mathcal{E})[U] - e^{i \theta \sigma_x} || & \leq ||(\mathcal{R} \circ \mathcal{E})[U]^{\dagger} e^{i \varphi \sigma_z/2} e^{i \theta \sigma_x} e^{-i \varphi \sigma_z/2} (\mathcal{R} \circ \mathcal{E})[U] - e^{i \theta \sigma_x} || + \nu
        \\ & \leq 2 || e^{i \theta \sigma_x} e^{-i \varphi \sigma_z/2} (\mathcal{R} \circ \mathcal{E})[U] - e^{i \theta \sigma_x} || + \nu
        \\ & \leq 2 || (\mathcal{R} \circ \mathcal{E})[U] - e^{i \varphi \sigma_z/2} || + \nu
        \\ & \leq (2\zeta_1 + 1)\nu + 2\varepsilon_1,
    \end{align}
    where we can drop the overall phases $e^{i\varphi/2}$ within the operator norm, via cancellation (see Rem.~\ref{rem:overall_phase}). Since the ancilla qubits can be re-used with each non-nested action of $\mathcal{R}$ (see Rem.~\ref{rem:ancilla_reuse}), the overall protocol will use two ancilla qubits, and makes $\zeta = 2\zeta_1 + 1$ black-box calls to $U$. Thus, setting $\varepsilon_1 = \varepsilon/2$, so the approximation error is $\zeta \nu + \varepsilon$, the success probability is greater than $(1 - (\zeta\nu + \varepsilon))^2$ (as $\mathcal{R}$ is called twice), and $\zeta = \mathcal{O}(\delta^{-1} \log(\varepsilon^{-2} \delta^{-1/2}))$ yields the desired result.

    In the case where we are provided ancillae, but do not have access \emph{a priori} to the signal oracles, we make use of the protocol of Thm.~\ref{thm:controlled_z_routine} to get, from $U$, an approximation of its controlled-variant. We then make use of the previous protocol $\mathcal{R}$ to achieve an approximation of the desired $\sigma_z$-square root. We call this composite protocol $\mathcal{R}'$. Because the first part of the protocol has the exact same asymptotic complexity as the protocol $\mathcal{R}$, and is effectively the implementation of two parallel QSP sequences interspersed by CSWAP gates, the error analysis will proceed in exactly the same way as above.
    
    Finally, from Thm.~\ref{thm:root}, not that there exists a single-qubit protocol $\mathcal{R}''$ making $\zeta_2 = \mathcal{O}(\gamma^{-2} \log(\varepsilon_2^{-2} \gamma^{-1/2}))$ black-box calls to $e^{i \varphi \sigma_z}$ such that $\mathcal{R}''[e^{i \varphi \sigma_z}]$ is an $\varepsilon_2$-approximation to $e^{i \varphi \sigma_z/2}$, assuming $\cos(\varphi) \in [\gamma, \sqrt{1 - \gamma^2}]$. Of course, since $\mathcal{R}''$ is effectively a QSP protocol, we can once again make use of Lem.~\ref{lem:error} to bound the accumulation of errors once again. We have
    \begin{align}
        || (\mathcal{R}'' \circ \mathcal{E})[U]^{\dagger} U (\mathcal{R}'' \circ \mathcal{E})[U] - e^{i \theta \sigma_x} || & \leq ||(\mathcal{R}'' \circ \mathcal{E})[U]^{\dagger} e^{i \varphi \sigma_z/2} e^{i \theta \sigma_x} e^{-i \varphi \sigma_z/2} (\mathcal{R}'' \circ \mathcal{E})[U] - e^{i \theta \sigma_x} || + \nu
        \\ & \leq 2 || e^{i \theta \sigma_x} e^{-i \varphi \sigma_z/2} (\mathcal{R}'' \circ \mathcal{E})[U] - e^{i \theta \sigma_x} || + \nu
        \\ & \leq 2 || (\mathcal{R}'' \circ \mathcal{E})[U] - e^{i \varphi \sigma_z/2} || + \nu
        \\ & \leq 2 || \mathcal{R}''[e^{i\varphi \sigma_z}] - e^{i \varphi \sigma_z/2} || + 2 || \mathcal{R}''[e^{\varphi i \sigma_z}] - \mathcal{R}''[\mathcal{E}[U]] || + \nu
        \\ & \leq 2\varepsilon_2 + 2\zeta_2\varepsilon_1 + (2\zeta_2\zeta_1 + 1)\nu
    \end{align}
    Clearly, a total of $\zeta' = 2\zeta_2 \zeta_1 + 1$ black-box calls are made to $U$ throughout the entire protocol. We set $\varepsilon_2 = \varepsilon/4$ and $\varepsilon_1 = \varepsilon/4\zeta_2 = \widetilde{\mathcal{O}}(\varepsilon \gamma^2)$, where we drop the log factor of $\log(\varepsilon^{-2} \gamma^{-1/2})^{-1}$ (as this will only contribute a $\log \log$ factor when plugged into the protocol depth upper-bound). Therefore, $\zeta_1 = \mathcal{O}(\delta^{-1} \log(\varepsilon^{-2} \gamma^{-4} \delta^{-1/2}))$, and
    \begin{equation}
        \zeta' = 2\zeta_2\zeta_1 + 1 = \mathcal{O}\left(\frac{1}{\gamma^2 \delta} \log^2 \left( \frac{1}{\gamma^4 \varepsilon^2 \delta^{1/2}} \right)\right).
    \end{equation}
    The approximation error is then $\zeta'\nu + \varepsilon$, and the proof is complete.
\end{proof}

\noindent
With this result, the main theorem characterizing gadget correction follows fairly immediately.

\correctinggadgets*

\begin{proof}
    This follows immediately from invoking Lem.~\ref{lemma:embeddable_correction} for each of the output gadgets, enforcing the condition that $\zeta \nu = \varepsilon \Rightarrow \nu = \varepsilon/\zeta = \widetilde{\mathcal{O}}(\varepsilon \delta)$ in the first and second cases, and $\zeta \nu = \varepsilon \Rightarrow \nu = \widetilde{\mathcal{O}}(\delta \gamma^{2} \varepsilon)$ in the final case, so that the total approximation error in both cases is $2\varepsilon = \mathcal{O}(\varepsilon)$.
\end{proof}

\begin{remark}[On the acceptable efficiency of correction protocols] \label{rem:acceptable_efficiency}
	Lemma~\ref{lemma:embeddable_correction} shows that, given an approximately embeddable unitary, the error accumulated during its correction is (1) proportional to the error of the input unitary and (2) the constant of proportionality is linear in the length of the correction procedure. We clarify in what follows that this scaling is not as poor as it appears, and connect this to natural but implicit assumptions for protocols given in this work:
		\begin{enumerate}[label=(\alph*)]
			\item Firstly, the error in the input approximately embeddable unitary ($\nu$ in Lemma~\ref{lemma:embeddable_correction}) arises \emph{only from previous correction protocols}, and thus can be suppressed by a factor of $\zeta$, for instance, by increasing the length of the previous correction procedure by a factor only \emph{logarithmic} in $\zeta$ (noting the $\varepsilon$-dependence in Lemma~\ref{lemma:embeddable_correction}, and that the argument is more subtle when corrections are changed). In other words, $\nu$ can be made small quite easily, inheriting the rapid convergence of QSP polynomials to desired smooth functions.
			\item Secondly, analogously to previous work on the self-composition of QSP and general classical signal processing \cite{kh_sharpening_77, saramaki_cascade_87, kv_passband_nesting_13,jones_nested_not_13, lyc_optimal_pulses_14}, this work considers circuits comprising logarithmically-many protocol-nesting operations, and thus only  logarithmically-many uses of correction protocols in the length of a given `base protocol.' While such scaling is more difficult to define where protocols of varying length are used, rather than where one protocol is recursively self-composed, this statement will be true even if scaling is taken with respect to the length of the longest `base protocol.' Ultimately, this restriction is equivalent to requiring that our composite protocols have length \emph{polynomial in the length of their smallest sub-protocols}.
			\item Taken together the (a) rapid convergence of polynomial approximation in QSP and the (b) natural restriction to logarithmic-depth composite \emph{gadgets} (Def.~\ref{def:qsp_gadget}; for now to be thought of as parallel circuits with the basic form of QSP/M-QSP), means the ultimate length of protocols scales, at worst, inverse polynomially in the desired output error.
		\end{enumerate}
	While the precise cost of a specific composite gadget can be involved to compute (Appx.~\ref{sec:gadget_cost}), and can change depending on properties of the achieved functions it comprises (Rems.~\ref{rem:add_mult_embed} and \ref{rem:incomp}), we can assert that QSP-based computations are \emph{acceptable} (in a colloquial sense) if their required length scales only inverse-polynomially in the desired error. While this is worse than the best case scenario of using QSP or M-QSP to approximate smooth functions, it is no worse than the scaling appearing when approximating common discontinuous functions. Note also this error is distinct from the error of functional approximation within a standard QSP or M-QSP protocol, but is of the same order, and can be independently manipulated in minimizing gadget depth.
\end{remark}

We are now able to discuss the cost associated with composing gadgets. First, a minor result, showing the simple relationship between unitary approximations and function approximations achieved with gadgets.

\begin{lemma}
    \label{lem:unitary_fn_approx}
    Let $\mathfrak{G}$ be an $(a, b)$ gadget which achieves unitaries $U_k'$ which $\varepsilon$-approximate unitaries $V_k$ for $k \in [b]$ over all $(x_0, \dots, x_{a - 1}) \in \mathcal{D}$. Then $\mathfrak{G}$ achieves functions $f_k$ which $\varepsilon$-approximate functions $g_k = \langle 0 | V_k |0\rangle$ for $k \in [b]$ over all $(x_0, \dots x_{a - 1}) \in \mathcal{D}$.
\end{lemma}
\begin{proof}
Note that $|f_k - g_k| = | \langle 0 | U_k' |0\rangle - \langle 0 | V_k |0\rangle | = | \langle 0| U_k' - V_k |0\rangle | \leq ||U_k' - V_k|| \leq \varepsilon$.
\end{proof}

\begin{figure}
    \centering
    \includegraphics[width=0.6\textwidth]{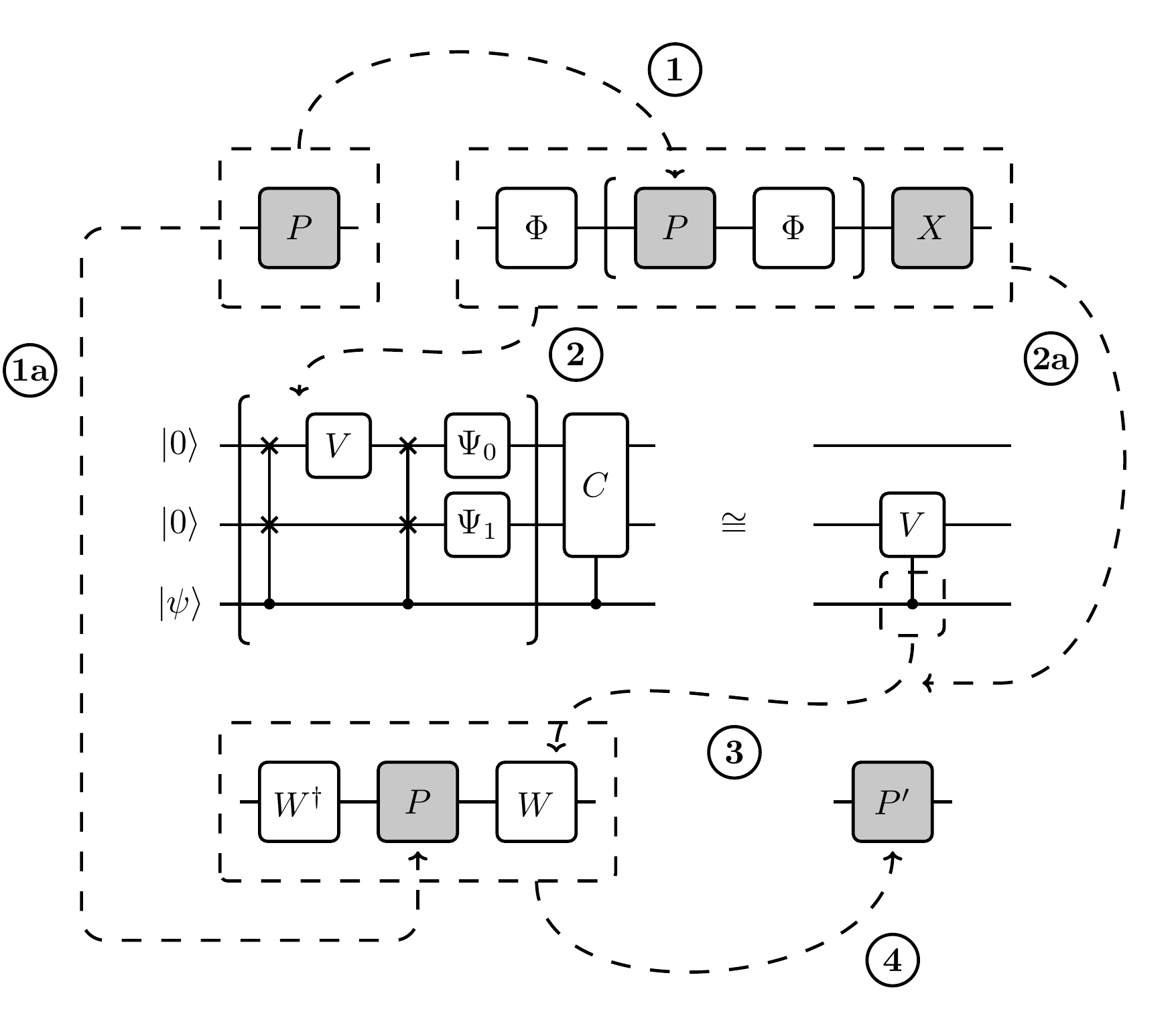}
    \caption{High-level breakdown of the correction protocol of Thm.~\ref{thm:qsp_correction}. We depict multiple pipelines for preparing, given a gadget's output leg which embeds the polynomial transform $P$, a leg which embeds $P$ in the same way but whose rotation axis is approximately fixed and known (denoted $P^\prime$). Dotted boxes indicate substitution of the enclosed circuit as a module into another circuit, and $\cong$ indicates circuits with approximately identical action (in the operator norm). Here (1) gives substitution into a QSP protocol, which generates $V$ (an approximate $z$-rotation). Using a pair of Fredkin gates and two additional qubits in fixed, known states, through (2) we can approximately apply the square root of this $z$-rotation to a qubit of our choice by creating approximate controlled-$V$ access using the pair of QSP protocols parameterized by $\Phi_0, \Phi_1$ as in Thm.~\ref{thm:controlled_z_routine}; here $C$ is a fixed, known correction unitary. If a different access assumption is made (see Tab.~\ref{tab:mean_lcu_comparison}), (2) can be skipped, as in (2a) one assumes either direct access to controlled-$V$ per Cor.~72 of \cite{gslw_qsvt_19}, or a domain restriction on the possible $V$ allows ancilla-free square roots to be taken. In (3) one conjugates the original protocol $P$ (1a), which is precisely, through (4), what was desired of $P^\prime$. Note both $W$ and $W^\dagger$ are preparable as gadget outputs are obliviously invertible for all arguments: $ZPZ \equiv P^\dagger$. Here boxes $\Phi, \Psi_0, \Psi_1$ are in general different, controllable $z$-rotations, and brackets indicate finite repetition of a block with possibly different phases.}
    \label{fig:qsp_correction}
\end{figure}


\completecompgadgets*

\begin{proof}
    We must correct each of the gadget $\mathfrak{G}$'s $e$ outputs involved in the interlink, via the protocol of Thm.~\ref{thm:qsp_correction}, before composition. The overall complexity will simply amount to the cost of the corrective protocol, multiplied by the total query complexity required to implement all of the output legs of the atomic $(c, d)$ gadget involved in the interlink, given its black-box use of input gadgets (a quantity upper-bounded by $d \lvert\Xi\rvert_{\infty}$). Of course, this query complexity will often be much lower, since the sequence of $\Xi$ will often only utilize a small number of queries to input legs involved in the interlink, which are the only legs which require correction.
    
    Moreover, given that corrective protocols used in Thm.~\ref{thm:qsp_correction} only have logarithmic dependence on $\varepsilon$, and accumulate errors linearly in circuit depth, making the transformation of $\varepsilon \mapsto \varepsilon/\lvert\Xi\rvert_{\infty}$ introduce only logarithmic factors in $\lvert\Xi\rvert_{\infty}$. This implies that we require the given $\widetilde{\mathcal{O}}(d |\Xi|_{\infty} \zeta)$ query-complexity (for $\zeta$ of Thm.~\ref{thm:qsp_correction}) to $\varepsilon$-achieve all of the unitaries of the M-QSP protocols/atomic gadget achieving functions $g_0, \dots, g_{d - 1}$, composed with the signal operators $e^{i \arccos(f_1(x)) \sigma_x}, \dots, e^{i \arccos(f_{b - 1}(x)) \sigma_x}$, for $x$ over a prescribed domain (depending on the correction protocol), with respect to the interlink. Lem.~\ref{lem:unitary_fn_approx} implies that the composite gadget will $\varepsilon$-approximate the desired function composition, $H$.
\end{proof}


\begin{remark}[Counting gadget interlinks]
    Note that given $\mathfrak{G}$ and $(\Xi, S)$ as $(a, b)$ and $(c, d)$ gadgets respectively, then there are
    \begin{equation}
            L_{(a, b), (c, d)} 
            \equiv 
            \sum_{e = 0}^{\min{(b, c)}}\binom{b}{e}\binom{c}{e} (e!)
        \end{equation}
    valid interlinks between these gadgets, and $\binom{b}{e}\binom{c}{e} (e!)$ interlinks which induce the same type transformation, but which may not in general be equivalent.
\end{remark}

\begin{definition}[Augmenting, eliding, pinning, and permuting gadgets] \label{def:aux_gadget_operations}
    There exist a number of simple operations over gadgets that take only one argument, and we discuss them here, terming them \emph{augmentation}, \emph{elision}, \emph{pinning}, and \emph{permutation}. Graphically, these correspond to adding or removing out-going legs, fixing in-going legs, or permuting out-going legs in the picture discussed in Fig.~\ref{fig:qsp_gadget_composition}. These have the following types
        \begin{align}
            (a, b) &\rightarrow (a, b + c), \quad \text{augmentation},\\
            (a, b) &\rightarrow (a, b - c), \quad \text{elision},\\
            (a, b) &\rightarrow (a - c, b), \quad \text{pinning},\\
            (a, b) &\rightarrow \mathmakebox[0pt][l]{(a, b),}\phantom{(a, b - c)\;\,} \quad \text{permutation}.
        \end{align}
    Given an $(a, b)$ gadget that achieves $\{f_0(x_0, \cdots, x_{a - 1}), f_1(x_0, \cdots, x_{a - 1}), \dots, f_{b - 1}(x_0, \cdots, x_{a - 1})\}$, the result of the above actions are gadgets which achieve the following:
        \begin{align}
            &\bigcup_{k \in B^\prime\hphantom{^\prime}} f_k(x_0, \cdots, x_{a - 1}),\\
            &\bigcup_{k \in B^{\prime\prime}} f_k(x_0, \cdots, x_{a - 1}),\\
            &\bigcup_{k \in B\hphantom{^\prime\prime}} f^{\prime}_k(x_{j_0}, \cdots, x_{j_{a - c - 1}}), \\
            &\bigcup_{k \in B\hphantom{^\prime\prime}} f_{W(k)}(x_0, \cdots, x_{a - 1}),
        \end{align}
    where $W(k)$ is the action of a permutation over $[b]$ on the index of an out-going leg of the gadget, $f_k^{\prime}$ is the function $f_k$ with some size-$c$ subset of its input arguments fixed, $B^\prime$ is a list containing $[b]$ in order, possibly interspersed with repeated elements from $[b]$, and $B^{\prime\prime}$ is an ordered sublist of $B$. We can again simply depict these operators graphically:
        \begin{center}
            \includegraphics[width=0.8\textwidth]{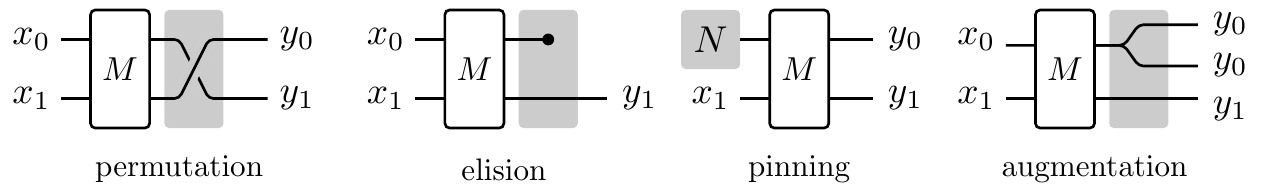}
        \end{center}
    where $M$ is chosen here to be a $(2,2)$ gadget, whose legs can be permuted, elided, pinned, or augmented, resulting in $(2,2)$, $(2, 1)$, $(1, 2)$, and $(2, 3)$ gadgets respectively. Here $N$ is some superoperator which returns a fixed unitary (often chosen by the computing party).
\end{definition}

    The cost of an operator over a gadget is most clearly expressed in the number of queries to the subsidiary gadget(s) (as their outputs are used in a black-box way), as well as any additional space required. Given an interlink (Def.~\ref{def:gadget_interlink}) defined by $\mathfrak{I} \equiv (B, C, W)$, it is evident that the permutation requires at most $\mathcal{O}(|B|)$ two-qubit gates. Moreover, for each of the $|B|$ out-going legs, the assumptions of Thm.~\ref{thm:qsp_correction} do not distinguish between their corrections, meaning that the same superoperator is applied to each, with the same asymptotic query complexity. This induces a cost dependent only on the out-going leg number and the maximum length among sub-components of the outer gadget. For augmentations (Def.~\ref{def:aux_gadget_operations}) of size $c$, a total of $c$ calls to the first gadget are necessary and sufficient (with additional space linear in $c$), while elisions and pinnings of size $c$ require only a single call, and no additional space.


To conclude the proofs of the main results, we revisit the equivalence theorem, Thm.~\ref{thm:polynomial_gadget_equivalence}.

\existence*

\begin{proof}
    To construct the desired polynomial, we will perform an $m$-fold composition of atomic gadgets, each of which have had the correction procedure applied to their output legs. We are able to bound the accumulated error within the gadget via Thm.~\ref{thm:n_fold_composition}, and we are able to guarantee that the desired polynomial is achieved to $\varepsilon$-accuracy via Thm.~\ref{thm:full_gadget_composition}. Since each gadget being considered is atomic and achieves polynomial is the variables $x_1, \dots, x_n$ (it could, of course, be constant in some of these variables), each gadget can have at most $n$ input legs. Thus, we think of each interlink $\mathfrak{I}_k$ as being between two sets of \emph{precisely} $n$ legs, where the action of a gadget on some of these legs may merely be the identity. This viewpoint, while leading to sub-optimal scaling in particular cases when interlinks are sparse, will simplify the proof of this theorem.

    Let $\mathfrak{G}^{(k)}$ be the gadget $\nu$-achieving polynomial $P^{(k)}$, with the correction procedure appended to the output legs. For the sake of simplicity, we are using the correction procedure which requires ancilla qubits, so there is no concern about half-twisted embedability. It follows that to implement a single output leg of this gadget on $\mathcal{D}_k$ requires circuit depth $\widetilde{O}(d_k \delta^{-1} \text{polylog}(\nu^{-1}))$, where $d_k = \max_j \deg(P^{(k)}_j)$ is the depth of the uncorrected gadget achieving $P^{(k)}$ (Thm.~\ref{thm:qsp_correction}). Because each of the gadgets $\mathfrak{G}^{(k)}$ has $\nu$-embeddable output, they can be interlinked in a way such that the polynomials will (approximately) compose as expected. In particular, we know that $\mathfrak{G}^{(k)}$ will $\nu$-approximate the unitary function $\mathfrak{F}^{(k)}$ with
    \begin{equation}
        \mathfrak{F}^{(k)} \left[ e^{i \arccos(x_0) \sigma_x}, \dots, e^{i \arccos(x_{n}) \sigma_x} \right] = \left( e^{i \arccos(P^{(k)}_1(x_1, \dots, x_n)) \sigma_x}, \dots, e^{i \arccos(P^{(k)}_n(x_1, \dots, x_n)) \sigma_x} \right)
    \end{equation}
    for $x \in \mathcal{D}_k$. Thus, by Thm.~\ref{thm:n_fold_composition}, it follows that the error between $\mathfrak{F}$: the composition of all the $\mathfrak{F}^{(k)}$ and the function achieved by $\mathfrak{G}$: the gadget which is a full interlink of the sequence of $\mathfrak{G}^{(k)}$ is given by
    \begin{align}
        &\left|\left| \mathfrak{G}\left[e^{i \arccos(x_0) \sigma_x}, \dots, e^{i \arccos(x_{n}) \sigma_x}\right] - \mathfrak{F}\left[e^{i \arccos(x_0) \sigma_x}, \dots, e^{i \arccos(x_{n}) \sigma_x}\right] \right|\right| \nonumber\\&\leq \displaystyle\sum_{k = -1}^{m - 1} v_{\nu}^{T} \displaystyle\prod_{k < \ell \leq m - 1} C^{(\ell)}
    \end{align}
    where $C^{(\ell)}$ is the cost matrix of the $\ell$-th gadget in the composition. Note that each of the columns of $C^{(\ell)}$ sum to some $d_{\ell}' = \widetilde{O}(d_{\ell} \text{polylog}(\nu^{-1}))$, and each entry of $v_{\nu}$ is equal to $\nu$. Therefore,
    \begin{equation}
        \displaystyle\sum_{k = -1}^{m - 1} v_{\nu}^{T} \displaystyle\prod_{k < \ell \leq m - 1} C^{(\ell)} = \widetilde{\mathcal{O}}\left( \nu d_1 \cdots d_m \left(K \delta^{-1} \text{polylog}(\nu^{-1}) \right)^{m} \right)
    \end{equation}
    for some constant $K$. Note that $D = d_1 \cdots d_m$. In addition, since $m = \mathcal{O}(\log(D))$,
    \begin{equation}
        \widetilde{O}\left( (K' \delta^{-1} \text{polylog}(\nu^{-1}))^{m} \right)  = \mathcal{O}\left( \exp \left( K'' \log(D) \log \log(\nu^{-1}) \right) \right),
    \end{equation}
    and thus an error of $\nu$ on an input generates an error of the asymptotic form $\widetilde{\mathcal{O}}(\nu D^{K''\log{\log{\nu^{-1}}}})$ on an output. If we want to achieve error at maximum some desired $\varepsilon$, one then finds that the required initial $\nu$ must satisfy the implicit equation
        \begin{equation}
            -\log{\nu^{-1}} + K''\log{\log{\nu^{-1}}}\log{D} = -\log{\varepsilon^{-1}},
        \end{equation}
    which can be rearranged into the suggestive form
        \begin{equation}
            -e^{\log{\log{\nu^{-1}}}} + K''\log{D}\log{\log{\nu^{-1}}} = -\log{\varepsilon^{-1}}
        \end{equation}
    This permits us to express $\log{\log{\nu^{-1}}}$ in terms of the Lambert W function (taking the $W_{-1}$ branch, which is valid for our intended range of small $\varepsilon$), namely
        \begin{equation}
            \log{\log{\nu^{-1}}} = -\frac{\log{\varepsilon^{-1}}}{K''\log{D}} - W_{-1}\left(-\frac{e^{-\log{\varepsilon^{-1}}/K''\log{D}}}{K''\log{D}}\right).
        \end{equation}
    For small $\varepsilon$ the the Lambert W function has simple behavior, approaching $\log{(-x)}$ for small argument $x$, meaning that the whole expression simplifies to $\log{\log{\nu^{-1}}} = \log{K'' D}$ at leading order. As such we see that for very small $\varepsilon$, $\log\nu^{-1}$ grows only as $K''D$; as this is the only multiplicative factor from the precision $\nu$ entering into the complexity of each sub-gadget, we see that their length does not grow superpolynomially in the logarithm of $\varepsilon^{-1}$. This yields the desired $\widetilde{\mathcal{O}}(\text{poly}(D)\,\text{polylog}(\varepsilon^{-1}))$ scaling.
\end{proof}


\section{Related prior techniques and performance comparisons} \label{appx:performance_comparisons}

\noindent
The resource requirements for enacting a multivariable polynomial transformation on the mutual singular values of a series of commuting block encodings requires some scaffolding to analyze, which we present here. We will call $d$ the (generalized) functional degree of the achieved multivariable polynomial, $r$ the number of variables within this transform, and $s$ is the number of additional qubits used for the original block encoding(s). In compressing these values to single scalars, we usually mean the maximum over an implicit tuple, e.g., of the degree of the desired polynomial in each variable, or the required auxiliary space of each block-encoded input. Note that the generalized degree of the implemented function $d$ generally scales as the polylogarithm of the inverse uniform approximation error $\varepsilon$ to some idealized continuous function. We will suppress this factor in the discussion below, though one is free to substitute $\text{poly}{(d)}\text{polylog}{(\varepsilon^{-1})}$ for $\text{poly}{(d)}$ where it appears, if one is trying to achieve only a polynomial approximation to a desired, ideal functional transform. We also note that $\delta$ is sometimes introduced to bound uniform approximation away from regions in which the function has an $\mathcal{O}(1)$ jump discontinuity, or a hard boundary condition; for this, $\text{poly}{(d)}\text{polylog}{(\varepsilon^{-1})}\text{poly}(\delta^{-1})$ can be substituted in the same way.

We give the defining complexities of multivariable quantum eigenvalue transformation (M-QET), which relies on techniques from linear combination of unitaries (LCU) based methods, from \cite{bss_commuting_matrices_23}. Their work does not consider successive composition of protocols; moreover works which do consider such compositions, e.g., \cite{mf_recursive_23}, do so only only in the single-variable setting, and only in the case of recursive function application. In this way our work is not directly comparable to either of these methods, as our correction procedure (Thm.~\ref{thm:qsp_correction}) confers an additional ability to not just compose functions or block encode multivariable polynomial transforms, but to link and partially-compose gadgets achieving such transforms.

Consequently, when we compare our methods to previous works below, we are doing so on the level of individual multivariable quantum eigenvalue transforms (a restriction of QSVT to normal operators). On this front both our work and \cite{bss_commuting_matrices_23} demonstrate near arbitrary block encodings of multivariable polynomials in commuting operators, allowing for general comparison.

Namely, we consider the resource costs of the basic units of our work: gadgets (Def.~\ref{def:qsp_gadget}), each of whose out-going legs yield an embeddable unitary which, if the gadget was built from the composition of atomic gadgets, can be seen to achieve a given (approximation to) a multivariable polynomial transform of the input leg variables. Viewed in this way our work approaches the problem of multivariable polynomial block encodings from a starkly different direction than \cite{bss_commuting_matrices_23}. Rather than considering the efficiency of an algorithmic method that achieves arbitrary functional block encodings, we investigate the closure of a highly-resource-efficient \emph{class of subroutines} under \emph{block encoding-preserving} superoperators, and show that we can achieve unexpectedly expressive multivariable block encodings with improved resource scaling, i.e., without the space use and post-selection requirements of LCU. In this way, our methods achieve a synthesis of the aims of \cite{bss_commuting_matrices_23} and \cite{mf_recursive_23}, albeit using tools from M-QSP \cite{rossi_m_qsp_22} language from work on semantic embedding \cite{rc_semantic_alg_23}, and abstractions from the classical theory of functional programming languages. The ultimate result is significant savings in space complexity and success probability at the cost of immediate generality of embeddable functions. It is also important to observe that, just as in initial comparisons between LCU and QSVT methods for Hamiltonian simulation \cite{bck_ham_sim_15,lc_ham_sim_17,gslw_qsvt_19}, the \emph{query-complexity's} dependence on polynomial degree and approximation error are both within polynomial of optimal already. Our benefit is however not just tied to the query complexity, but in the modular re-use of pre-computed gadgets, substantial auxiliary space savings, and ability to leverage the powerful functional-analytic results underlying the quite mature theory of QSP-based algorithms.

Comparisons of our methods along major resource axes in terms of $d, r, s$ are made in Table~\ref{tab:qsvt_lcu_comparison}. We also emphasize that our methods are not algorithmically incompatible with those of \cite{bss_commuting_matrices_23}, just resource incomparable. I.e., our gadgets are agnostic to the circuit having constructed a given block encoding used as input, and thus similar constructions to those of \cite{bss_commuting_matrices_23} can furnish immediately applicable upper bounds for certain multivariable eigenvalue transforms (following amplifying into the proper block), and in turn be fed into our gadgets if desired. In this way, we hope multiple methods for efficient block encoding can mutually inform one another, and settle sophisticated questions of the minimum-required space and query complexity to achieve certain (approximations to) desired functional transforms.

\begin{table}[htpb]
    \centering
    \begin{tabular}{l | l l l}
        Method & Query complexity\; & Space complexity & Subnormalization\\\hline
        \cite{bss_commuting_matrices_23} & $\mathcal{O}(\text{poly}{(r)} d^{\text{poly}(r)})$ & $\mathcal{O}(\text{poly}{(r)} [\text{poly}{(d)} + \text{poly}{(s)}])$ & $\mathcal{O}(\text{poly}{(d)}^{\text{poly}{(r)}})$ \\
        This work \ & $\mathcal{O}(d)$ & $\mathcal{O}(s + c)$ & $\mathcal{O}(1)$
    \end{tabular}
    \caption{Asymptotic query and space complexity, as well as block encoding sub-normalization factor for this work's multivariable block encodings and those of other methods based on LCU methods. Note that the $\varepsilon, \delta$ dependencies (i.e., $\mathcal{O}(\varepsilon)$-uniformly approximating functions on regions at most $\mathcal{O}(\delta)$-close to discontinuities or singularities) are here suppressed, but are in general identical between methods \cite{gslw_qsvt_19}, and that the polynomial scaling described for LCU methods are often low-degree or linear in practice. Here subnormalization relates roughly linearly to the number of runs required to measure the proper sub-block of the block encoding. Note also that this comparison does not consider the new ability, proposed in this paper, to curry gadgets comprising M-QSP protocols in a semantically clear way, and instead considers the complexity of an initial atomic gadget's construction. Note that the variable $c$ is a constant that can be zero for half-twisted-embeddable protocols (Def.~\ref{def:embeddable}).}
    \label{tab:qsvt_lcu_comparison}
\end{table}

Beyond query and space complexity, we take some time to discuss additional benefits to our method in relation to observations made by previous works. One of these relates to the fact that, if a gadget uses a discrete set of phases, and a composite computation is built from repeated interlinks of that gadget with itself (or some other finite set of gadgets), then the phases used in the resulting circuit necessarily come from the union of phases used by the constituting gadgets. As discussed by \cite{mf_recursive_23}, there exist settings \cite{tltc_alec_23} in which the number of distinct QSP phases used in a protocol defines the complexity of certain error-correction procedures on that protocol. It is an intriguing open question whether the limited distinct phase character possible with our methods (already an experimental advantage) can be exploited in the same way as those of the single variable setting \cite{tltc_alec_23}. Moreover, even if a variety of gadgets are used, those most often repeated, or prone to error, can be selectively optimized, given their ubiquity, and applied wherever given our method's preservation of gadget contiguity as subroutines.

As a concluding note, we repeat that the comparison made here cannot be entirely fair, due to the different goals of the previous works considered and our own. We assume, within Table~\ref{tab:qsvt_lcu_comparison}, that a given multivariable polynomial transformation \emph{can} be reached by both methods, and then show that the query complexities, space complexities, and subnormalization factors differ. Nevertheless, the transforms possible with our method are sufficiently expressive to warrant comparison within this implied overlap, and moreover, a wealth of results exist in classical signal processing showing the success of considering composite protocols comprising multiple uses of inter-related and simpler sub-filters \cite{kh_sharpening_77,saramaki_cascade_87}. 

\begin{remark}[On the general inequivalence of composite gadgets] \label{rem:gadget_inequivalence}
    As discussed previously, as well as in the caption of Tab.~\ref{tab:qsvt_lcu_comparison}, there exist non-obvious conditions on polynomial transforms in M-QSP which effect the space and query complexity of the correction procedures used to compose gadgets. This is in tandem with the (approximate) polynomial decomposition theorems \cite{dso_exact_approx_decomp_13, cgjw_approx_decomp_99, gm_approx_decomp_07, kl_poly_decomp_89} specifying how to break a composite polynomial into smaller sub-protocols combined through composition. Both of these subtleties to understanding the resource requirements of achieving a given transform are further confounded by the non-trivial inequivalence of gadgets which apparently achieve the same transform. For instance, consider the simple $(1,1)$ product gadget (Ex.~\ref{ex:product_gadget}, with a pre-applied square root) which outputs $x_0^2$, as well as the $(1,1)$ gadget which squares an oracle, achieving $2x_0^2 - 1$, and averages it with the $(1,1)$ trivial gadget achieving $1$ (using the $(2,1)$ gadget of Ex.~\ref{ex:sum_gadget}) to achieve the same $x_0^2$. The averaging gadget, by merit of using the product gadget as a subroutine, has strictly higher query complexity, complicating the question of determining the minimal composition of gadgets achieving a given function. Evidently, the resources required for a given functional transform (or its approximate version) are sensitive to small changes in objective or constitutive subroutines.
\end{remark}

The relative incomparability of our work and those of others extends further, considering that our work explicitly characterizes the cost of gadget composition (Sec.~\ref{sec:gadget_cost}), in direct generalization of techniques discussed in \cite{mf_recursive_23, rossi_m_qsp_22}, which considered the simpler single-variable setting. Modular filter design, from which our work draws some inspiration, is classically well-developed, as mentioned, and undergirded by a correspondingly rich theory of polynomial decomposition theorems, whose approximate forms and associated algorithms \cite{dso_exact_approx_decomp_13, gm_approx_decomp_07, kl_poly_decomp_89}, can be used to approach the generality offered \cite{bss_commuting_matrices_23}. By focusing on smaller, modular, and well-characterized subsets of possible transforms, algorithmic clarity is improved, automatic optimization and verification becomes possible, and experimental realization is simpler.


\section{On computing the cost of gadgets} \label{sec:gadget_cost}

\noindent General gadgets are composite objects with multiple inputs and outputs, and consequently defining their implementation cost requires some bookkeeping. Here, we define basic objects associated with this cost, such that the cost of compositions of (possibly composite) gadgets is simple to compute.

\begin{definition}[Gadget cost matrix] \label{def:cost_matrix}
    Let $\mathfrak{G}$ be an $(m, n)$ gadget. Then there is a $m\times n$ \emph{cost matrix} $C$ associated with $\mathfrak{G}$ 
    \begin{equation} \label{eq:cost_matrix}
        \begin{gathered}
            \includegraphics[width=0.68\textwidth]{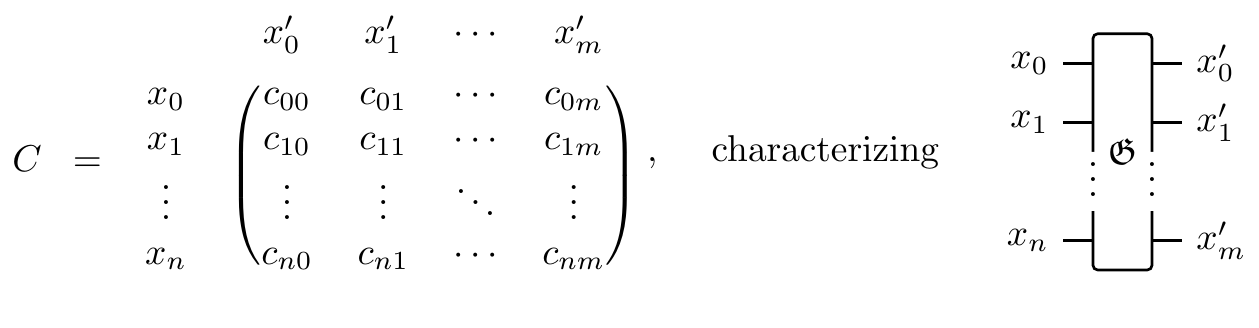}
        \end{gathered}
    \end{equation}
    where element $c_{jk}$ indicates the required number of queries to the input leg encoding $x_{j}$ to achieve an output leg encoding $x_{k}^\prime$ up to some precision $\varepsilon_k \geq 0$ for $x_{j}, j \in [n]$ in a particular domain $\mathcal{D}_j$.
\end{definition}

It is not so difficult to determine these cost matrices in common cases. For instance, let $(\Xi, S)$ define an atomic gadget $\mathfrak{G}$; then the cost matrix $C$ of $\mathfrak{G}$ has elements $c_{jk} = S_{jk}$, where $S_{jk}$ is the number of occurrences of $j$ in the $k$-th element of $S$. Furthermore, if one considers the correction procedure defined previously, then the cost matrix of a corrected gadget has each $c_{jk}$ for the non-corrected gadget multiplied by a factor $\zeta$ determined by Thm.~\ref{thm:qsp_correction}, for a desired precision $\varepsilon$ to an embeddable output, over a desired range. For our purposes, we care primarily about the asymptotic behavior of elements of the cost matrix, and will thus treat them up to logarithmic factors. Cost matrices provide a clean formalism in which we can compute the cost of implementing composite gadgets.

\begin{lemma}[Cost matrix for a composite gadget] \label{lem:cost_matrix_composition}
    We compute the cost matrix of a composite gadget. Let $\mathfrak{G}$ be an $(a, b)$ gadget and $\mathfrak{G}^\prime$ a $(c, d)$ gadget. Given an interlink $\mathfrak{I} \equiv (\tilde{B}, \tilde{C}, W)$ (Def.~\ref{def:gadget_interlink}, where tildes have been added to prevent variable overloading), we know that the resulting gadget has type $(a + c - e, b + d - e)$ where $e \equiv |\tilde{B}| = |\tilde{C}|$ is the number of contracted legs. We can define block versions of the cost matrices $C, C^\prime$
        \begin{equation}
            C =
            \left[
            \begin{array}{c|c}
                A & A^\prime
            \end{array}
            \right],
            \quad
            C^\prime =
            \left[
            \begin{array}{c}
                B\\\hline
                B^\prime
            \end{array}
            \right],
        \end{equation}
    where $A$ is $a\times e$, $A^\prime$ is $a\times(b - e)$, $B$ is $e\times d$, and $B^\prime$ is $(c - e)\times d$; this will result in $C^{\prime\prime}$ as defined below as $(a + c - e)\times(d + b - e)$, as expected. Note that input and output variables have been permuted such that $A, B$ encompass the marked subsets $\tilde{B}, \tilde{C}$ of the interlink. The block version of the cost matrix $C^{\prime\prime}$ for the composition according to the interlink $\mathfrak{I}$ can then be simply defined in terms of $C, C^\prime$:
        \begin{equation}
            C^{\prime\prime} =
            \left[
            \begin{array}{c | c}
                 W(A)B & A^\prime \\\hline
                 B^\prime & 0
            \end{array}
            \right].
        \end{equation}
    Here $W$, from the interlink, possibly permutes the columns of $A$, and $W(A)B$ is the standard matrix product between $W(A)$ and $B$. The bottom right block is all zeros. Note that in the case that $\mathfrak{G}^\prime$ is an atomic gadget, the cost matrix $C^\prime$ can be explicitly computed as discussed previous to this lemma.
\end{lemma}

It is apparent that the general cost matrix of a series of interlinked gadgets can be quite complex, with elements that can grow quickly in the number of compositions. It follows that in the generic case, using cost matrices to determine the required input-leg precision to achieve a desired output-leg precision is a task very particular to individual gadget networks. However, certain generalizations can be made, provided a gadget is of a particular form. To conclude this section, we present a large class of gadgets for which the cost matrix formalism allows computation of accumulating errors upon iterated, complete gadget composition to go forward easily.

\begin{definition}[Circuit gadget] \label{def:circuit_gadget}
    Let an $(a, b)$ gadget $\mathfrak{G}$ be a \emph{circuit gadget} if it achieves unitaries $U_k'$ via mapping input unitaries $U_k$ to an $m$-qubit quantum circuit, where the unitary $U_k'$ is achieved via $(c_{0k}, \dots, c_{(a - 1) k})$ (the $k$-th column of a cost matrix) queries to $(U_0, \dots, U_{a - 1})$ on any of the $m$ qubits, interspersed with fixed $m$-qubit unitaries.
\end{definition}

\noindent Clearly, antisymmetric atomic gadgets are a particular type of circuit gadget. There is a simple upper-bound on the error accumulated within circuit gadgets.

\begin{restatable}[Linear error accumulation in circuit gadgets]{theorem}{circuitgadgeterror}
    \label{thm:linear_error_gadgets}
   Suppose $\mathfrak{G}$ is an $(a, b)$ circuit gadget. In addition, suppose $V_k$ and $V'_k$ are sets of unitaries such that $||V_k - V'_k|| \leq \varepsilon_k$, for each $k \in [a]$. Then the output error of the $k$-th output unitary achieved by the gadget can be upper-bounded as
   \begin{equation}
       || \mathfrak{G}[V_0, \dots, V_{a - 1}]_k - \mathfrak{G}[V'_0, \dots, V'_{a - 1}]_k || \leq \displaystyle\sum_{j = 0}^{a - 1} c_{jk} \varepsilon_j = v_{\varepsilon}^{T} C e_k
   \end{equation}
   where $v_{\varepsilon} = (\varepsilon_0, \dots, \varepsilon_{a - 1})$ and $e_k$ is the $k$-th standard basis vector.
\end{restatable}

\begin{proof}
    By definition of the circuit gadget, we can write
    \begin{equation}
        \mathfrak{G}[V_0, \dots, V_{a - 1}]_k = W_0 \displaystyle\prod_{j = 1}^{M} V_{s_j} W_{j}
    \end{equation}
    for some length-$M$ sequence $s \in [b]^{M}$, where $M = \sum_{j} C_{jk}$ (the total number of queries to input oracles characterizing the $k$-th output leg), and each $W_j$ is a fixed $M$-qubit unitary. It follows that we can bound the difference of gadget actions via a telescoping series,
    \begin{align}
    \left|\left| \mathfrak{G}[V_0, \dots, V_{a - 1}]_k - \mathfrak{G}[V'_0, \dots, V'_{a - 1}]_k \right|\right| &= \left|\left| W_0 \displaystyle\prod_{j = 1}^{M} V_{s_j} W_{j} - W_0 \displaystyle\prod_{j = 1}^{M} V'_{s_j} W_{j} \right|\right| 
        \nonumber
        \\ & = \left|\left| \displaystyle\sum_{i = 1}^{M} \left[ \displaystyle\prod_{j < i} V_{s_j} W_j \right] (V_{s_{i}} W_i - V'_{s_i} W_i) \left[ \displaystyle\prod_{j > i} V'_{s_j} W_j \right] \right|\right| \nonumber
        \\ & \leq \displaystyle\sum_{i = 1}^{M} \left|\left| V_{s_i} W_i - V'_{s_i} W_i \right|\right| \nonumber
        \\ & = \displaystyle\sum_{i = 1}^{M} ||V_{s_i} - V'_{s_i}|| \leq \displaystyle\sum_{j = 1}^{a - 1} C_{jk} \varepsilon_j = v_{\varepsilon}^{T} C e_k,
        \end{align}
        and the proof is complete.
\end{proof}

\noindent In general, in a sequence of gadget compositions, there are two sources of error: error in the unitaries inputted to a particular gadget, and error due to the imperfections in the gadget protocol itself. Both types of error can be understood in the the following results. First, consider what is essentially a corollary of the previous theorem.

\begin{corollary}[Output error in circuit gadgets]
\label{cor:output_error}
Let $\mathfrak{G}$ be an $(a, b)$ circuit gadget with cost matrix $C$ which achieves, in its $j$-th leg, an $\varepsilon'_j$-approximation of the function $\mathfrak{F}_j$, mapping $a$ individual unitaries to a single unitary, over the domain $\mathcal{D}$, so for $W_0, \dots, W_{a - 1} \in \mathcal{D}$,
\begin{equation}
    \left|\left| \mathfrak{G}[W_0, \dots, W_{a - 1}]_j - \mathfrak{F}_j[W_0, \dots, W_{a - 1}] \right|\right| \leq \varepsilon'_j.
\end{equation}
for each $j$. Then, given sets of unitaries $V_0, \dots, V_{a - 1}$ and $V'_0, \dots, V_{a - 1}'$ for which $||V_k - V'_k|| \leq \varepsilon_k$ for $k \in [a]$ with $V'_k \in \mathcal{D}$, $\mathfrak{G}[V_0, \dots, V_{a - 1}]_j$, for all $j \in [b]$, $(v_{\varepsilon}^{T} C e_j  + \varepsilon'_j)$-approximately achieves $\mathfrak{F}_j[V'_0, \dots, V'_{a - 1}]$, where $v_{\varepsilon} = (\varepsilon_0, \dots, \varepsilon_{a - 1})$.
\end{corollary}
\begin{proof}
    Via Thm.~\ref{thm:linear_error_gadgets}, $|| \mathfrak{G}[V_0, \dots, V_{a - 1}]_j - \mathfrak{G}[V'_0, \dots, V'_{a - 1}]_j || \leq v_{\varepsilon} C e_j$. Then, from a triangle inequality,
    \begin{multline}
        || \mathfrak{G}[V_0, \dots, V_{a - 1}]_j - \mathfrak{F}_j[V'_0, \dots, V'_{a - 1}]|| \leq ||\mathfrak{G}[V_0, \dots, V_{a - 1}]_j - \mathfrak{G}[V'_0, \dots, V'_{a - 1}]_j|| \\ + ||\mathfrak{G}[V'_0, \dots, V'_{a - 1}]_j - \mathfrak{F}_j[V'_0, \dots, V'_{a - 1}]|| \leq v_{\varepsilon} C e_j + \varepsilon'_j,
  \end{multline}
  which completes the proof.
\end{proof}

\noindent Immediately, we are able to prove more sophisticated results about the error accumulation in $n$-fold gadget compositions.

\begin{restatable}[Error in $n$-fold composition of circuit gadgets]{theorem}{nfoldgadget}
\label{thm:n_fold_composition}
Let $\mathfrak{G}^{(k)}$ be $(a_k, a_{k + 1})$ circuit gadgets for $k$ ranging from $0$ to $n - 1$ with cost matrices $C^{(k)}$. Let $\mathfrak{G}$ be the circuit gadget obtained from performing full composition of gadget legs, $\mathfrak{G} = \mathfrak{G}^{(n - 1)} \circ \cdots \circ \mathfrak{G}^{(0)}$. Suppose the gadget $\mathfrak{G}^{(k)}$ achieves, in its $j$-th leg, an $\varepsilon^{(k)}_j$-approximation of the unitary-valued function $\mathfrak{F}_j^{(k)}$ over the domain $\mathcal{D}^{(k)}$. This is to say that for all  $W_0, \dots, W_{a_k - 1} \in \mathcal{D}^{(k)}$,
\begin{equation}
    \left|\left| \mathfrak{G}^{(k)}[W_0, \dots, W_{a_k - 1}]_j - \mathfrak{F}_j^{(k)}[W_0, \dots, W_{a_k - 1}] \right|\right| \leq \varepsilon_j^{(k)},
\end{equation}
for each $j$. Moreover, suppose $\mathfrak{F}_j^{(k)}[\mathcal{D}^{(k)}] \subset \mathcal{D}^{(k + 1)}$ for all $j$ and all $k$. Let $\mathfrak{F}^{(k)} = (\mathfrak{F}^{(k)}_0, \dots, \mathfrak{F}^{(k)}_{a_{k + 1}})$ be the multi-output unitary function. Let $\mathfrak{F} = \mathfrak{F}^{(n - 1)} \circ \cdots \circ \mathfrak{F}^{(0)}$. Consider sets of unitaries $V_0, \dots, V_{a_0 - 1}$ and $V_0', \dots, V_{a_0 - 1}'$ with $V_k' \in \mathcal{D}^{(0)}$, and $||V_j - V_j'|| \leq \varepsilon^{(-1)}_{j}$ for each $j$. Note that $k = -1$ indexing is used for notational convenience. The composition error is bounded via multiplication and summation of cost matrices, as
\begin{equation}
    \left|\left| \mathfrak{G}[V_0, \dots, V_{a_0 - 1}]_j - \mathfrak{F}[V'_0, \dots, V'_{a_0 - 1}]_j \right|\right| \leq \displaystyle\sum_{k = -1}^{n - 1} v_{\varepsilon^{(k)}}^{T} \left[ \displaystyle\prod_{k < \ell \leq n - 1} C^{(\ell)} \right] e_j,
\end{equation}
where $v_{\varepsilon^{(k)}} = (\varepsilon^{(k)}_{0}, \dots, \varepsilon^{(k)}_{a_{k + 1} - 1})$.
\end{restatable}

\begin{proof}
We can proceed via induction. Cor.~\ref{cor:output_error} proves the case of $n = 1$ immediately. We assume the case of $n$, and prove the case of $n + 1$ via induction. Letting $\mathfrak{G}' = \mathfrak{G}^{(n - 1)} \circ \cdots \circ \mathfrak{G}^{(0)}$ and $\mathfrak{F}' = \mathfrak{F}^{(n - 1)} \circ \cdots \mathfrak{F}^{(0)}$, note that $\mathfrak{F}'[V_0', \dots, V'_{a_0 - 1}] \in \mathcal{D}^{(n)}$. Thus,
\begin{align}
    \left|\left| \mathfrak{G}^{(n)}[\mathfrak{F}'[V'_0, \dots, V'_{a_0 - 1}]]_j - \mathfrak{F}_j^{(n)}[\mathfrak{F}'[V'_0, \dots, V'_{a_0 - 1}]] \right|\right| \leq \varepsilon_j^{(n)}
\end{align}
for each $j$. By the inductive hypothesis, we know that
\begin{equation}
    || \mathfrak{G}'[V_0, \dots, V_{a_0 - 1}]_j - \mathfrak{F}'[V'_0, \dots, V'_{a_0 - 1}]_j || \leq \displaystyle\sum_{k = -1}^{n - 1} v_{\varepsilon^{(k)}}^{T} \left[ \displaystyle\prod_{k < \ell \leq n - 1} C^{(\ell)} \right] e_j
\end{equation}
for each $j$. Let $v$ be the vector containing each of these errors, so
\begin{align}
    v^T &= \left( \displaystyle\sum_{k = -1}^{n - 1} v_{\varepsilon^{(k)}}^{T} \left[ \displaystyle\prod_{k < \ell \leq n - 1} C^{(\ell)} \right] e_0, \dots, \displaystyle\sum_{k = -1}^{n - 1} v_{\varepsilon^{(k)}}^{T} \left[ \displaystyle\prod_{k < \ell \leq n - 1} C^{(\ell)} \right] e_{a_n - 1} \right)
    \\ & = \displaystyle\sum_{k = -1}^{n - 1} v_{\varepsilon^{(k)}}^{T} \left[ \displaystyle\prod_{k < \ell \leq n - 1} C^{(\ell)} \right].
\end{align}
It follows from Thm.~\ref{thm:linear_error_gadgets} that
\begin{align}
    || \mathfrak{G}^{(n)}[\mathfrak{F}'[V'_0, \dots, V'_{a_0 - 1}]]_j &- \mathfrak{G}^{(n)}[\mathfrak{G}'[V_0, \dots, V_{a_0 - 1}]]_j || \leq v^{T} C^{(n)} e_j
    \\ & \leq \displaystyle\sum_{k = -1}^{n - 1} v_{\varepsilon^{(k)}}^{T} \left[ \displaystyle\prod_{k < \ell \leq n - 1} C^{(\ell)} \right] C^{(n)} e_j \leq \displaystyle\sum_{k = -1}^{n - 1} v_{\varepsilon^{(k)}}^{T} \left[ \displaystyle\prod_{k < \ell \leq n} C^{(\ell)} \right] e_j,
\end{align}
so via a triangle inequality,
\begin{align}
     \left|\left| \mathfrak{G}^{(n)}[\mathfrak{G}'[V_0, \dots, V_{a_0 - 1}]]_j - \mathfrak{F}_j^{(n)}[\mathfrak{F}'[V'_0, \dots, V'_{a_0 - 1}]] \right|\right| &\leq \varepsilon_j^{(n)} + \displaystyle\sum_{k = -1}^{n - 1} v_{\varepsilon^{(k)}}^{T} \left[ \displaystyle\prod_{k < \ell \leq n} C^{(\ell)} \right] e_j
     \\ &=\displaystyle\sum_{k = -1}^{n} v_{\varepsilon^{(k)}}^{T} \left[ \displaystyle\prod_{k < \ell \leq n} C^{(\ell)} \right] e_j.
\end{align}
This completes the proof.
\end{proof}

\begin{remark}[Domain nesting property]
\label{rem:domain_nesting}
As a final point, it is important to note that when gadgets are composed in networks, they will often have be required to have different domain constraints on which they achieve their corresponding functions, in order to achieve a valid $n$-fold composition. In particular, if we wish to compose the functions induced by gadgets $\mathfrak{G}^{(0)}, \dots, \mathfrak{G}^{(n - 1)}$, which achieve functions $f_0, \dots, f_{n - 1}$ on domains $\mathcal{D}_0, \dots, \mathcal{D}_{n - 1}$ respectively, via interlinks $\mathfrak{I}_0, \dots, \mathfrak{I}_{n - 2}$ the domains must satisfy the \textit{nesting property} in order to induce the function $f_{n - 1} \circ_{\mathfrak{I}_{n - 2}} \cdots \circ_{\mathfrak{I}_0} f_0$ over the domain $\mathcal{D} = \mathcal{D}_0$: that is to say that the sequence of sets $\mathcal{D}'_k = (f_k \circ_{\mathfrak{I}_{k - 1}} \cdots \circ_{\mathfrak{I}_0} f_{0})(\mathcal{D}_0)$ for all $k$ from $0$ to $n - 1$ must be valid, in the sense that each $\mathcal{D}'_k$ must be well-defined. For example, in the case that each interlink is complete composition of functions, we must have $\mathcal{D}_k' = f_k(\mathcal{D}_k) \subset \mathcal{D}_{k + 1}$ for each $k$ to satisfy the domain nesting property. This constraints was alluded to in Thm.~\ref{thm:n_fold_composition}, where the successive domains satisfy this constraint. Satisfaction of this constraint will often have non-trivial effects on the cost of composing gadgets, and care must be exercised to ensure that this condition is always met within a gadget network.
\end{remark}


\section{On QSP variants, constraints, protocols, and QSP-derived gadgets} \label{appx:variants_qsp}

\noindent
This section summarizes the main results of quantum signal processing (QSP) \cite{lyc_equiangular_16, lc_ham_sim_17, lc_qubitization_19} and multivariate quantum signal processing (M-QSP) \cite{rossi_m_qsp_22}, and clarifies the associated notation. As a supplement to the definition of M-QSP given in the main body (Def.~\ref{def:m_qsp_func_prog}), we depict the associated circuits again in Fig.~\ref{fig:m_qsp_form}. Recall the original result which characterizes the standard QSP unitary.

\begin{theorem}[Quantum signal processing; following \cite{gslw_qsvt_19, mrtc_unification_21}]
\label{thm:qsp}
Let $k \in \mathbb{N}$; there exists $\Phi = \{\phi_0, \dots, \phi_k\}\in \mathbb{R}^{k + 1}$ such that $\forall x \in [-1, 1]$,
\begin{equation}
    U_{\text{QSP}}(x; \Phi) \coloneqq \Phi[W(x)] = e^{i \phi_0 \sigma_z} \displaystyle\prod_{j = 1}^{k} \left( W(x) e^{i \phi_j \sigma_z} \right) = \begin{pmatrix} P(x) & i Q(x)
    \sqrt{1 - x^2} \\ i Q^{*}(x) \sqrt{1 - x^2} & P^{*}(x) \end{pmatrix}
\end{equation}
where $\sigma_z$ is the Pauli-$Z$ operation, and
\begin{equation}
    \label{eq:wx}
    W(x) \coloneqq \begin{pmatrix} x & i \sqrt{1 - x^2} \\ i \sqrt{1 - x^2} & x \end{pmatrix} = e^{i \arccos(x) \sigma_x},
\end{equation}
where $\arccos(x) \in [0, \pi]$, if and only if complex polynomials $P(x), Q(x) \in \mathbb{C}[x]$ satisfy the following conditions:
\begin{enumerate}
    \item $\deg(P) \leq k$ and $\deg(Q) \leq k - 1$
    \item $P$ has parity $k \ \text{mod} \ 2$ and $Q$ has parity $(k - 1) \ \text{mod} \ 2$
    \item For all $x \in [-1, 1]$, $|P(x)|^2 + (1 - x^2) |Q(x)^2| = 1$
\end{enumerate}
\end{theorem}

\noindent
Also recall the main result characterizing the unitaries generated by M-QSP protocols.

\begin{theorem}[Laurent-picture M-QSP, following \cite{rossi_m_qsp_22}]
Let $\Phi = \{\phi_0, \dots , \ \phi_k\} \in \mathbb{R}^{k + 1}$ be a sequence of phase angles and let $s = \{s_1, \dots , s_k\} \in \{0, 1\}^{k}$. Then

\begin{equation}
    U_{\text{M-QSP}}(x_1, x_2 ; \Phi, s) \coloneqq \Phi[W(x_1), W(x_2)] = e^{i \phi_0 \sigma_z} \displaystyle\prod_{j = 1}^{k} W(x_1)^{s_j} W(x_2)^{1 - s_j} e^{i \phi_j \sigma_z} = \begin{pmatrix} P & Q \\ -Q^{*} & P^{*} \end{pmatrix}
\end{equation}
where 
\begin{equation}
W(x) \coloneqq \frac{x + x^{-1}}{2} \mathbb{I} + \frac{x - x^{-1}}{2} \sigma_x
\end{equation}
for all $(x_1, x_2) \in \mathbb{T}^2$ if and only if $P, Q \in \mathbb{C}[x_1, x_2]$ are Laurent polynomials in $x_1$ and $x_2$, satisfying the following conditions:
\begin{enumerate}
    \item $\text{deg}(P) \pleq (m, n - m)$ and $\text{deg}(Q) \pleq (m, n - m)$
    where $n = |s|$ is the Hamming weight of $s$.
    \item $P$ has parity $n \ \text{mod} \ 2$ under $(x_1, x_2) \mapsto (x_1^{-1}, x_2^{-1})$ and $Q$ has parity $(n - 1) \ \text{mod} \ 2$ under $(x_1, x_2) \mapsto (x_1^{-1}, x_2^{-1})$.
    \item $P$ has parity $m \ \text{mod} \ 2$ under $x_1 \mapsto -x_1$ and parity $m - n \ \text{mod} \ 2$ under $x_2 \mapsto -x_2$. $Q$ has parity $m - 1 \ \text{mod} \ 2$ under $x_1 \mapsto -x_1$ and parity $m - n - 1 \ \text{mod} \ 2$ under $x_2 \mapsto -x_2$.
    \item For all $(x_1, x_2) \in \mathbb{T}^2$, we have $|P|^2 + |Q|^2 = 1$.
    \item The FRT = QSP condition (that satisfaction of the multivariable Fejér-Riesz lemma, and thus a unitary completion, implies the existence of QSP phases), formulated in~\cite{rossi_m_qsp_22} as a conjecture, holds.
\end{enumerate}
\end{theorem}

\begin{figure}[htpb]
    \centering
    \includegraphics[width=0.9\textwidth]{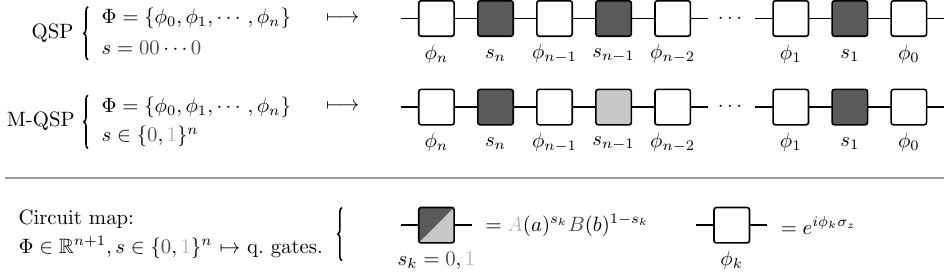}
    \caption{Explicit depiction of the circuit forms of QSP and M-QSP, following Def.~\ref{def:m_qsp_func_prog}, simplified to the case of two variables. Protocols are specified by an ordered list of real phases $\Phi$ and an ordered list of oracle labels $s$, where $s$ prescribes the order in which oracles (from a set of possible oracles) are applied, and $\Phi$ the $\sigma_z$-rotations interspersing them. Standard QSP is the special case where $s$ consists of a single repeated label. Such circuits are also given by Def.~\ref{def:m_qsp_func_prog}, and depicted in (a) of Fig.~\ref{fig:composing_qsp_gadgets} using the same style. Note that this notation is distinct from that of linked gadgets, as discussed in Rem.~\ref{rem:qsp_qsvt_pictures}, and shown in (c) of Fig.~\ref{fig:composing_qsp_gadgets}.}
    \label{fig:m_qsp_form}
\end{figure}

When it comes to embedding particular polynomials $P$ and $Q$ in QSP and M-QSP protocols, it becomes necessary to formulate existence results which allow one, given some desired polynomial $P$, to find a polynomial $Q$ such that $P$ and $Q$ satisfy the constraints of QSP or M-QSP.~\cite{gslw_qsvt_19} proves a range of existence results for the case of standard, single-variable QSP. We restate some of these results below.

\begin{theorem}
\label{thm:p_exist}
    Let $Q \in \mathbb{C}[x]$. There exists a $P$ such that $P$ and $Q$ satisfy the conditions set forth in Thm.~\ref{thm:qsp} if and only if
    \begin{enumerate}
        \item For all $x \in [-1, 1]$, $\sqrt{1 - x^2} |Q(x)| \leq 1$.
        \item In the case that $k$ is odd (so $Q$ has even parity), for all $x \in \mathbb{R}$, $(1 + x^2) Q^{*}(ix) Q(ix) \geq 1$.
    \end{enumerate}
    Similarly, given $P \in \mathbb{C}[x]$, there exists a $Q$ such that $P$ and $Q$ satisfy the conditions of Thm.~\ref{thm:qsp} if and only if
    \begin{enumerate}
        \item For all $x \in [-1, 1]$, $|P(x)| \leq 1$.
        \item For all $x \in (-\infty, -1] \cup [1, \infty)$, $|P(x)| \geq 1$.
        \item In the case that $k$ is even (so $P$ has even parity), for all $x \in \mathbb{R}$, $P^{*}(ix) P(ix) \geq 1$.
    \end{enumerate}
\end{theorem}

\begin{theorem}
\label{thm:qsp_real}
   Take $k \in \mathbb{N}$, let $\tilde{P}, \tilde{Q} \in \mathbb{R}[x]$. There exist $P, Q \in \mathbb{C}[x]$ such that $P$ and $Q$ with $\tilde{P} = \mathrm{Re}[P]$ and $\tilde{Q} = \mathrm{Re}[Q]$ satisfying conditions 1-3 of Thm.~\ref{thm:qsp}, if and only if $P$ and $Q$ satisfy conditions 1-2 of Thm.~\ref{thm:qsp} and $|\tilde{P}(x)|^2 + (1 - x^2) |\tilde{Q}(x)|^2 \leq 1$.
\end{theorem}
\noindent
Generally speaking, proving existence results for pairs of QSP polynomials is much harder in the multivariate case, as is discussed in Sec.~\ref{sec:introduction}. It is this very issue, among others, that this paper attempts to address with the techniques introduced.


\subsection{QSP on \texorpdfstring{$\sigma_z$}{sigma-z}-rotations}
\label{appx:z_qsp}

\noindent
We proceed by describing a handful of QSP-related results which are ``folkloric" within the existing QSP literature, and will be of use in proofs throughout the appendices, in particular those relating to the root protocols of Appx.~\ref{appx:root_without_ancilla}. In a standard QSP protocol, the ``signal operator" is taken to be a fixed $\sigma_x$-rotation. Via conjugation by Hadamards, a $\sigma_x$-rotation is clearly dual to a $\sigma_z$-rotation, and vice versa. Thus, given access to a $\sigma_z$-rotation oracle, one can perform a simple rotation where the oracle becomes a $\sigma_x$-rotation, perform a QSP protocol with this signal, and rotate back into the $\sigma_z$-picture. It is protocols of this form that are considered in this section.

As was stated, $H e^{i \sigma_z t} H = e^{i H \sigma_z H t} = e^{i \sigma_x t}$. It follows from this fact that when $t \in [0, \pi]$,
\begin{equation}
    \Phi[H e^{i \sigma_z t} H] = \Phi[e^{i \sigma_x t}] = U_{\text{QSP}}(\cos(t); \Phi),
\end{equation}
which is implied by Thm.~\ref{thm:qsp}. Now, consider the case when $t \in [-\pi, 0]$. It is true that $e^{i \sigma_x t} = \sigma_z e^{-i \sigma_x t} \sigma_z$, where $-t \in [0, \pi]$, so it follows that in this case,
\begin{equation}
    \Phi[H e^{i \sigma_z t} H] = \sigma_z \Phi[e^{-i \sigma_x t}] \sigma_z = \sigma_z U_{\text{QSP}}(\cos(-t); \Phi) \sigma_z = \sigma_z U_{\text{QSP}}(\cos(t); \Phi) \sigma_z.
\end{equation}
This implies the following result.
\begin{lemma}
   Given $t \in [-\pi, \pi]$,
   \begin{equation}
       \Phi[H e^{i \sigma_z t} H] = \begin{pmatrix} P(\cos(t)) & i \sin(t) Q(\cos(t)) \\ i \sin(t) Q^{*}(\cos(t)) & P^{*}(\cos(t))\end{pmatrix}
   \end{equation}
   where $P$ and $Q$ are the QSP-polynomials generated by $\Phi$.
\end{lemma}
\noindent
Going forward, for some $\Phi$, let $H \Phi[H e^{i \sigma_z t} H] H = \Phi'[e^{i \sigma_z t}]$. We call the operator $\Phi'$ a $\sigma_z$-QSP protocol. It is straightforward to verify that
\begin{align}
    \label{eq:z_1}
    \langle 0 | \Phi'[e^{i \sigma_z t}] | 0\rangle &= \text{Re}[P](\cos(t)) + i \sin(t) \text{Re}[Q](\cos(t)) \\
    \label{eq:z_2}
     \langle 0 | \Phi'[e^{i \sigma_z t}] | 1\rangle &= \text{Im}[P](\cos(t)) + i \sin(t) \text{Im}[Q](\cos(t))
\end{align}
as well as $\langle 1 | \Phi'[e^{i \sigma_z t}] | 1\rangle = \langle 0 | \Phi'[e^{i \sigma_z t}] | 0\rangle^{*}$ and $\langle 1 | \Phi'[e^{i \sigma_z t}] | 0\rangle = -\langle 0 | \Phi'[e^{i \sigma_z t}] | 1\rangle^{*}$. Now, suppose $p$ is a degree-$n$ Laurent polynomial $p(x) = \sum_{j = -n}^{n} p_j x^{j}$ in $\mathbb{R}[x, x^{-1}]$. Clearly,
\begin{equation}
    \label{eq:exp_poly}
    p(e^{it}) = \displaystyle\sum_{j = -n}^{n} p_j e^{i j t} = \displaystyle\sum_{j = -n}^{n} p_j \left[ \cos(jt) + i \sin(jt) \right] = p_0 + \displaystyle\sum_{j = 1}^{n} (p_j + p_{-j}) \cos(jt) + i (p_j - p_{-j}) \sin(jt).
\end{equation}
Note that $p(e^{-it}) = p(e^{it})^{*}$, for a real polynomial $p$. For any positive integer $j$, note that $\cos(jt)$ can be written as a degree-$j$ polynomial $T_j$ acting on $\cos(t)$: this is simply the $j$-th Chebyshev polynomial of the first-kind. The identity $\sin(j t) = \sin(t) U_{j - 1}(\cos(t))$ also holds, where $U_{j}$ is the $j$-th Chebyshev polynomial of the second-kind. It follows that
\begin{equation}
    \label{eq:gamma_def}
    p(e^{i t}) = \displaystyle\sum_{j = 0}^{m} a_j T_j(\cos(t)) + i \displaystyle\sum_{j = 1}^{m} b_j \sin(t) U_{j - 1}(\cos(t)) = p_A(\cos(t)) + i \sin(t) p_B(\cos(t)),
\end{equation}
where we have let $a_j = p_j + p_{-j}$, $b_j = p_j - p_{-j}$ for $j \geq 1$, $a_0 = p_0$, and we have defined $p_A$ and $p_B$, yielded from $p$, as the degree $m$ and $m - 1$ polynomials
\begin{equation}
\label{eq:def}
    p_A(x) \coloneqq \displaystyle\sum_{j = 0}^{m} a_j T_j(x) \ \ \ \ \text{and} \ \ \ \ p_B(x) \coloneqq \displaystyle\sum_{j = 1}^{m} b_j U_{j - 1}(x).
\end{equation}
Since the Chebyshev polynomials of both kinds, up to degree-$n$, form bases for the space of all degree-$n$ polynomials, it follows that there exists a well-defined, natural map between real Laurent polynomials of degree-$n$, and pairs of real polynomials of degree $m$ and $m - 1$ respectively. We characterize this map in the following result.

\begin{definition}
    For $\mathbb{F}$ a field, let $\mathbb{F}[x_1,\dots,x_k]_n$ be the vector space of all polynomials in $x_1,\dots,x_k$ of degree less than or equal to $n$ with coefficients in $\mathbb{F}$.
\end{definition}

\begin{theorem}
    \label{thm:gamma_map}
   The map $\Gamma : \mathbb{R}[z, z^{-1}]_m \rightarrow \mathbb{R}[x]_{m} \times \mathbb{R}[x]_{m - 1}$ such that $\Gamma(p) = (p_A, p_B)$, where $p_A$ and $p_B$ are defined in Eq.~\eqref{eq:def}, is well-defined and bijective.
\end{theorem}
\begin{proof}
    Clearly, this map is well-defined, by construction. Moreover, given a pair $(P, Q) \in \mathbb{R}[x]_m \times \mathbb{R}[x]_{m - 1}$, we have $P(x) = \sum_{j = 0}^{m} a_j T_j(x)$ and $Q(x) = \sum_{j = 1}^{m} b_j U_{j - 1}(x)$. The Chebyshev polynomials are bases, so the induced $p_j$ and $q_j$ are unique. We then let $r_j = \frac{a_j + b_j}{2}$ and $r_{-j} = \frac{a_j - b_j}{2}$, for $j \geq 1$. We also let $r_0 = a_0$, and let $r(z) = \sum_{j = -m}^{m} r_j z^{j}$. Define $\Gamma^{-1}(P, Q) = r$. It is trivial to check that $\Gamma^{-1}$ is a well-defined inverse of $\Gamma$, so the map is bijective.
\end{proof}

\begin{remark}[Bijection of complex polynomials]
\label{rem:complex_qsp}
We know that $\mathbb{C}[z, z^{-1}]_m \simeq \mathbb{R}[z, z^{-1}]_m \times \mathbb{R}[z, z^{-1}]_m$, and $\mathbb{C}[x]_m \simeq \mathbb{R}[x]_m \times \mathbb{R}[x]_m$. Thus, $\Lambda \simeq \Gamma \times \Gamma$, where the identification of maps is done in the obvious way of 
\begin{equation}
p + iq \simeq (p, q) \mapsto \Gamma(p) \times \Gamma(q) = (p_A, p_B) \times (q_A, q_B) \simeq (p_A + iq_A, p_B + iq_B),
\end{equation}
so $p + iq \mapsto (p_A + iq_A, p_B + iq_B)$ acts as a bijective map from $\mathbb{C}[z, z^{-1}]_m$ to $\mathbb{C}[x]_m \times \mathbb{C}[x]_{m - 1}$, where the real and imaginary parts of the standard and Laurent polynomials are in bijective correspondence.
\end{remark}

\begin{remark}[Correspondence of QSP and $\sigma_z$-QSP]
\label{rem:z_regular_qsp}
    Using this map, it is easy to see the correspondence between standard QSP protocols and $\sigma_z$-QSP protocols induced by this map. In particular, if $\Phi$ is a QSP protocol inducing $(P, Q)$, then it follows immediately from Eq.~\eqref{eq:gamma_def}, as well as Eq.~\eqref{eq:z_1} and Eq.~\eqref{eq:z_2} and the definition of $\Gamma$ that $\Phi'[e^{i \sigma_z t}]$ is the complex Laurent polynomial $\Lambda^{-1}(P, Q)$.
\end{remark}

\noindent
Now, let us consider relevant subsets of $\mathbb{C}[x]_m \times \mathbb{C}[x]_{m - 1}$, which will correspond to all admissible pairs of QSP polynomials of particular degree.

\begin{lemma}
\label{lem:s1}
    Let $S_1 \subset \mathbb{C}[x]_m \times \mathbb{C}[x]_{m - 1}$ be the subset of polynomials $(P, Q)$ such that $|P(x)|^2 + (1 - x^2) |Q(x)|^2 \leq 1$ for all $x \in [-1, 1]$. Then $\Lambda^{-1}(S_1) \subset \mathbb{C}[z, z^{-1}]_{m}$ is precisely the subset of $\mathbb{C}[z, z^{-1}]_m$ with $|p(z)| \leq 1$ for all $z \in \mathbb{T}$.
\end{lemma}

\begin{proof}
   Using Eq.~\ref{eq:gamma_def}, and the definition of $\Lambda$, as well as the fact that $p = \Re[p] + i \Im[p]$, we note that
   \begin{align}
       |p(e^{it})|^2 &= |\Re[p](e^{it})|^2 + |\Im[p](e^{it})|^2
       \\ & = |\Re[p]_A(\cos(t))|^2 + |\sin(t)|^2 |\Re[p]_B(\cos(t))|^2 + |\Im[p]_A(\cos(t))|^2 + |\sin(t)|^2 |\Im[p]_B(\cos(t))|^2
       \\ & = |(\Re[p]_A + i \Im[p]_A)(\cos(t))|^2 + |\sin(t)|^2 |(\Re[p]_B + i \Im[p]_B)(\cos(t))|^2
       \\ & = |\Lambda(p)_1(\cos(t))|^2 + |\sin(t)|^2 |\Lambda(p)_2(\cos(t))|^2
   \end{align}
   for all $t$. Given $|p(e^{it})| \leq 1$, for some $x \in [-1, 1]$, we can pick $t$ such that $\cos(t) = x$ and $\sin(t)^2 = 1 - x^2$, so we will have $|\Lambda(p)_1(x)|^2 + (1 - x^2) |\Lambda(p)_2(x)|^2 \leq 1$. Thus, $\Lambda(p) \in S_1$, so $p \in \Lambda^{-1}(S_1)$. Given $p \in \Lambda^{-1}(S_1)$, note that $\Lambda(p) \in S_1$, so again using the above formula, $|p(e^{it})| \leq 1$.
\end{proof}

\begin{lemma}
    \label{lem:s2}
    Let $S_2 \subset \mathbb{C}[x]_m \times \mathbb{C}[x]_{m - 1}$ be the subset of polynomials $(P, Q)$ such that $P(x)$ has parity $m \ \text{mod} \ 2$ and $Q(x)$ has parity $(m - 1) \ \text{mod} \ 2$. Then $\Lambda^{-1}(S_2) \subset \mathbb{C}[z, z^{-1}]_{m}$ is precisely the subset of $\mathbb{C}[z, z^{-1}]_{m}$ with $p$ having parity $m \ \text{mod} \ 2$.
\end{lemma}

\begin{proof}
    Since the parities of $T_j(x)$ and $U_{j}(x)$ are $j \ \text{mod} \ 2$, when polynomials $P$ and $Q$ with $(P, Q) \in S_2$ are expanded in this basis according to Eq.~\eqref{eq:def}, only coefficients of index $j$ such that $j \ \text{mod} \ 2 = m \ \text{mod} \ 2$ can be non-zero. It follows immediately that only coefficients of $\Lambda^{-1}(P, Q)$ of power $j$ with $j \ \text{mod} \ 2 = m \ \text{mod} \ 2$ can be non-zero, impyling that $\Lambda^{-1}(P, Q)$ has parity $m \ \text{mod} \ 2$.

    Conversely, suppose $p \in \mathbb{C}[z, z^{-1}]$ has parity $m \ \text{mod} \ 2$, so only powers $j$ with $j \ \text{mod} \ 2 = m \ \text{mod} \ 2$ can be non-zero in $p$. Then it follows from the definition of $\Lambda$ and the parities of $T_j$ and $U_{j}$ that the polynomials of the pair $(P, Q) = \Lambda(p)$ will have parity $j \ \text{mod} \ 2$ and $(m - 1) \ \text{mod} \ 2$, respectively.
\end{proof}

\noindent Clearly, $S = S_1 \cap S_2$ is the set of all admissible QSP polynomial pairs $(P, Q)$. Since $\Lambda$ is a bijection, $\Lambda^{-1}(S) = \Lambda^{-1}(S_1 \cap S_2) = \Lambda^{-1}(S_1) \cap \Lambda^{-1}(S_2)$. In the dual $\sigma_z$-picture, it follows from this fact, along with Rem.~\ref{rem:z_regular_qsp}, that we have the next result.

\begin{theorem}
    \label{thm:z_qsp}
    There exists a length-$n$ $\sigma_z$-QSP protocol $\Phi'$ which yields a unitary $\Phi'[e^{i \sigma_z t}]$ with $\langle 0 | \Phi'[e^{i \sigma_z t}] | 0\rangle = \Re[p] \in \mathbb{R}[e^{it}, e^{-it}]$ and $\langle 0 | \Phi'[e^{i \sigma_z t}] | 1\rangle = \Im[p] \in \mathbb{R}[e^{it}, e^{-it}]$ if and only if $\deg(p) \leq n$, $p$ has parity $n \ \text{mod} \ 2$, and $|p(e^{it})| \leq 1$ for all $e^{it} \in \mathbb{T}$. Moreover, for some $p \in \mathbb{C}[e^{it}, e^{-it}]$, the protocol $\Phi$ yields the pair $(P, Q) = \Lambda(p)$ in the standard $\sigma_x$-picture.
\end{theorem}
\begin{proof}
    If $p$ satisfies the given properties, then $p \in \Lambda^{-1}(S)$ from the above results, and $\Lambda(p) = (P, Q) \in S$ will be a valid pair of QSP polynomials, so there exists a protocol $\Phi$ yielding them. It is then easy to see that $\Phi'$ yields $p$ in the $\sigma_z$-picture. Conversely, if $\Phi'$ is a $\sigma_z$-QSP protocol yielding $p$, then $\Phi$ yields polynomials $(P, Q) = \Lambda(p) \in S$, so $p \in \Lambda^{-1}(S)$, and therefore satisfies the outlined properties.
\end{proof}



\subsection{Embeddability properties of M-QSP protocols}
\label{appx:embeddability_criterion}

\noindent It is important to make note of certain useful facts about the embeddability properties of different types of M-QSP protocols.
\begin{theorem}
    \label{thm:antisymmetric}
   Given a length-$n$ antisymmetric $M$-QSP protocol $(\Phi, s)$, the unitary $\Phi[U_0, \dots, U_{t - 1}]$ is twisted embeddable for any set of twisted embeddable unitaries $U_0, \dots, U_{t - 1}$.
\end{theorem}
\begin{proof}
    When using the superoperator $(N \circ R)(\Phi)$, we also assume implicitly that we have sent $s \mapsto R(s)$ as well. In the case that $(\Phi, s)$ is antisymmetric, $(N \circ R)(\Phi) = \Phi$ and $R(s) = s$. Thus, for some set of embeddable unitaries $U_0, \dots, U_{t - 1}$,
    \begin{equation}
        \Phi[U_1, \dots, U_{t - 1}]^{\dagger} = e^{-i \phi_n \sigma_z} \displaystyle\prod_{k = 1}^{n} U_{s_{n - k + 1}}^{\dagger} e^{-i \phi_{n - k} \sigma_z}.
    \end{equation}
    Since each $U_k$ is twisted embeddable, we have $U_k = e^{i \varphi_k \sigma_z/2} e^{i \theta_k \sigma_x} e^{-i \varphi_k \sigma_z/2}$, for some $\theta_k$ and $\varphi_k$, immediately implying that $U_k^{\dagger} = \sigma_z U_k \sigma_z$ for each $k$. Moreover, because $(N \circ R)(\Phi) = \Phi$, we have $-\phi_n = \phi_0$ and $-\phi_{n - k} = \phi_k$ for each $k$. Since $R(s) = s$, we have $s_{n - k + 1} = s_k$ for each $k$. Thus, 
    \begin{align}
        e^{-i \phi_n \sigma_z} \displaystyle\prod_{k = 1}^{n} U_{s_{n - k + 1}}^{\dagger} e^{-i \phi_{n - k} \sigma_z} = e^{-i \phi_n \sigma_z} \displaystyle\prod_{k = 1}^{n} \sigma_z U_{s_{n - k + 1}} \sigma_z e^{-i \phi_{n - k} \sigma_z} &= \sigma_z e^{i \phi_{0} \sigma_z} \displaystyle\prod_{k = 1}^{n} U_{s_k} e^{i \phi_k \sigma_z} \sigma_z
        \\ & = \sigma_z \Phi[U_1, \dots, U_{t - 1}] \sigma_z,
    \end{align}
    and immediately, $\langle 0 | \Phi[U_1, \dots, U_{t - 1}] | 0\rangle = \langle 0 | \Phi[U_1, \dots, U_{t - 1}]^{\dagger} |0\rangle$, so the the top-left entry is real. Since $\Phi[U_1, \dots, U_{t - 1}]$ is $\text{SU}(2)$, this implies immediately that the resulting unitary has a real diagonal, and thus the complex phase of the off-diagonal can be factored out as a $\sigma_z$-conjugation, implying $\Phi[U_1, \dots, U_{t - 1}]$ is twisted embeddable. 
\end{proof}

\begin{lemma}[Linear error accumulation in M-QSP protocols]
\label{lem:error}
Let $(\Phi, s)$ be a length-$(n + 1)$ M-QSP protocol. Let $U_0, \dots, U_{t - 1}$ be a set of unitaries. Let $U_0', \dots, U_{t - 1}'$ be a set of unitaries such that $|| U_k - U_k'|| \leq \varepsilon$ for all $k$. Then
\begin{equation}
    || \Phi[U_0, \dots, U_{t - 1}] - \Phi[U_0', \dots, U_{t - 1}'] || \leq n\varepsilon.
\end{equation}
\end{lemma}
\begin{proof}
This follows trivially from Thm.~\ref{thm:linear_error_gadgets}.
\end{proof}
\begin{corollary}
    \label{cor:twist_emb}
    The set of length-$(n + 1)$, antisymmetric M-QSP protocols map $\varepsilon$-twisted embeddable unitaries to $n\varepsilon$-twisted embeddable unitaries.
\end{corollary}
\begin{proof}
    Let $U$ be $\varepsilon$-twisted embeddable, so $||U - U'|| \leq \varepsilon$ for twisted embeddable $U'$. From Thm.~\ref{thm:antisymmetric}, $\Phi[U']$ is twisted embeddable. From Lem.~\ref{lem:error}, $||\Phi[U] - \Phi[U']|| \leq n\varepsilon$, so $\Phi[U]$ is $n\varepsilon$-twisted embeddable.
\end{proof}

\noindent The correction protocols outlined in Thm~\ref{thm:qsp_correction}, while effective, are sometimes costly, and become much more costly in contexts where they must be called many times for a recursive nesting of gadgets. Thus, it is a problem of practical importance to consider instances in which the correction protocols are either not required, or can be implemented with fewer resources. Many gadget networks, such as those presented in the examples of Sec.~\ref{sec:examples}, make use of both single-variable and multi-variable M-QSP protocols (atomic gadgets). In fact, it is often the case that the bulk of the gadgets implemented within a network are $(1, 1)$ atomic (single-variable QSP protocols). As is remarked upon through this work, the single-variable QSP protocol is much easier to analyze than its multivariate generalization. In fact, as we will demonstrate in this section, it is possible to characterize classes of QSP protocols which will be \textit{generically} $\varepsilon$-embeddable, and thus will not require correction prior to composition. To begin, consider the following example.
\begin{example}[Embeddability of trivial QSP protocols]
\label{ex:trivial_qsp_emb}
  Let $\Phi^{(k)} = \{0, \dots, 0\}$ be the length-$k$ trivial QSP protocol, which is known the achieve the function $T_k(x)$: the $k$-th Chebyshev polynomial. Given an embeddable unitary of the form $U = e^{i \theta \sigma_x}$, $U' = \Phi^{(k)}[U] = e^{i k \theta \sigma_x}$ is \textit{not} generally embeddable, due to the fact that $k\theta$ could be in the region $[-\pi, 0]$, in which case $e^{ik\theta \sigma_x} = \sigma_z e^{-ik\theta \sigma_x} \sigma_z$, where $e^{-ik\theta \sigma_x}$ is embeddable. Thus, in general, $U'$ is twisted embeddable, but it is always the case that either $U'$ \emph{or} $\sigma_z U' \sigma_z$ is embeddable. Moreover, $U'$ can be easily made embeddable (via $\sigma_z$-conjugation) when $\cos(k \theta)$ lies between two, \textit{known} neighbouring roots of $T_k(x)$.
\end{example}
\noindent We will demonstrate that the Chebyshev polynomials are the \emph{only} family of polynomials which satisfy this condition: all other single-variable QSP protocols are twisted embeddable (Thm.~\ref{thm:antisymmetric}). We can provide a complete description of all single-variable QSP protocols which are $\varepsilon$-embeddable.

\begin{remark}[$\varepsilon$-embeddability of QSP protocols]
\label{rem:emb_crit}
    $\Phi$, a QSP protocol, will yield an $\varepsilon$-embeddable unitary if and only if $\Phi[e^{i \theta \sigma_x}]$ is an $\varepsilon$-approximation of some phase function $e^{i f(\theta) \sigma_x}$ for $f(\theta) \in [0, \pi]$, for each $\theta$. This is somewhat easier to interpret in the $\sigma_z$-QSP picture, where $\varepsilon$-embeddable QSP protocols $\Phi$ will correspond to $\sigma_z$-QSP protocols $\Phi'$ such that $\Phi'[e^{i \sigma_z t}]$ is an $\varepsilon$-approximation of $e^{i \sigma_z f(t)}$ for $f(t) \in [0, \pi]$. In this picture, since the functional object being considered is a Laurent polynomial of $e^{it}$, it is easy to see that the only case where $\Phi'[e^{i \sigma_z t}]$ can yield \emph{precisely} an operator of the form $e^{i \sigma_z f(t)}$ is when the Laurent polynomial of $\Phi'$, $p$, satisfies $\Re[p](e^{it}) = e^{i f(t)}$ and $\Im[p] = 0$. Thus,
   \begin{equation}
       \displaystyle\sum_{j = -m}^{m} p_j e^{i t j} = e^{i f(t)} \Longrightarrow |\Re[p](e^{it})|^2 = \displaystyle\sum_{k > j} 2 p_k p_j \cos((k - j) t) = 1 - \displaystyle\sum_{k} p_k^2
   \end{equation}
   for all $t$. One can verify inductively, using the linear independence of the family of functions $\cos(kt)$, that this implies all but a single coefficient in the polynomial is non-zero. In the standard QSP picture, these cases will correspond precisely to the trivial, Chebyshev polynomial protocols.
\end{remark}

\noindent Despite the fact that the unitaries which are precisely embeddable are rather limited, there does exist a rich class of approximately embeddable QSP protocols (up to possible $\sigma_z$-conjugation): namely those which approximate phase functions $e^{it} \mapsto e^{i f(t)}$ in the Laurent picture. These functions include arbitrary real roots $e^{it} \mapsto e^{i r t}$, exponentiated polynomials more generally $e^{it} \mapsto e^{i p(t)}$, or even exponentiated trigonometric functions of the form $e^{it} \mapsto e^{i \cos(t)}$, which can be approximated via the Jacobi-Anger expansion.

\begin{remark}[Half-twisted embeddability]
\label{rem:half_twisted_criteria}
The condition of half-twisted embeddability is required for implementation of the ancillla-free root protocol, and is in general harder to characterize than embeddability. Consider the output of an antisymmetric single-variable QSP protocol, which will take the form
\begin{equation}
    e^{i \varphi(x) \sigma_z / 2} e^{i \sigma_x \arccos(P(x))} e^{-i \varphi(x) \sigma_z / 2} = \begin{pmatrix} P(x) & i \sqrt{1 - P(x)^2} e^{i \varphi(x)} \\ i \sqrt{1 - P(x)^2} e^{-i \varphi(x)} & P(x) \end{pmatrix}
\end{equation}
where $Q(x) = |Q(x)|e^{i \varphi(x)} = \sqrt{\frac{1 - P(x)^2}{1 - x^2}} e^{i \varphi(x)}$. Here, $Q \in \mathbb{C}[x]$ is separated into its magnitude and phase, which are both functions of $x$. The condition that $\varphi(x) \in [-\pi/2, \pi/2]$ for all $x$ in a given range (in other words, half-twisted embeddability) is precisely equivalent to the real component of $Q(x)$ being non-negative. Thus, we are particularly interested in pairs of QSP polynomials $(P, Q)$ for which $P$ is real and $Q$ has non-negative real part. This is an opaque constraint, and therefore, we will often be forced to check whether a given QSP or M-QSP protocol achieves this condition on a case-by-case basis.
\end{remark}



\subsection{On QSP-derived and general gadgets} \label{sec:gadget_types}

\noindent We conclude this section by taking some time to discuss and disambiguate the various forms of \emph{gadgets} define in this work; some of these definitions relate closely to QSP and M-QSP circuits, and the explicit circuits analyzed in this work are all derived from such circuits. The abstract classes of gadgets we define do not all constitute strict subsets or supersets of one-another, although again the relevant sub-classes we work with most often do obey such simple containment relations.

In this work we define \emph{atomic gadgets} (Def.~\ref{def:qsp_gadget}), \emph{gadgets} (Def.~\ref{def:qsp_gadget_general}), as well as \emph{circuit gadgets} (Def.~\ref{def:circuit_gadget}). The first of these, \emph{atomic gadgets}, (Def.~\ref{def:qsp_gadget}) is a much more general definition of collections of M-QSP circuits, and only become true \emph{gadgets} under special conditions, namely when they are antisymmetric (Def.~\ref{def:m_qsp_func_prog}), or atypical (Def.~\ref{def:atypical_gadget}) and properly pinned. Consequently \emph{atomic gadgets} are \emph{not} a subset of \emph{gadgets}, though the sub-class of atomic gadgets we consider in this work are all also gadgets. Finally, \emph{circuit gadgets} (Def.~\ref{def:circuit_gadget}) are indeed all gadgets, and merely restrict the abstract definition of gadgets to those realized from discrete circuits over a finite collection of qubits. Summarizing these relations diagrammatically, we refer to Fig.~\ref{fig:gadget_def_relations}.

\begin{figure}[htpb]
	\centering
	\includegraphics[width=0.75\textwidth]{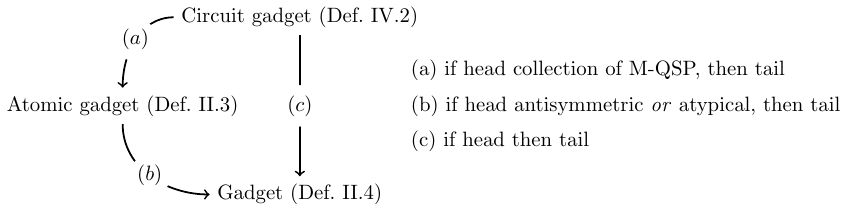}
	\caption{A diagrammatic summary of the relation between \emph{atomic gadgets} (Def.~\ref{def:qsp_gadget}), \emph{gadgets} (Def.~\ref{def:qsp_gadget_general}), and \emph{circuit gadgets} (Def.~\ref{def:circuit_gadget}). The labels on arrows indicate that the relation should be read as follows: if the object is the \emph{head} of the arrow, and satisfies the given condition(s), then the object is also the \emph{tail} of the arrow.}
	\label{fig:gadget_def_relations}
\end{figure}

In addition to antisymmetric atomic gadgets, which are the primary focus of this work, the second subclass of atomic gadgets which are also standard gadgets are \emph{atypical gadgets} (Def.~\ref{def:atypical_gadget}), which we define below. We discuss these only in the $(1, 1)$ case, but in general such atypical gadgets can exist for any number of input legs; the general condition follows from finding roots of the imaginary part of the $P$ achieved by a general atomic gadget, and pinning (Def.~\ref{def:aux_gadget_operations}) to one of those roots. As we will not need such gadgets for our examples, and their theory is considerably less well-understood given its connection to M-QSP \cite{rossi_m_qsp_22}, we leave discussion of them to future work. Nevertheless, their existence leave a variety of interesting questions open in the construction of highly expressive atomic gadgets.

\begin{definition}[Atypical atomic gadget] \label{def:atypical_gadget}
	An atomic $(a, b)$ gadget $\mathfrak{G}$ is \emph{atypical} if there exists a pinning $\mathfrak{G}^\prime$ of $\mathfrak{G}$ such that all output legs $\mathfrak{G}^\prime_k, k \in [b]$ are twisted embeddable, but this \emph{does not hold for all pinnings} $\mathfrak{G}^\prime$. Note that $\mathfrak{G}$ need \emph{not} be antisymmetric. For brevity, we call \emph{both} $\mathfrak{G}$ and $\mathfrak{G}^\prime$ atypical gadgets, and almost always work directly with $\mathfrak{G}^\prime$.
\end{definition}

\begin{example}[Atypical atomic $(2, 1)$ gadget] \label{ex:atypical_gadget}
	Consider the prescription given by \cite{gslw_qsvt_19} (reproduced in Thm.~\ref{thm:qsp_real}) for finding a QSP protocol with phases $\Phi$ for achieving the pair of functions $P, Q$ for arbitrary choice of $\tilde{P}, \tilde{Q}$ such that $|\tilde{P}|^2 + (1 - x^2)|\tilde{Q}|^2 \leq 1$, with $\Re[P] = \tilde{P}$ and $\Re[Q] = \tilde{Q}$. Moreover, choose $\tilde{Q} = 0$. Then consider the (non-antisymmetric) atomic $(2, 1)$ gadget (Def.~\ref{def:atypical_gadget}) defined by the following $(\Xi, S)$:
		\begin{align}
			\Xi &= \{\{0, \pi/4, 0, -\pi/4\} \cup \Phi \cup \{\pi/4, 0, \pi/4, 0\}\},\\
			S &= \{\{1, 1, 0, 0, \dots, 0, 0, 1, 1\}\},
		\end{align}
	where $\Phi$ are the (symmetric \cite{wdl_sym_qsp_22}) phases achieving $\tilde{P}$ and $\tilde{Q} = 0$, and where $S$ contains $0$ where not otherwise indicated. It is easy to determine that this gadget does \emph{not} output a twisted-embeddable unitary for general choice of $U_1$; in fact, this gadget achieves the following function, where $U_1 = e^{i\theta\sigma_z}$:
		\begin{equation}
			\langle 0 |\mathfrak{G}[U_0, U_1]|0\rangle =
			i (\tilde{P} + \tilde{P}^\prime)\cos{(2\theta)} 
			- \tilde{P} \sin^2{(2\theta)},
		\end{equation}
	where $\tilde{P}^\prime = \Im[P]$ as achieved by the QSP protocol with phases $\Phi$. For $\theta = \pi/4 + k\pi/2, k \in \mathbb{Z}$, however, this gadget achieves $-\tilde{P}$, and in fact these are the only $\theta$ at which a real function is achieved for general $\tilde{P}$. Thus pinning (Def.~\ref{def:aux_gadget_operations}) $U_1 = e^{i(\pi/4)\sigma_z}$ produces an atypical atomic $(1, 1)$ gadget achieving an arbitrary real, bounded, definite parity $\tilde{P}$.
\end{example}


\section{Extraction of an unknown  \texorpdfstring{$\sigma_z$}{sigma-z}-conjugation} \label{appx:extraction}

\noindent This section is dedicated to the problem of mapping a twisted embeddable oracle,
\begin{equation}
U = e^{i \varphi \sigma_z/2} e^{i \theta \sigma_x} e^{-i \varphi \sigma_z/2} = e^{i \varphi \sigma_z/2} e^{i \arccos(x) \sigma_x} e^{-i \varphi \sigma_z/2} = e^{i \varphi \sigma_z/2} W(x) e^{-i \varphi \sigma_z/2},
\end{equation}
to a power of the corresponding, unknown $\sigma_z$-rotation (in particular, $e^{-i \varphi \sigma_z}$). Recall that QSP protocols are invariant (in a sense) under $\sigma_z$-conjugation. That is to say, given a QSP protocol $\Phi$ and some arbitrary $\sigma_z$-rotation $e^{i \varphi \sigma_z/2}$,
\begin{equation}
    \Phi[e^{i \varphi \sigma_z/2} U e^{-i \varphi \sigma_z/2}] = e^{i \varphi \sigma_z/2} \Phi[U] e^{-i \varphi \sigma_z/2}
\end{equation}
for any unitary $U$. Our general strategy for extracting the desired $\sigma_z$-rotation is to utilize this property of QSP to map $W(x)$ to a fixed, known unitary over all $x$, so that it can be cancelled, leaving behind the desired $\sigma_z$-rotation.


\subsection{Constructing the extraction QSP polynomial}

\noindent
We wish to construct a polynomial $Q(x)$ which is close to $1$ for $x \in \mathcal{D} = [-1 + \delta, 1 - \delta]$. The associated QSP protocol will drive an arbitrary gate of the form $W(x)$ towards $i \sigma_x$ when $x \in \mathcal{D}$. Let $A : (-1, 1) \rightarrow \mathbb{R}$ be the function defined by $A(x) = \frac{1}{\sqrt{1 - x^2}}$. Let $A_n(x)$ denote the $(n - 1)$-th order Taylor approximation of $A(x)$. From the generalized binomial formula,
\begin{equation}
    A_n(x) = \sum_{k = 0}^{n - 1} \binom{-\frac{1}{2}}{k} (-1)^{k} x^{2k}.
\end{equation}
From Lem.~\ref{lem:binom_half} of Appx.~\ref{sec:poly_bnds}, it follows immediately that the coefficients of $A_n(x)$ are weakly decreasing in magnitude, with positive sign. Thus, the Taylor series has a radius of convergence that at least contains the interval $(-1, 1)$. Cor.~\ref{cor:weak_coeffs} now holds with $\lambda = 2$. Note that 

\begin{equation}
M = |f(1 - \delta)| = \frac{1}{\sqrt{1 - (1 - \delta)^2}} = \frac{1}{\sqrt{2\delta - \delta^2}} \leq \frac{1}{\sqrt{\delta}}.
\end{equation}
Thus, setting
\begin{equation}
    n = \frac{1}{2 \delta} \log \left( \frac{1}{\varepsilon \sqrt{\delta}} \right) \geq \frac{1}{2 \delta} \log \left( \frac{M}{\varepsilon} \right),
\end{equation}
it immediately follows that $A_n(x)$ is an $\varepsilon$-approximation of $A(x)$ for all $x \in \mathcal{D}$. All that remains is to show the existence of a QSP protocol yielding $A_n(x)$ as the off-diagonal polynomial $Q$. Clearly, $A_n(x)$ is even, so both criterion outlined in Thm.~\ref{thm:p_exist} must be demonstrated to hold. To begin, note (again from Lem.~\ref{lem:binom_half}) that because each term in the sum of $A_n(x)$ is positive for any $x \in [-1, 1]$, we have $A_{n + 1}(x) \geq A_n(x)$ for all $n$ and all $x \in [-1, 1]$. Since $\lim_{n \to \infty} A_n(x) = A(x)$, it follows that $A_n(x) \leq A(x)$ for all $n$, which means that for $x \in [-1, 1]$,
\begin{equation}
   \sqrt{1 - x^2} |A_n(x)| = \sqrt{1 - x^2} A_n(x) \leq \sqrt{1 - x^2} \frac{1}{\sqrt{1 - x^2}} = 1, 
\end{equation}
so the first condition of Thm.~\ref{thm:p_exist} holds. As for the second condition, note that
\begin{equation}
    A_n(ix) = \sum_{k = 0}^{n - 1} \binom{-\frac{1}{2}}{k} (-1)^{k} (ix)^{2k} = \sum_{k = 0}^{n - 1} \binom{-\frac{1}{2}}{k} x^{2k}.
\end{equation}
Let $a_n(x) = \sum_{k = 0}^{n - 1} \binom{-\frac{1}{2}}{k} x^{k}$: the $(n - 1)$-th order Taylor series of $a(x) = \frac{1}{\sqrt{1 + x}}$. Since $a(x) = 2 \frac{d}{dx} \sqrt{1 + x}$, it follows immediately from Lem.~\ref{lem:sqrt} that $a^{(n)}(x) \neq 0$ for all $x \in (0, \infty)$. Moreover, from Lem.~\ref{lem:binom_half}, the $n$-th coefficient in the Taylor series for $a(x)$ has sign $(-1)^{n}$. Thus, from Thm.~\ref{thm:taylor}, when $n$ is odd, we have $a_n(x) > a(x)$ for all $x \in (0, \infty)$. It immediately follows that $A_n(ix) = a_n(x^2) > a(x^2) = \frac{1}{\sqrt{1 + x^2}}$ for all $x \in \mathbb{R} - \{0\}$, with equality at $x = 0$. Therefore,
\begin{equation}
    (1 + x^2) A_n^{*}(ix) A_n(x) = (1 + x^2) A_n(ix)^2 \geq (1 + x^2) \frac{1}{1 + x^2} = 1
\end{equation}
for all $x \in \mathbb{R}$, when $n$ is odd. The second condition of Thm.~\ref{thm:p_exist} is satisfied. This leads to the following result.

\begin{lemma}[Existence of an extraction QSP polynomial]
    \label{lem:ex_poly}
    Given $0 < \varepsilon, \delta \leq 1$, there exists a QSP protocol $\Phi$ of length $\delta^{-1} \log(\varepsilon^{-1} \delta^{-1/2})$ such that $Q$, the off-diagonal QSP polynomial, satisfies $|Q(x) - A(x)| \leq \varepsilon$ for all $x \in [-1 + \delta, 1 - \delta]$. Moreover, the coefficients of $Q(x)$ can be specified analytically and are highly efficient to compute.
\end{lemma}


\subsection{The extraction protocol}

\noindent Using a QSP protocol associated with the polynomial defined in the previous section, it is now possible extract the desired $\sigma_z$-rotation. First, a brief lemma.

\begin{lemma}
    \label{lem:su2_bnd}
    Let $U$ be an element of $\text{SU}(2)$. Let $a = \langle 0 | U | 0\rangle$, $b = \langle 0 | U | 1 \rangle$. If $1 - |a| \leq \varepsilon$, then
    \begin{equation}
        ||U - \text{diag}(U)|| = \left| \left| \begin{pmatrix} a & b \\ -b^{*} & a^{*} \end{pmatrix} - \begin{pmatrix} a & 0 \\ 0 & a^{*} \end{pmatrix} \right| \right| \leq \sqrt{2\varepsilon}
    \end{equation}
    where $|| \cdot ||$ is the matrix spectral norm. Similarly, if $1 - |b| \leq \varepsilon$, then
    \begin{equation}
        \left| \left| \begin{pmatrix} a & b \\ -b^{*} & a^{*} \end{pmatrix} - \begin{pmatrix} 0 & b \\ -b^{*} & 0 \end{pmatrix} \right| \right| \leq \sqrt{2\varepsilon}
    \end{equation}
\end{lemma}
\begin{proof}
   Note that for $U \in \text{SU}(2)$, we have $\det(U) = |a|^2 + |b|^2 = 1$. Thus, $|b|^2 = 1 - |a|^2 = (1 + |a|) (1 - |a|) \leq 2(1 - |a|)$, implying that $|b| \leq \sqrt{2 (1 - |a|)} \leq \sqrt{2\varepsilon}$. Thus,
   \begin{equation}
       \left| \left| \begin{pmatrix} a & b \\ -b^{*} & a^{*} \end{pmatrix} - \begin{pmatrix} a & 0 \\ 0 & a^{*} \end{pmatrix} \right| \right| = \left| \left| \begin{pmatrix} 0 & b \\ -b^{*} & 0 \end{pmatrix} \right| \right| = |b| \leq \sqrt{2\varepsilon}
   \end{equation}
   The proof for the case of $|b|$ being close to $1$ is essentially identical.
\end{proof}
\noindent
From here, we can prove the main theorem of this section.

\begin{theorem}[Existence of an extraction superoperator]
\label{thm:extraction_superoperator}
Given a twisted embeddable unitary oracle of the form $U = e^{i \varphi \sigma_z/2} e^{i \arccos(x) \sigma_x} e^{-i \varphi \sigma_z/2} =  e^{i \varphi \sigma_z/2} W(x) e^{-i \varphi \sigma_z/2}$ and $0 < \varepsilon, \delta \leq 1$, there exists a single-qubit circuit superoperator $\mathcal{E}$ calling $U$ and single-qubit $\sigma_z$-rotations $\delta^{-1} \log(8 \varepsilon^{-2} \delta^{-1/2}) = \mathcal{O}(\delta^{-1} \log(\varepsilon^{-2} \delta^{-1/2}))$ times each, such that when $x \in [-1 + \delta, 1 - \delta]$, $|| \mathcal{E}[U] - e^{-i \varphi \sigma_z} || \leq \varepsilon$.
\end{theorem}

\begin{proof}
    Define the superoperator $\mathcal{E}$ to be
    \begin{equation}
        \mathcal{E}[U] = -i \sigma_x \Phi[U]
    \end{equation}
    where $\Phi$ is a set of $\delta^{-1} \log( 8 \varepsilon^{-2} \delta^{-1/2})$ QSP phase angles which induces the real polynomial $Q(x)$ such that $\left|Q(x) - \frac{1}{\sqrt{1 - x^2}} \right| \leq \frac{\varepsilon^2}{8}$ for all $x \in [-1 + \delta, 1 - \delta]$ (existence of $\Phi$ is guaranteed from Lem.~\ref{lem:ex_poly}). It then follows that
    \begin{equation}
        | \sqrt{1 - x^2} Q(x) - 1 | = \sqrt{1 - x^2} \left| Q(x) - \frac{1}{\sqrt{1 - x^2}} \right| \leq \frac{\varepsilon^2}{8}
    \end{equation}
    for all $x \in [-1 + \delta, 1 - \delta]$. We then note that
    \begin{align}
        \left| \left| \mathcal{E}[U] - e^{- i\varphi \sigma_z} \right| \right| &= \left| \left| -i \sigma_x \Phi[U] -  e^{-i \varphi \sigma_z} \right| \right| \nonumber \\ & = \left| \left| -i \sigma_x e^{i \varphi \sigma_z/2}  \Phi[W(x)] e^{-i \varphi \sigma_z/2} - e^{-i \varphi \sigma_z} \right| \right| \nonumber \\ &= \left| \left| -i e^{-i \varphi \sigma_z/2}  \sigma_x \Phi[W(x)] e^{-i \varphi \sigma_z/2} - e^{-i \varphi \sigma_z} \right| \right| \nonumber
        \\ & = \left|\left| e^{-i\varphi\sigma_z/2} \left( -i \sigma_x \Phi[W(x)] - \mathbb{I} \right) e^{-i\varphi\sigma_z/2} \right|\right| \nonumber
        \\ & = || -i \sigma_x \Phi[W(x)] - \mathbb{I} || \nonumber
        \\ & = \left| \left| \begin{pmatrix} \sqrt{1 - x^2} Q(x) & -i P^{*}(x) \\ -i P(x) & \sqrt{1 - x^2} Q(x) \end{pmatrix} - \begin{pmatrix} 1 & 0 \\ 0 & 1 \end{pmatrix} \right| \right|
    \end{align}
    Then, from Lem.~\ref{lem:su2_bnd} and the triangle inequality,
    \begin{multline}
        \left| \left| \begin{pmatrix} \sqrt{1 - x^2} Q(x) & -i P^{*}(x) \\ -i P(x) & \sqrt{1 - x^2} Q(x) \end{pmatrix} - \begin{pmatrix} 1 & 0 \\ 0 & 1 \end{pmatrix} \right| \right| \leq  \left| \left| \begin{pmatrix} \sqrt{1 - x^2} Q(x) & 0 \\ 0 & \sqrt{1 - x^2} Q(x) \end{pmatrix} - \begin{pmatrix} 1 & 0 \\ 0 & 1 \end{pmatrix} \right| \right| \\ + \left| \left| \begin{pmatrix} \sqrt{1 - x^2} Q(x) & -i P^{*}(x) \\ -i P(x) & \sqrt{1 - x^2} Q(x) \end{pmatrix} - \begin{pmatrix} \sqrt{1 - x^2}Q(x) & 0 \\ 0 & \sqrt{1 - x^2}Q(x) \end{pmatrix} \right| \right|
        \\ \leq |\sqrt{1 - x^2}Q(x) - 1| + \frac{\varepsilon}{2} \leq \frac{\varepsilon^2}{8} + \frac{\varepsilon}{2} < \varepsilon
    \end{multline}
    for all $x \in [-1 + \delta, 1 - \delta]$. This completes the proof.
\end{proof}


\section{Roots of unknown  \texorpdfstring{$\sigma_z$}{sigma-z}-rotations}
\label{appx:roots}

\noindent
In the following sections, we discuss various techniques for performing oblivious roots of $\sigma_z$-rotations, with assumptions of various different constraints and access models.


\subsection{Constructing a controlled  \texorpdfstring{$\sigma_z$}{sigma-z}-rotation}

\noindent We begin by describing a technique which allows for the construction of a controlled-$\sigma_z$ rotation, given oracle access to a $\sigma_z$-rotation. To begin, we provide a similar construction to Thm.~\ref{thm:extraction_superoperator}, which can take twisted embeddable unitaries to the identity.

\begin{lemma}[Existence of a nullification superoperator]
Given a twisted embeddable unitary oracle of the form $U = e^{i \varphi \sigma_z/2} e^{i \arccos(x) \sigma_x} e^{-i \varphi \sigma_z/2} =  e^{i \varphi \sigma_z/2} W(x) e^{-i \varphi \sigma_z/2}$ and $0 < \varepsilon \leq 1/2$, $0 < \delta \leq 1$, there exists a single-qubit circuit superoperator $\mathcal{E}'$ calling $U$ and single-qubit $\sigma_z$-rotations $\mathcal{O}(\delta^{-1} \log(\varepsilon^{-1}))$ times, such that when $x \in [\delta, 1 - \delta]$, $|| \mathcal{E}[U] - \mathbb{I} || \leq \varepsilon$.
\end{lemma}

\begin{proof}
Define the superoperator $\mathcal{E}'$ to be $\mathcal{E}'[U] = \Phi'[U]$, where $\Phi'$ is the set of $\mathcal{O}(\delta^{-1} \log(\varepsilon^{-1}))$ phase angles achieving a $(\varepsilon^2/16)$-approximating polynomial of the step function $\text{sign}(x)$ on the domain $[-1, -\delta] \cup[\delta, 1]$ as the real part of the QSP polynomial $P(x)$ (Lem.~\ref{lem:step}). Since such a polynomial approximation is an odd parity function, the sequence $\Phi'$ will be of even length. For $x \in [\delta, 1]$, we have
\begin{align}
    ||\mathcal{E}'[U] - \mathbb{I}|| &= ||\Phi'[U] - \mathbb{I}|| = ||e^{i \varphi \sigma_z/2} \Phi'[W(x)] e^{-i \varphi \sigma_z/2} - \mathbb{I}|| = ||\Phi'[W(x)] - \mathbb{I}||
    \\ & = \left|\left| \begin{pmatrix} P(x) & i \sqrt{1 - x^2} Q(x) \\ i \sqrt{1 - x^2} Q^{*}(x) & P^{*}(x) \end{pmatrix} - \begin{pmatrix} 1 & 0 \\ 0 & 1 \end{pmatrix} \right|\right|
\end{align}
Note that $|\text{Re}[P](x) - \text{sign}(x)| = |\text{Re}[P](x) - 1| \leq \varepsilon^2/16$. This implies that
\begin{equation}
    1 - |P(x)| \leq 1 - |\text{Re}[P](x)| = |1 - \text{Re}[P](x)| \leq \frac{\varepsilon^2}{16}.
\end{equation}
Then, from Lem.~\ref{lem:su2_bnd} and the triangle inequality,
 \begin{multline}
        \left| \left| \begin{pmatrix} P(x) & i \sqrt{1 - x^2} Q(x) \\ i \sqrt{1 - x^2} Q^{*}(x) & P^{*}(x) \end{pmatrix} - \begin{pmatrix} 1 & 0 \\ 0 & 1 \end{pmatrix} \right| \right| \leq  \left| \left| \begin{pmatrix} P(x) & 0 \\ 0 & P^{*}(x) \end{pmatrix} - \begin{pmatrix} 1 & 0 \\ 0 & 1 \end{pmatrix} \right| \right| \\ + \left| \left| \begin{pmatrix} P(x) & i \sqrt{1 - x^2} Q(x) \\ i \sqrt{1 - x^2} Q^{*}(x) & P^{*}(x) \end{pmatrix} - \begin{pmatrix} P(x) & 0 \\ 0 & P^{*}(x) \end{pmatrix} \right| \right|
        \\ \leq |P(x) - 1| + \frac{\varepsilon}{2} = \frac{\varepsilon}{2} + \sqrt{|1 - \text{Re}[P](x)|^2 + |\text{Im}[P](x)|^2}.
    \end{multline}
    Note that $|\text{Im}[P](x)|^2 \leq 1 - |\text{Re}[P](x)|^2 \leq 2 (1 - |\text{Re}[P](x)|)$. Therefore,
    \begin{equation}
        \frac{\varepsilon}{2} + \frac{\varepsilon}{2} + \sqrt{|1 - \text{Re}[P]|^2 + |\text{Im}[P](x)|^2} \leq \sqrt{3(1 - |\text{Re}[P](x)|)} \leq \frac{\varepsilon}{2} + \frac{\varepsilon}{2} = \varepsilon.
    \end{equation}
    This completes the proof.
\end{proof}

\noindent From here, we construct a \emph{controlled} variant of the desired $\sigma_z$-rotation oracle.

\begin{theorem}[Construction of a controlled-$\sigma_z$] \label{thm:controlled_z_routine}
    Let $U$ be a twisted embeddable unitary such that $U = e^{i \varphi \sigma_z/2} e^{i \arccos(x) \sigma_x} e^{-i \varphi \sigma_z/2}$ with $x \in [\delta, 1 - \delta]$. Then there exists a protocol using two ancilla qubits and $\zeta \in \widetilde{\mathcal{O}}(\delta^{-1} \log(\varepsilon^{-1}))$ queries to $U$, as well as CSWAP gates, which achieves a $\varepsilon$-approximation of a controlled-$e^{-i \varphi \sigma_z}$ gate.
\end{theorem}
\begin{proof}
    Define
    \begin{equation}
        V = \text{CSWAP} (\mathbb{I}_{c} \otimes \mathbb{I}_{t_1} \otimes [U]_{t_2}) \text{CSWAP},
    \end{equation}
    where we have labeled the control qubit with $c$ and the target qubits with $t_1$ and $t_2$. We then define the unitary $V'$ as
    \begin{equation}
        V' = \left( \left[ -i \sigma_x e^{i \phi_0 \sigma_z} \right]_{t_1} \left[ e^{i \phi'_0 \sigma_z} \right]_{t_2} \displaystyle\prod_{k = 1}^{\zeta-1} V \left[ e^{i \phi_{k} \sigma_z} \right]_{t_1} \left[ e^{i \phi'_{k} \sigma_z} \right]_{t_2}\right)
    \end{equation}
    where $\Phi = (\phi_0, \dots, \phi_{\zeta-1})$ and $\Phi' = (\phi'_0, \dots, \phi'_{\zeta-1})$ are the length-$\zeta$ QSP sequences of protocols $\mathcal{E}$ and $\mathcal{E}'$ respectively. Note that in both cases, the QSP sequences are of even length. In the case that the control qubit is in state $|0\rangle$, we have $V|0\rangle_c |\psi\rangle_{t_1} |\xi \rangle_{t_2} = |0\rangle_c|\psi\rangle_{t_1} U |\xi\rangle_{t_2}$. In addition, $V|1\rangle_c |\psi\rangle_{t_1} |\xi\rangle_{t_2} = |1\rangle_c U |\psi\rangle_{t_1} |\xi\rangle_{t_2}$. It follows that
    \begin{align}
        V' |0\rangle |\psi\rangle_{t_1} |\xi\rangle_{t_2} &= -i |0\rangle \sigma_x e^{i \phi \sigma_z} |\psi\rangle_{t_1} \mathcal{E}'[U] |\xi\rangle_{t_2},
        \\
        V' |1\rangle |\psi\rangle_{t_1} |\xi\rangle_{t_2} &= |1\rangle \mathcal{E}[U] |\psi\rangle_{t_1} e^{i \phi' \sigma_z} |\xi\rangle_{t_2}.
    \end{align}
    where $\phi = \phi_0 + \cdots + \phi_{\zeta-1}$ and $\phi' = \phi'_0 + \cdots + \phi'_{\zeta-1}$. Finally, define $V'' = W V'$, where 
    \begin{equation}
    W = |0\rangle\langle 0|_c \otimes [i e^{-i \phi \sigma_z} \sigma_x]_{t_1} + |1\rangle\langle 1|_c \otimes [e^{-i \phi' \sigma_z}]_{t_2}.
    \end{equation}
    Therefore,
    \begin{align}
        V'' |0\rangle |\psi\rangle_{t_1} |\xi\rangle_{t_2} &= |0\rangle |\psi\rangle_{t_1} \mathcal{E}'[U] |\xi\rangle_{t_2},
        \\
        V'' |1\rangle |\psi\rangle_{t_1} |\xi\rangle_{t_2} &= |1\rangle \mathcal{E}[U] |\psi\rangle_{t_1} |\xi\rangle_{t_2}.
    \end{align}
    Our claim is that the gate $V''$ will be the desired approximate controlled-$\sigma_z$ rotation. Let $\mathcal{C}$ be the unitary superoperator sending some unitary $W$ to its controlled counterpart, $\mathcal{C}[W]$. We take $\mathcal{C}[e^{-i \varphi \sigma_z}]$ as having $c$ as its control qubit, and $t_1$ as its target. Note that
    \begin{equation}
        || \mathcal{C}[e^{-i \varphi \sigma_z}] |0\rangle |\psi\rangle_{t_1} |\xi\rangle_{t_2} - V'' |0\rangle |\psi\rangle_{t_1} |\xi\rangle_{t_2} || = || |\xi\rangle_{t_2} - \mathcal{E}'[U] |\xi\rangle_{t_2}|| \leq || \mathcal{E}'[U] - \mathbb{I}|| \leq \varepsilon.
    \end{equation}
    In addition,
    \begin{equation}
        || \mathcal{C}[e^{-i \varphi \sigma_z}] |1\rangle |\psi\rangle_{t_1} |\xi\rangle_{t_2} - V'' |1\rangle |\psi\rangle_{t_1} |\xi\rangle_{t_2} || = || e^{-i \varphi \sigma_z} |\psi\rangle_{t_1} - \mathcal{E}[U] |\psi\rangle_{t_1}|| \leq || \mathcal{E}[U] - e^{-i \varphi \sigma_z}|| \leq \varepsilon.
    \end{equation}
    since we have bounded the norms on a basis, it follows that the matrix spectral norm $||\mathcal{C}[e^{-i\varphi \sigma_z}] - V''||$ is too upper-bounded by $\varepsilon$, and we have the desired controlled approximate $\sigma_z$-rotation.
\end{proof}

\subsection{Taking roots with ancilla qubits}
 \label{appx:root_with_ancilla}

\noindent Now, that we have outlined techniques for realizing a controlled-$\sigma_z$ rotation, let $U$ be an $N$-dimensional unitary, which can be written in its eigenbasis as
\begin{equation}
    U = \displaystyle\sum_{j = 1}^{N} e^{i \lambda_i} |\lambda_i\rangle \langle \lambda_i |.
\end{equation}
Let $|\lambda_j\rangle$ be a particular eigenstate of $U$, and let $V_j$ be a unitary such that $V_j|0\rangle = |\lambda_j \rangle$. In this case,
\begin{align}
   (X \otimes V_j^{\dagger}) \mathcal{C}[U] (X \otimes V_j) &= |1\rangle \langle 1| \otimes \mathbb{I} + |0\rangle \langle 0| \otimes V_j^{\dagger} U V_j \\ &= \displaystyle\sum_{j = 1}^{N} |1\rangle \langle 1| \otimes V_j^{\dagger} |\lambda_i\rangle \langle \lambda_i| V_j + e^{i \lambda_i} |0\rangle \langle 0| \otimes V_j^{\dagger} |\lambda_i\rangle \langle \lambda_i| V_j 
   \\ &= \displaystyle\sum_{j = 1}^{N} e^{i\lambda_i/2} \begin{pmatrix} e^{i \lambda_i /2} & 0 \\ 0 & e^{-i \lambda_i/2} \end{pmatrix} \otimes V_j^{\dagger} |\lambda_i\rangle \langle \lambda_i| V_j.
\end{align}
It follows immediately that
\begin{align}
    (X \otimes V_j^{\dagger}) \mathcal{C}[U] (X \otimes V_j) |\psi\rangle |0\rangle &= \left[ \displaystyle\sum_{j = 1}^{N} e^{i\lambda_i/2} \begin{pmatrix} e^{i \lambda_i /2} & 0 \\ 0 & e^{-i \lambda_i/2} \end{pmatrix} \otimes V_j^{\dagger} |\lambda_i\rangle \langle \lambda_i| V_j  \right] |\psi\rangle |0\rangle
    \\ &= e^{i \lambda_j/2} e^{i \sigma_z \lambda_j / 2} |\psi\rangle V_j^{\dagger} |\lambda_j\rangle = e^{i \lambda_j/2} e^{i \sigma_z \lambda_j / 2} |\psi\rangle |0\rangle.
\end{align}
Therefore, up to an overall phase, the protocol of conjugating controlled-$U$ with $X \otimes V_j$ yields a circuit which applies the principal root of $e^{i \sigma_z \lambda_i}$ to any $|\psi\rangle$. Repetition of this protocol, with the newly-constructed $\sigma_z$-rotation in the place of $U$, can yield any $2^n$-th root.

\begin{remark}[The power of controlled oracle access]
Having access to an ancilla qubit and a controlled oracle yields a strictly more powerful access model than in the single-qubit setting. Via the mechanism presented above, it is sensible to argue that in the case when the eigenstates of a unitary are known, this power is closely related to the ability to easily take roots of the unitary eigenvalues. It is precisely this access model which allows results like those of \cite{wang2022quantum, yu2022power, motlagh2023generalized} to drop the parity constraints associated with QSP protocols. Given access to half-powers of the unitary eigenvalues, $e^{it/2}$, polynomials with respect to $e^{it/2}$ of even parity are precisely those in the variable $e^{it}$ with no definite parity. As will be demonstrated in App.~\ref{appx:root_without_ancilla}, it is possible to access the roots of unitary eigenvalues up to a desired precision without controlled access, \emph{but this is only possible when the eigenvalues lie in a known half of the complex unit circle}. Therefore, at an even finer-grained level, the power of the controlled access model in this context can be solely related to determination of which half of the unit circle a particular phase lies within.
\end{remark}

\begin{theorem}
    \label{thm:sqrt}
    Given a unitary $U$ of dimension $N$, with eigenvalues $e^{i \lambda_i}$ and eigenvectors $|\lambda_i\rangle$, where $|\lambda_j\rangle$ for some $j$ is known and preparable on a quantum computer via unitary $V_j$, and $n \in \mathbb{Z}^{+}$, there exists a quantum circuit using $n(N + 1) = \mathcal{O}(nN)$ ancilla qubits, a single application of $U$, two applications of $V_j$ and $V_j^{\dagger}$, and $2n = \mathcal{O}(n)$ controlled-SWAP gates which deterministically sends the state $|0\rangle |\psi\rangle$ to $|0\rangle e^{i \sigma_z \lambda_j / 2^n} |\psi\rangle$, for any $|\psi\rangle$, up to a global phase.
\end{theorem}
\noindent
This theorem is an immediate consequence of the procedure described above. The next result follows as a corollary.
\begin{corollary}
    \label{cor:ancilla_sqrt}
    Given access to a single-qubit unitary $U$, such that $|| U - e^{i \sigma_z \varphi} || \leq \varepsilon$ for $\varphi \in [-\pi, \pi]$, there exists a quantum circuit utilizing two ancilla qubits and a single call to $U$ which yields a gate $\widetilde{U}$ such that $||\widetilde{U}|\psi\rangle - e^{i \varphi/2} e^{i \sigma_z \varphi / 2} |\psi\rangle|| \leq \varepsilon$ for every $|\psi\rangle$, with success probability at least $1 - 2\varepsilon$
\end{corollary}
\begin{proof}
   We use the protocol of Thm.~\ref{thm:sqrt} in the case of $n = 1$. In this case, eigenvalue $e^{i \varphi}$ of $e^{i \varphi \sigma_z}$ corresponds to eigenvector $|0\rangle$, so $V_j = \mathbb{I}$. Note that $\mathcal{C}[e^{i \sigma_z \varphi}] |0\rangle |0\rangle |\psi\rangle = |0\rangle|0\rangle e^{i \varphi/2} e^{i \sigma_z \varphi/2} |\psi\rangle$. We then have
   \begin{align}
       \left| \left| \mathcal{C}[U] |0\rangle|0\rangle|\psi\rangle - |0\rangle |0\rangle e^{i \varphi/2} e^{i \sigma_z \varphi / 2} |\psi\rangle \right|\right| &= \left| \mathcal{C}[U] |0\rangle|0\rangle|\psi\rangle - \mathcal{C}[e^{i \varphi \sigma_z}] |0\rangle |0\rangle |\psi\rangle \right| \nonumber
       \\ & \leq \left| \left| \mathcal{C}[U] - \mathcal{C}[e^{i \varphi \sigma_z}] \right| \right| \nonumber \\
       &= \left| \left| \text{CSWAP} (\mathbb{I} \otimes \mathbb{I} \otimes U - \mathbb{I} \otimes \mathbb{I} \otimes e^{i \varphi \sigma_z}) \text{CSWAP} \right| \right|\nonumber  \\ & = || U - e^{i \sigma_z \varphi} || \leq \varepsilon,
   \end{align}
   This implies that $\mathcal{C}[U] |0\rangle |0\rangle |\psi\rangle = |0\rangle |0\rangle e^{i\varphi/2} e^{i \sigma_x \varphi/2} |\psi\rangle + |0\rangle |0\rangle |\psi'\rangle + |\psi''\rangle$, where $\Pi_0 |\psi''\rangle = 0$ (the projection onto $|0\rangle |0\rangle$) and $|0\rangle |0\rangle |\psi'\rangle + |\psi''\rangle$ has norm bounded by $\varepsilon$. Measuring the ancilla qubits in the computational basis yields (up to normalization) the state $e^{i \varphi/2} e^{i \sigma_x \varphi/2} |\psi\rangle + |\psi'\rangle$ where $|| |\psi'\rangle || \leq \varepsilon$, with probability $|| e^{i \varphi/2} e^{i \sigma_x \varphi/2} |\psi\rangle + |\psi'\rangle||^2 \geq 1 + \varepsilon^2 - 2\Re[ \langle \psi | e^{i \varphi/2} e^{i \sigma_x \varphi/2} |\psi'\rangle] \geq 1 - 2\varepsilon$ (the last inequality follows from Cauchy-Schwarz). Upon this measurement outcome, the ancilla qubits are left in the state $|0\rangle|0\rangle$, so they may be utilized for the same subroutine at subsequent steps of a circuit.
\end{proof}

\noindent There are multiple subtleties associated with this technique. We conclude this section by remarking on each of them, demonstrating that none hinder the utility or effectiveness of the protocol.

\begin{remark}[Success probability and gadget composition]
The fact that the success probability of the above protocol is of the same order as the approximation error of the overall sequence implies that the non-unit success probability is not really a significant constraint in protocols which involve iterative nesting of corrected gadgets. 

Suppose we compose $m$ gadgets, making use of $m$ individual correction protocols. Suppose each gadget achieves approximation error of $\varepsilon_1, \dots, \varepsilon_m$, so that the overall error is on the order $\varepsilon_1 + \cdots + \varepsilon_m = \varepsilon$. Then the success probability $p$ satisfies
\begin{equation}
    p \geq (1 - 2\varepsilon_1) \cdots (1 - 2\varepsilon_m) \geq 1 - 2(\varepsilon_1 + \cdots + \varepsilon_m) = 1 - 2\varepsilon
\end{equation}
(this is easy to show via induction). Thus, the success probability is automatically made high be requiring approximation error to be low, which will be necessary \textit{a priori} to make proper use of gadgets for accurate function approximation.
\end{remark}

\begin{remark}[Ancilla reuse]
\label{rem:ancilla_reuse}
    Since successful application of an approximation of the desired half-rotation corresponds to measurement of $|0\rangle|0\rangle$ in the ancilla, which is also the required input into the above protocol, it follows that the ancilla register can be reused throughout successive applications of ``corrected" gadgets, at the same recursion level/depth in a gadget network. Thus, the number of ancilla qubits utilized in a gadget network scales not with the number of individual correction operators utilized, but rather the \textit{depth} of the network, which is usually short, for practical purposes.
\end{remark}

\begin{remark}[Overall phase]
\label{rem:overall_phase}
This protocol induces an overall phase. This can, generally speaking, be an issue if this protocol is utilized in a ``lifted" (QSVT) context, where it is effectuated in individual singular subspaces of a larger block-encoded operator, as it will yield different relative phases between the subspaces. However, in the case of the corrective protocol of Thm.~\ref{thm:qsp_correction}, we only \textit{conjugate} unitaries by this root protocol, and the induced phases cancel. Thus, this subroutine can be utilized for the correction procedure, in a lifted/QSVT context, without incurring relative phases between subspaces.
\end{remark}

\begin{remark}[Gadget nesting and multi-controlled oracles]
\label{rem:toffoli}
In order to recursively perform interlinks of gadgets beyond depth two, it will be necessary to implement controlled-$\sigma_z$ operations while utilizing black-box calls to circuits which themselves make use of controlled-$\sigma_z$ operations. This composition implies that depth $d$ interlinks will, in the most general case, require the implementation of $(d - 1)$-controlled-$\sigma_z$. This requirement is in reality simple to fulfill once given access to single-controlled-$\sigma_z$ operations: it is possible to implement an arbitrarily controlled gate (say, with $d$ controls) via the conjugation of the single-controlled gate with Toffoli gates acting on $\mathcal{O}(d)$ ancilla qubits (see Fig.~4.10 of Nielsen and Chuang, Ref.~\cite{nc_textbook_11}; for the complexity of such operations, we refer the reader to foundational \cite{bbcdmsssw_elem_gates_95} and recent \cite{czfdpp_polylog_cnot_24} works on compiling and controllizing special subsets of quantum gates with and without the use of auxiliary space).
\end{remark}


\subsection{Taking roots without ancilla qubits} \label{appx:root_without_ancilla}

\noindent
As it turns out, in certain regimes, it is possible to take the square root of an unknown $\sigma_z$-rotation, $e^{i \varphi \sigma_z}$, using \textit{no extra qubits}. The downside of this protocol is that it is only valid when the range of $\varphi$ is restricted to the subset $\left[-\pi/2, \pi/2\right]$ of $[-\pi, \pi]$. In general, there are multiple strategies for performing such a transformation, all of which involve a $\sigma_z$-QSP protocol (introduced in Appx.~\ref{appx:z_qsp}) of the unknown $\sigma_z$-rotation. Recall the half-angle formulas of elementary trigonometry, namely
\begin{equation}
    \label{eq:half_angles}
   \cos \left( \frac{t}{2} \right) = \sqrt{\frac{1 + \cos(t)}{2}} \ \ \ \ \text{and} \ \ \ \ \sin \left( \frac{t}{2} \right) = \frac{\sin(t)}{\sqrt{2}} \frac{1}{\sqrt{1 + \cos(t)}}
\end{equation}
for all $t \in [-\pi, \pi]$. This implies immediately that the functions $f_{2^n}(e^{it})$ (i.e., the function which when applied to $e^{it}$ produces [an approximation] to $e^{it/2^n}$) can be written in terms of $\cos(t)$ and $\sin(t)$ in a systematic way, via repeated application of these identities.

The strategy developed in this section will be the computation of polynomial approximations of these functions, then utilizing the tools developed in Sec.~\ref{appx:z_qsp} to map these polynomials to the $\sigma_z$-picture, where they will form a Laurent polynomial approximation of $F_n = f_{2^n}$. $G_n(z) = F_n(z^2)$ will then be an even-parity approximation of the square root function $F_{n - 1}$, when $z$ is restricted to half of the unit circle.

\begin{lemma}
    Let $f(x) = \sqrt{(1 + x)/2}$, and let $B_n(x) = f^{\circ n}(x)$ (the $n$-fold self composition). Then, for $e^{it} \in \mathbb{T}$,
    \begin{equation}
        f_{2^n}(e^{it}) = B_n(\cos(t)) + i \sin(t) 2^n B_n^{(1)}(\cos(t)).
    \end{equation}
\end{lemma}
\begin{proof}
    It is clear from Eq.~\eqref{eq:half_angles} that $B_n(\cos(t)) = \cos(t/2^n)$. Thus, computing the first derivative (notated here by $B^{(1)}$, to match the following lemma) yields
    \begin{equation}
        B_n^{(1)}(\cos(t)) \sin(t) = \frac{1}{2^n} \sin\left( \frac{t}{2^n}\right),
    \end{equation}
    which implies that
    \begin{equation}
        B_n(\cos(t)) + i \sin(t) 2^n B_n^{(1)}(\cos(t)) = \cos\left(\frac{t}{2^n}\right) + i \sin\left(\frac{t}{2^n}\right) = f(e^{it}).
    \end{equation}
    This completes the proof.
\end{proof}
\noindent
Going forward, let $C_n(x) = 2^n B_n^{(1)}(x)$.

 \begin{lemma}
    \label{lem:signs}
   The derivatives of $B_n(x)$ and $C_n(x)$ are non-zero for $x \in (-1, \infty)$, with $\text{sign}(B^{(k)}_n(x)) = (-1)^{k + 1}$ for all $k \geq 1$ and all $x \in (-1, 1]$. In addition, $\text{sign}(C^{(k)}_n(x)) = (-1)^{k}$ for all $k \geq 1$. Here the superscript $B^{(k)}, C^{(k)}$ indicate $k$-th derivatives with respect to $x$.
    \end{lemma}
    \begin{proof}
    \label{lem:bbx}
    We proceed via induction. Recall that we have $B_1(x) = f(x) = \sqrt{(1 + x)/2}$. Lem.~\ref{lem:sqrt} immediately shows that $\text{sign}(B_1^{(k)}(x)) = (-1)^{k + 1}$ for $x \in (-1, \infty)$. We assume the case of $B_n(x)$ and consider $B_{n + 1}(x)$.
    
    Recall Faà di Bruno's formula, which generalizes the chain rule to arbitrary derivatives:
  \begin{equation}
      \frac{d^{k}}{dx^k} f(g(x)) = \displaystyle\sum_{M} C_{M} f^{(m_1 + \cdots + m_k)}(g(x)) \displaystyle\prod_{j = 1}^{k} \left( g^{(j)}(x) \right)^{m_j},
  \end{equation}
  where $C_{M} > 0$ and we sum over all $k$-tuples $M = (m_1,\dots, m_k)$ such that $\sum_{i = 1}^{k} i m_i = k$. It follows that for some $k \geq 1$,
  \begin{equation}
      B^{(k)}_{n + 1}(x) = \frac{d^k}{dx^k} f(B_n(x)) = \displaystyle\sum_{M} C_{M} f^{(m_1 + \cdots + m_k)}(B_n(x)) \displaystyle\prod_{j = 1}^{k} \left( B_n^{(j)}(x) \right)^{m_j} \coloneqq \displaystyle\sum_{M} C_M s_M(x).
  \end{equation}
  Consider a single term in the above sum, of the form
  \begin{equation}
      s_M(x) = f^{(m_1 + \cdots + m_k)}(B_n(x)) \displaystyle\prod_{j = 1}^{k} \left( B_n^{(j)}(x) \right)^{m_j}.
  \end{equation}
  Clearly, it holds that
  \begin{align}
      \text{sign}(s_M(x)) &= \text{sign} \left( f^{(m_1 + \cdots + m_k)}(b(x)) \right) \displaystyle\prod_{j = 1}^{k} \left( \text{sign}\left(B_n^{(j)}(x)\right) \right)^{m_j}\\
       &= (-1)^{1 + \sum_{i = 1}^{k} m_i} \displaystyle\prod_{j = 1}^{m} (-1)^{(j + 1) m_j} \\ &= (-1) (-1)^{2 \sum_{i = 1}^{k} m_i} (-1)^{\sum_{i = 1}^{k} i m_i} \\ & = (-1)^{k + 1},
  \end{align}
  where we make use of the inductive hypothesis, and our knowledge of the signs of the derivatives of $f(x)$. This in fact holds for all terms $s_M(x)$. Thus, $\text{sign}(B_{n + 1}^{(k)}(x)) = (-1)^{k + 1}$ for all $x \in (-1, \infty)$ and $k \geq 1$. Since $C_{n}(x) = 2^n B_{n}^{(1)}(x)$, it immediately follows that $\text{sign}(C_n^{(k)}(x)) = (-1)^{k}$ for $k \geq 1$, and $x$ in the same domain.
\end{proof}
\noindent
Let $B_{(n, m)}(x)$ and $C_{(n, m)}(x)$ denote the $(m - 1)$-th order Taylor approximations of $B_n(x)$ and $C_n(x)$, respectively, centred at $x = 1$.
\begin{lemma}[Existence of approximating polynomials]
\label{lem:approx_poly}
    Let $0 < \varepsilon, \delta \leq 1$. There exist polynomials $P_{B_n}$ and $P_{C_n}$ of degree (either even or odd) $d_{B_n}, d_{C_n} \leq r$ with $r = \mathcal{O}\left(\delta^{-1} \log(\varepsilon^{-1} \delta^{-1/2})\right)$ such that for all $x \in [-1 + \delta, 3 - \delta]$, $|P_{B_n}(x) - B(x)| \leq \varepsilon$ and $|P_{C_n}(x) - C(x)| \leq \varepsilon$. In addition,
    \begin{equation}
    \label{eq:b}
        || P_B(x) ||_{[-1, 1]} \leq 2 + \varepsilon \ \ \ \text{and} \ \ \ || P_C(x) ||_{[-1, 1]} \leq \frac{2 \sqrt{2}}{3\sqrt{\delta}} + \varepsilon
    \end{equation}
    where $|| \cdot ||_{S}$ denotes the maximum magnitude over the set $S$.
\end{lemma}
\begin{proof}
We begin by bounding the absolute sums of the Taylor series of both $B_n(x)$ and $C_n(x)$. From Lem.~\ref{lem:signs}, the signs of the coefficients in each of the Taylor series will be alternating. In particular,
    \begin{equation}
        B_{(n, m)}(x + 1) = \displaystyle\sum_{k = 0}^{m - 1} b_{(n, k)} x^{k} = b_{(n, 0)} - \displaystyle\sum_{k = 1}^{m - 1} (-1)^{k} |b_{(n, k)}| x^k = b_{(n, 0)} - \displaystyle\sum_{k = 1}^{m - 1} |b_{(n, k)}| (-x)^{k},
    \end{equation}    
    implying that $\sum_{k = 0}^{m - 1} |b_{(n, k)}| (-x)^{k} = 2b_{(n, 0)} - B_{(n, m)}(x + 1)$ for any $x$. Following similar logic, one can conclude that $\sum_{k = 0}^{m - 1} |c_{(n, k)}| (-x)^{k} = C_{(n, m)}(x + 1)$. Let $0 < \delta \leq 1$. It follows that
    \begin{align}
    \label{eq:b_bnd}
        \displaystyle\sum_{k = 0}^{\infty} |b_{(n, k)}| \left( 2 - \frac{\delta}{2} \right)^{k} = \displaystyle\sum_{k = 0}^{\infty} |b_{(n, k)}| \left( -\left(-2 + \frac{\delta}{2}\right) \right)^{k} &= \lim_{k \to \infty} 2b_{(n, 0)} - B_{(n, k)}\left( -1 + \frac{\delta}{2}\right) \\ &= 2b_{(n, 0)} - B_n\left( -1 + \frac{\delta}{2}\right)
    \end{align}
    Clearly, $0 \leq B_n(x)$ , so the above sum is upper-bounded by $2$. Since the above equality holds for arbitrary $\delta \in [0, 2]$ (and thus $|x| \in [0, 1]$), this also implies that $||B_n(x)|| \leq 2$ over the entire interval $[-1, 1]$. In addition,
    \begin{align}
    \label{eq:c_bnd}
        \displaystyle\sum_{k = 0}^{\infty} |c_{(n, k)}| \left( 2 - \frac{\delta}{2} \right)^{k} = \lim_{k \to \infty} C_{(n, k)} \left( -1 + \frac{\delta}{2} \right) = C_n \left( -1 + \frac{\delta}{2} \right)
    \end{align}
    Because $C_n(\cos(t)) = \sin(t)^{-1} \sin(t/2^n)$, it follows that for $x \in [-1, 1]$, $|C_n(x)| \leq (1 - x^2)^{-1/2}$. Thus,
    \begin{equation}
        \left| C_n \left( -1 + \frac{\delta}{2} \right) \right| \leq \frac{1}{\sqrt{\delta - \delta^2/4}} \leq \frac{2}{\sqrt{3\delta}}
    \end{equation}
    This also implies that $||C_n(x)||_{[-1 + \delta, 1 - \delta]} \leq 2/\sqrt{6\delta}$, as $\delta/2$ was an arbitrary choice. Thus, via Cor.~\ref{lem:poly_exist_special}, there exists a polynomial $P_{B_n}(x)$ of degree $\mathcal{O}(\delta^{-1} \log(\varepsilon^{-1}))$ which $\varepsilon$-approximates $B_n(x)$ on $[-1 + \delta, 3 - \delta]$ and satisfies the first condition of Eq.~\eqref{eq:b}, as well as a polynomial $P_{C_n}(x)$ of degree $\mathcal{O}(\delta^{-1} \log(\varepsilon^{-1} \delta^{-1/2}))$ which $\varepsilon$-approximates $C_n(x)$ on on $[-1 + \delta, 3 - \delta]$ and satisfies the second condition of Eq.~\eqref{eq:b}. The upper bound $r$ on both of the polynomial degrees will be in $\mathcal{O}(\delta^{-1} \log(\varepsilon^{-1} \delta^{-1/2}))$ as well.
\end{proof}
\begin{remark}
    It is possible that $r$ in the above lemma scales with $n$ as well, in a way that is independent of $\varepsilon$ and $\delta$: this behaviour is not captured in the theorem utilized to prove the result.
\end{remark}
\begin{lemma}[Existence of normalized approximating polynomials]
\label{lem:norm_poly_exists}
Let $0 < \delta \leq 1$ and $0 < \varepsilon \leq 1/2$. There exist polynomials $\widetilde{P}_{B_n}(x)$ and $\widetilde{P}_{C_n}(x)$ satisfying all of the conditions of Lem.~\ref{lem:approx_poly}, along with the normalization condition that
\begin{equation}
    \label{eq:norm}
    |\widetilde{P}_{B_n}(x)|^2 + (1 - x^2) |\widetilde{P}_{C_n}(x)|^2 \leq 1
\end{equation}
for all $x \in [-1, 1]$.
\end{lemma}
\begin{proof}
Let $P_{B_n}(x)$ and $P_{C_n}(x)$ be $(\varepsilon/4)$-approximating polynomials of $B_n(x)$ and $C_n(x)$ on the interval $[-1 + \delta/2, 3 - \delta/2]$, satisfying the conditions of Lem.~\ref{lem:approx_poly}. It follows that $P_{B_n}(x) - \varepsilon/4$ and $P_{C_n}(x) - \varepsilon/4$ will be $\varepsilon/2$-approximations of $B_n$ and $C_n$ such that on $[-1 + \delta/2, 3 - \delta/2]$ which are less than or equal to $B_n$ and $C_n$, respectively, on this interval as well. Moreover, since $\varepsilon \leq 1/2$, both $\widetilde{P}_{B_n}(x)$ and $\widetilde{P}_{C_n}(x)$ are positive on $[-1, 1]$. Therefore,
\begin{align}
    ||P_{B_n}(x) - \varepsilon/4||_{[-1, 1]} = ||P_{B_n}(x)||_{[-1, 1]} - \frac{\varepsilon}{4} \leq 2, \\
    ||P_{C_n}(x) - \varepsilon/4||_{[-1, 1]} = ||P_{C_n}(x)||_{[-1, 1]} - \frac{\varepsilon}{4} \leq \frac{4}{3\sqrt{\delta}}.
\end{align}
from Eq.~\ref{eq:b}. From Lem.~\ref{lem:step}, there exists a real polynomial $\widetilde{S}(x) = \frac{1 + S\left(x + 1 - \frac{3\delta}{4}\right)}{2}$ of degree $\mathcal{O}(\delta^{-1} \log(\varepsilon^{-1}))$ which is bounded in magnitude by $1$ on $[-2, 2]$, and satisfies
  \begin{align}
       \widetilde{S}(x) \in \left[0, \frac{3\sqrt{\delta} \varepsilon}{8}\right] \ \text{for} \ x \in \left[-3 + \frac{3\delta}{4}, -1 + \frac{\delta}{2} \right], \widetilde{S}(x) \in \left[1 - \frac{3\sqrt{\delta} \varepsilon}{8}, 1 \right] \ \text{for} \ x \in \left[-1 + \delta, 1 + \frac{3\delta}{4} \right]
   \end{align}
   (this requires replacing $\delta \mapsto \delta/4$ in the lemma). Let $\widetilde{P}_{B_n}(x) = \widetilde{S}(x) (P_{B_n}(x) - \varepsilon/4)$ and $\widetilde{P}_{C_n}(x) = \widetilde{S}(x) (P_{C_n}(x) - \varepsilon/4)$. From here, note that for $x \in [-1 + \delta, 1]$, we have
   \begin{align}
   \left| B_n(x) - \widetilde{P}_{B_n}(x) \right| &= B_n(x) - \widetilde{S}(x) (P_{B_n}(x) - \varepsilon/4)
   \\ &= \left( B_n(x) - (P_{B_n}(x) - \varepsilon/4) \right) + (P_{B_n}(x) - \varepsilon/4) \left(1 - \widetilde{S}(x) \right)
   \\ &\leq \frac{\varepsilon}{2} + |P_{B_n}(x) - \varepsilon/4| \frac{3\sqrt{\delta} \varepsilon}{8} \leq \frac{\varepsilon}{2} + \frac{\sqrt{\delta}\varepsilon}{2} \leq \varepsilon.
   \end{align}
   In addition,
   \begin{align}
   \left| C_n(x) - \widetilde{P}_{C_n}(x) \right| &= C_n(x) - \widetilde{S}(x) (P_{C_n}(x) - \varepsilon/4)
   \\ &= \left( C_n(x) - (P_{C_n}(x) - \varepsilon/4) \right) + (P_{C_n}(x) - \varepsilon/4) \left(1 - \widetilde{S}(x) \right)
   \\ &\leq \frac{\varepsilon}{2} + |P_{C_n}(x) - \varepsilon/4| \frac{3\sqrt{\delta}\varepsilon}{8} \leq \varepsilon.
   \end{align}
   Now, on the interval $\left[-1 + \delta/2, -1 + \delta\right]$, note that $|\widetilde{P}_{B_{n}}(x)| \leq |B_{n}(x)|$ and $|\widetilde{P}_{C_n}(x)| \leq |C_n(x)|$, as $|\widetilde{S}(x)| \leq 1$. Thus, on this interval 
   \begin{equation}
   |\widetilde{P}_{B_{n}}(x)|^2 + (1 - x^2) |\widetilde{P}_{C_n}(x)|^2 \leq |B_n(x)| + (1 - x^2) |C_n(x)|^2 = 1,
   \end{equation}
   which is easy to verify from the definitions of $B_n(x)$ and $C_n(x)$. Finally, on $[-1, -1 + \delta/2] \cup [1 - \delta/2, 1]$, note that
   \begin{align}
       |\widetilde{P}_{B_n}(x)| \leq |\widetilde{S}(x)| |P_{B_n}(x) - \varepsilon/4| \leq \frac{\sqrt{\delta} \varepsilon}{2} \leq \frac{1}{4}, \\
       |\widetilde{P}_{C_n}(x)| \leq |\widetilde{S}(x)| |P_{C_n}(x) - \varepsilon/4| \leq \frac{\varepsilon}{2} \leq \frac{1}{4}.
   \end{align}
   Thus, it is straightforward to check that the normalization condition is satisfied on this domain as well. We conclude that $\widetilde{P}_{B_n}(x)$ and $\widetilde{P}_{C_n}(x)$ are polynomials with degree in $\mathcal{O}(\delta^{-1} \log(\varepsilon^{-1} \delta^{-1/2}))$ which are $\varepsilon$-approximations of $B_n(x)$ and $C_n(x)$ respectively, such that $|\widetilde{P}_{B_n}(x)|^2 + (1 - x^2) |\widetilde{P}_{C_n}(x)|^2 \leq 1$ on the entire interval $[-1, 1]$. This completes the proof.
\end{proof}
\noindent
Equipped with the polynomials of Lem.~\ref{lem:norm_poly_exists}, we can complete our construction.


\begin{theorem}[Existence of an approximating QSP polynomial]
\label{thm:root}
Given $0 < \delta \leq 1$ and $0 < \varepsilon \leq 1/2$, there exists a $\sigma_z$-QSP protocol $\Phi'_{n}$ of length $\mathcal{O}(\delta^{-2} \log(\varepsilon^{-2} \delta^{-1/2}))$ such that for all $\varphi \in [-\pi, \pi]$ satisfying $\cos(\varphi) \in [\delta, 1]$, $\Phi_{n}'[e^{i \varphi \sigma_z}]$ is $\varepsilon$-close to $e^{i \varphi \sigma_z/2^n}$.
\end{theorem}
\begin{proof}
Recall the bijective map $\Lambda : \mathbb{C}[z, z^{-1}]_m \rightarrow \mathbb{C}[x]_m \times \mathbb{C}[x]_{m - 1}$ defined in Thm.~\ref{thm:gamma_map}. Let $0 < \delta \leq 1$ and $0 < \varepsilon \leq 1/4$. Let $(\widetilde{P}_{B_n}(x), \widetilde{P}_{C_n}(x))$ be constructed from Lem.~\ref{lem:norm_poly_exists}, such that they are real and $\varepsilon^2/8$-approximate real polynomials $B_n(x)$ and $C_n(x)$ respectively on the interval $[-1 + 2\delta^2, 1]$, and satisfy the normalization constraint of Eq.~\eqref{eq:norm}. Clearly, $(\widetilde{P}_{B_n}(x), \widetilde{P}_{C_n}(x)) \in S_1 \cap \mathbb{R}[x]_{m} \times \mathbb{R}[x]_{m - 1}$ for some even $m \in \mathcal{O}(\delta^{-2} \log(\varepsilon^{-2} \delta^{-1/2}))$: the subset of real polynomial pairs satisfying the normalization condition of Eq.~\eqref{eq:norm}. It follows from Lem.~\ref{lem:s1} that the real Laurent polynomial $p = \Lambda^{-1}(\widetilde{P}_{B_n}, \widetilde{P}_{C_n})$ satisfies $|p(z)| \leq 1$ for all $z \in \mathbb{T}$. Let $g(z) = p(z^2)$. Clearly, $p$ will be even, and thus has parity $m \ \text{mod} \ 2$. It follows from Thm.~\ref{thm:z_qsp} that there exists a length $2m \in \mathcal{O} \left( \delta^{-2} \log(\varepsilon^{-2} \delta^{-1/2}) \right)$ $\sigma_z$-QSP protocol $\Phi_{n}'$ such that
\begin{align}
    \langle 0 | \Phi_{n}'[e^{i \varphi \sigma_z}] |0\rangle &= g(e^{i \varphi}) = p(e^{2i\varphi}) = \widetilde{P}_{B_n}(\cos(2\varphi)) + i \sin(2\varphi) \widetilde{P}_{C_n}(\cos(2\varphi))
\end{align}
where, for $\cos(2\varphi) \in [-1 + 2\delta^2, 1]$, $|\widetilde{P}_{B_n}(\cos(2\varphi)) - B_n(\cos(2\varphi))| \leq \varepsilon^2/8$ and $|\widetilde{P}_{C_n}(\cos(2\varphi)) - C_n(\cos(2\varphi))| \leq \varepsilon^2/8$. Thus, via triangle inequality,
\begin{equation}
    \left| \langle 0 | \Phi_{n}'[e^{i \varphi \sigma_z}] H |0\rangle - B(\cos(2\varphi)) + i \sin(2\varphi) C(\cos(2\varphi)) \right| \leq \frac{\varepsilon^2}{4}.
\end{equation}
For all $2\varphi \in [-\pi, \pi]$, so for all $\varphi \in [-\pi/2, \pi/2]$ ($\cos(\varphi) \in [0, 1]$), we know that $B_n(\cos(2\varphi)) + i \sin(2\varphi) C_n(\cos(2\varphi)) = e^{i \varphi/2^n}$. Note that
\begin{equation}
    \cos(2\varphi) = 2\cos^2(\varphi) - 1 \in [-1 + 2\delta^2, 1] \ \ \ \text{and} \ \ \cos(\varphi) \geq 0 \Longleftrightarrow \cos(\varphi) \in [\delta, 1].
\end{equation}
Therefore, for $\cos(\varphi) \in [\delta, 1]$,
\begin{align}
    \left| \langle 0 | \Phi_{n}'[ e^{i \varphi \sigma_z}] |0\rangle - e^{i \varphi/2^n} \right| \leq \frac{\varepsilon^2}{4} \ \ \ \text{and} \ \ \ \left| \langle 1 | \Phi_{n}'[e^{i \varphi \sigma_z}] |1\rangle - e^{-i \varphi/2^n} \right| \leq \frac{\varepsilon^2}{4}.
\end{align}
It is easy to verify that $\Phi_{n}'[e^{i \varphi \sigma_z}]$ is $\text{SU}(2)$. Letting $a = \langle 0 | \Phi_{n}'[e^{i \varphi \sigma_z}] |0\rangle$, it is clear that
\begin{align}
    1 - |a| \leq 1 - |a|^2 = e^{-i\varphi/2^n} e^{i \varphi/2^n}  - a^{*} a &= e^{-i \varphi/2^n} (e^{i \varphi/2^n} - a) + a (e^{-i \varphi/2^n} - a^{*})
    \\ & |e^{i \varphi/2^n} - a| + |e^{-i\varphi/2^n} - a^{*}| \leq \frac{\varepsilon^2}{2}
\end{align}
so from Lem.~\ref{lem:su2_bnd}, $|| \Phi_{n}'[e^{i \varphi \sigma_z}] - e^{i \varphi \sigma_z/2^n}|| \leq \varepsilon$, and the proof is complete.
\end{proof}


\noindent As a corollary, it is easy to see that the procedure provided above can be utilized to take roots of $\sigma_x$-rotations as well (since $e^{i \varphi \sigma_x} = H e^{i \varphi \sigma_z} H$). This result will be useful in Sec.~\ref{sec:examples}.

\begin{corollary}
\label{cor:root}
Given $0 < \delta \leq 1$ and $0 < \varepsilon \leq 1/2$, there exists a QSP protocol $\Phi_{n}$ of length $\zeta = \mathcal{O}(\delta^{-2} \log(\varepsilon^{-2} \delta^{-1/2}))$ such that for all $\varphi \in [-\pi, \pi]$ satisfying $\cos(\varphi) \in [\delta, 1]$, $\Phi_{n}[e^{i \varphi \sigma_x}]$ is $\varepsilon$-close to $e^{i \varphi \sigma_x/2^n}$.
\end{corollary}


\section{Techniques for polynomial approximation}
\label{sec:poly_bnds}

\noindent
In this section, we present a toolbox of useful techniques for constructing polynomial approximations, in the context of QSP. The first results showcased are constructive, and are utilized for when defining and proving properties of the extraction and root polynomials of Appx.~\ref{appx:extraction} and Appx.~\ref{appx:roots}.

    \begin{theorem}[Polynomial approximation with bounding series]
    \label{thm:bnd}
   Let $f(x)$ be a function, let $P_n(x) = \sum_{k = 0}^{n - 1} p_k x^{\lambda k}$ be a family of $(n - 1)$-th order approximating polynomials, where $\lambda \in \mathbb{Z}^{+}$, such that $|\lim_{n \to \infty} P_n(x) - f(x)|$ = 0 for each $x \in (-1, 1)$ (uniform convergence is not required). Let $0 < \varepsilon, \delta \leq 1$.
   Suppose $|p_k| \leq q_k$ for a weakly decreasing sequence $q_k$. Let $g(x) = \sum_{k = 0}^{\infty} q_k x^{\lambda k}$. Let $M = |g(1 - \delta)|$. Then, if $n \geq (\lambda \delta)^{-1} \log(M \varepsilon^{-1})$, $|P_n(x) - f(x)| \leq \varepsilon$ for all $x \in (-1 + \delta, 1 - \delta)$.
\end{theorem}

\begin{proof}
    Clearly, for each $x \in (-1 + \delta, 1 - \delta)$, $n$ can be chosen such that $|P_n(x) - f(x)| \leq \varepsilon$, as $|\lim_{n \to \infty} P_n(x) - f(x)| = 0$ for all $x \in (-1, 1)$. For some $n$, note that
    \begin{align}
        |P_n(x) - f(x)| &\leq |\lim_{n \to \infty} P_n(x) - P_n(x)| + |\lim_{n \to \infty} P_n(x) - f(x)| = |\lim_{n \to \infty} P_n(x) - P_n(x)| \\
        &= \left| \displaystyle\sum_{k = n}^{\infty} p_k x^{\lambda k} \right| = |x|^{\lambda n} \left| \displaystyle\sum_{k = n}^{\infty} p_k x^{\lambda (k - n)} \right| \leq |x|^{\lambda n} \displaystyle\sum_{k = n}^{\infty} |p_k| |x|^{\lambda(k - n)}
        \leq |x|^{\lambda n} \displaystyle\sum_{k = n}^{\infty} g_k |x|^{\lambda(k - n)} \\ &\leq |x|^{\lambda n} \displaystyle\sum_{k = n}^{\infty} g_{n - k} |x|^{\lambda(k - n)} = |x|^{\lambda n} \displaystyle\sum_{k = 0}^{\infty} g_k |x|^{\lambda x} \leq |x|^{\lambda n} |g(1 - \delta)|,
    \end{align}
    where we use the fact that $|p_k| \leq q_k \leq q_{k - n}$. Thus, 
    \begin{align}
        \label{eq:series_bnd_diff_eps}
        |P_n(1 - \delta) - f(1 - \delta)| \leq M |x|^{\lambda n} \leq M (1 - \delta)^{\lambda n}.
    \end{align}
    Note that
    \begin{equation}
        \label{eq:series_bnd_n}
        M (1 - \delta)^{\lambda n} \leq \varepsilon \Leftrightarrow n \geq \frac{1}{\lambda} \log \left( \frac{M}{\varepsilon} \right) \log \left( \frac{1}{1 - \delta} \right)^{-1}.
    \end{equation}
    Finally, note that for $\delta \in (0, 1]$, we have $\log(1 - \delta) \leq -\delta$, so $\log(1/(1 - \delta)) = -\log(1 - \delta) \geq \delta$. This implies that $\log(1/(1 - \delta))^{-1} \leq 1/\delta$. It follows that if we set
    \begin{align}
        n \geq \frac{1}{\lambda \delta} \log \left( \frac{M}{\varepsilon} \right) = \mathcal{O} \left( \frac{1}{\delta} \log \left(\frac{M}{\varepsilon}\right) \right),
    \end{align}
    then by Eq.~\eqref{eq:series_bnd_n} and Eq.~\eqref{eq:series_bnd_diff_eps}, the bound will hold. Note that the resulting polynomial will have degree $\lambda n$.
\end{proof}
\noindent
This result is fairly general, and yields the following corollaries immediately.

\begin{corollary}[Polynomial approximation with weakly decreasing coefficients]
\label{cor:weak_coeffs}
Let $f(x)$ be a function, let $P_n(x) = \sum_{k = 0}^{n - 1} p_k x^{\lambda k}$ be a family of $(n - 1)$-th order approximating polynomials, where $\lambda \in \mathbb{Z}^{+}$, such that $|\lim_{n \to \infty} P_n(x) - f(x)|$ = 0 for all $x \in (-1, 1)$. Let $0 < \varepsilon, \delta \leq 1$. Suppose $|p_k|$ is weakly decreasing. Moreover, suppose the sign of each $p_k$ is constant, or alternating between positive and negative (and, in this latter case, $\lambda$ is odd). In the former case, let $N = |f(1 - \delta)|$. In the latter, let $N = |f(-1 + \delta)|$. Then, if $n \geq (\lambda \delta)^{-1} \log(N \varepsilon^{-1})$, then $|P_n(x) - f(x)| \leq \varepsilon$ for all $x \in (-1 + \delta, 1 - \delta)$.
\end{corollary}

\begin{proof}
    Let $g_k = |p_k|$. In the case that each $p_k$ has the same sign, $g(x) = \pm \sum_{k = 0}^{\infty} p_k x^{k} = \pm f(x)$, so $M = |f(1 - \delta)|$, and we can apply Thm.~\ref{thm:bnd}. In the latter case, $g(x) = \pm \sum_{k = 0}^{\infty} (-1)^{k} p_k x^{\lambda k} = \pm \sum_{k = 0}^{\infty} p_k (-x)^{\lambda k}$ (the final equality follows from the fact that $\lambda$ is odd), so $M = |g(1 - \delta)| = |f(-1 + \delta)|$, and we once again can apply Thm.~\ref{thm:bnd}.
\end{proof}

\begin{corollary}[Polynomial approximation with constant bound]
    \label{cor:const_bnd}
    Let $f(x)$ be a function, let $P_n(x) = \sum_{k = 0}^{n - 1} p_k x^{\lambda k}$ be a family of $(n - 1)$-th order approximating polynomials, where $\lambda \in \mathbb{Z}^{+}$, such that $|\lim_{n \to \infty} P_n(x) - f(x)|$ = 0 for all $x \in (-1, 1)$. Let $0 < \varepsilon, \delta \leq 1$. Suppose $|p_k| \leq C$. Then, if $n \geq (\lambda \delta)^{-1} \log(C \delta^{-1} \varepsilon^{-1})$, $|P_n(x) - f(x)| \leq \varepsilon$ for all $x \in (-1 + \delta, 1 - \delta)$.
\end{corollary}

\begin{proof}
   Let $q_k = C$, so $g(x) = C \sum_{i = 0}^{\infty} x^i = \frac{C}{1 - x}$ for $x \in (-1, 1)$. In this case, $M = |g(1 - \delta)| = \frac{C}{\delta}$. We can then apply Thm.~\ref{thm:bnd} and the result follows.
\end{proof}

\begin{theorem}
    \label{thm:taylor}
    Let $C(x)$ be an analytic function, let $C_{n}(x) = \sum_{i = 0}^{n - 1} c_i x^i$ be the $(n - 1)$-th order Taylor series of $f$ centered at $0$. Suppose that the $n$-th derivative of $C(x)$ is non-zero for all $x \in [0, d)$ where $d > 0$. Moreover, suppose $\text{sign}(c_n) = (-1)^{n}$. If $n$ is even, then $C_n(x) < C(x)$ for all $x \in \mathcal{D} = (0, d)$. If $n$ is odd, then $C_n(x) > C(x)$ for all $x \in \mathcal{D}$.
\end{theorem}
\begin{proof}
Clearly, $C^{(k)}(0) = C_n^{(k)}(0)$ for all $k$ from $0$ to $n - 1$. Therefore, we have
    \begin{equation}
        C(x) - C_n(x) = \displaystyle\sum_{k = 0}^{n - 1} c_i x^{i} - \displaystyle\sum_{k = 0}^{\infty} c_i x^{i} = \displaystyle\sum_{k = n}^{\infty} c_i x^{i} = (-1)^{n} |c_n| x^n + \displaystyle\sum_{k = n + 1}^{\infty} c_k x^{i},
    \end{equation}
     where we know that $|c_n| > 0$. Therefore the term with lowest degree in $x$ of the above sum has a coefficient of non-zero magnitude: positive for even $n$ and negative for odd $n$. Thus, for an open neighbourhood $U = (0, \gamma)$ with $\gamma > 0$, $\text{sign}(C(x) - C_n(x)) = (-1)^{n}$, so when $n$ is even, $C(x) > C_n(x)$ and when $n$ is odd, $C_n(x) < C(x)$. 
     
     Now, suppose there exists some $r_0 \in \mathcal{D}$ such that $C(r_0) = C_n(r_0)$. We know that $r_0 > 0$ as $\mathcal{D} \subset (0, \infty)$. Recall the single-variable mean value theorem (MVT), which states that for $a, b \in \mathbb{R}$ with $a < b$ and differentiable functions $f$, $g$, if $f(a) = g(a)$ and $f(b) = g(b)$, there exists some $c \in (a, b)$ such that $f^{(1)}(c) = g^{(1)}(c)$. Since $C(0) = C_n(0)$, there must exist some $r_1 \in (0, r_0)$ such that $C^{(1)}(r_1) = C_n^{(1)}(r_1)$. Recall that $C^{(k)}(0) = C_n^{(k)}(0)$ for all $k$ from $0$ to $n - 1$. Thus, we can inductively apply MVT until we arrive at a point $r_n \in \mathcal{D}$ such that the $n$-th derivatives satisfy $C^{(n)}(r_n) = C_n^{(n)}(r_n) = 0$. But this is a contradiction as $C^{(n)}(x) \neq 0$ for $x \in \mathcal{D}$. It follows that we cannot have $C(x) = C_n(x)$ for $x \in \mathcal{D}$. Since $C_n(x) < C(x)$ in $(0, \gamma)$ for $n$ even, we must then have that $C_n(x) < C(x)$ on the whole interval $\mathcal{D}$ in this case. Similarly, for $n$ odd, we must have $C_n(x) > C(x)$ for $x \in \mathcal{D}$.
\end{proof}
\noindent
In addition, we restate some of the theorems of~\cite{gslw_qsvt_19}, which give existence results for Fourier-based approximations of polynomials, with particular properties.

\begin{lemma}[Corollary 66, Ref.~\cite{gslw_qsvt_19}]
    \label{lem:poly_exist}
    Let $x_0 \in [-1, 1]$, $r \in (0, 2]$, $\delta \in (0, r]$, and let $f : [x_0 - r - \delta, x_0 + r + \delta] \rightarrow \mathbb{C}$ be such that $f(x_0 + x) = \sum_{\ell = 0}^{\infty} a_{\ell} x^{\ell}$ for all $x \in [-r - \delta, r + \delta]$. Suppose $B > 0$ is chosen such that $\sum_{\ell = 0}^{\infty} (r + \delta)^{\ell} |a_{\ell}| \leq B$. Then given $\varepsilon \in (0, \frac{1}{2B}]$, there exists an efficiently-computable polynomial $P \in \mathbb{C}[x]$ of degree $\mathcal{O}\left(\frac{1}{\delta} \log\left( \frac{B}{\varepsilon} \right) \right)$ such that $|f(x) - P(x)| \leq \varepsilon$ for $x \in [x_0 - r, x_0 + r]$, and we have
    \begin{equation}
    \label{eq:poly_bnds}
        || P(x) ||_{[-1, 1]} \leq \varepsilon + ||f(x)||_{\left[ x_0 - r - \delta/2, x_0 + r + \delta/2\right]}.
    \end{equation}
    where $|| \cdot ||_I$ denotes the infinite norm over the interval $I$.
\end{lemma}

\noindent This allows us to prove the following corollary.

\begin{corollary}
    \label{lem:poly_exist_special}
    Suppose $f : [-1 + \delta/2, 3 - \delta/2] \rightarrow \mathbb{C}$ is such that $f(x + x_0) = \sum_{\ell = 0}^{\infty} a_{\ell} x^{\ell}$ for all $x \in [-2 + \delta/2, 2 - \delta/2]$. Suppose $B > 0$ is chosen such that $\sum_{\ell = 0}^{\infty} (2 - \delta/2)^{\ell} |a_{\ell}| \leq B$. Then given $\varepsilon \in (0, \frac{1}{2B}]$, there exists an efficiently-computable polynomial $P \in \mathbb{C}[x]$ of degree $\mathcal{O}\left(\frac{1}{\delta} \log\left( \frac{B}{\varepsilon} \right) \right)$ such that $|f(x) - P(x)| \leq \varepsilon$ for $x \in [-1 + \delta, 3 - \delta]$, and we have
    \begin{equation}
        || P(x) ||_{[-1, 1]} \leq \varepsilon + ||f(x)||_{\left[ -1 + 3\delta/4, 3 - 3\delta/4\right]}.
    \end{equation}
\end{corollary}
\begin{proof}
    We set $x_0 = 1$, $r = 2 - \delta$. We then apply Lem.~\ref{lem:poly_exist} (where we send $\delta \mapsto \delta/2$ in the statement of the lemma).
\end{proof}

\begin{remark}
    In the case that $f$ has a real codomain, then this lemma can be amended to require that $P \in \mathbb{R}[x]$: this simply follows from the fact that if $f(x)$ is real, then 
    \begin{equation}
    |f(x) - P(x)| = \sqrt{(\text{Re}[P](x) - f(x))^2 + \text{Im}[P](x)^2} \geq |\text{Re}[f](x) - P(x)|
    \end{equation}
    and $||P(x)||_{\mathcal{D}} \geq ||\text{Re}[P(x)]||_{\mathcal{D}}$ for some subset $\mathcal{D} \subset [-1, 1]$. Thus, the bounds of Lem.~\ref{lem:poly_exist} still hold when $P$ is replaced with $\text{Re}[P] \in \mathbb{R}[x]$.
\end{remark}

\begin{lemma}[Polynomial approximation of the sign function, Ref.~\cite{gslw_qsvt_19}]
\label{lem:step}
For each $\delta > 0$, $0 < \varepsilon < 1/2$, there exists an efficiently computable odd polynomial $S(x) \in \mathbb{R}[x]$ of degree in $\mathcal{O}\left(\frac{1}{\delta} \log\left( \frac{1}{\varepsilon} \right)\right)$ such that $|S(x)| \leq 1$ for all $x \in [-2, 2]$ and for all $x \in [-2, -\delta] \cup [\delta, 2]$, $|S(x) - \text{sign}(x)| \leq \varepsilon$.
\end{lemma}

\begin{lemma}[Polynomial approximation of rectangle functions, Ref.~\cite{gslw_qsvt_19}]
    \label{lem:rectangle}
   Let $0 < \varepsilon, \delta < \frac{1}{2}$. Let $t \in [-1, 1]$. There exists an even polynomial $R(x) \in \mathbb{R}[x]$ of degree in $\mathcal{O} \left( \frac{1}{\delta} \log \left( \frac{1}{\varepsilon} \right) \right)$ such that $|R(x)| \leq 1$ for all $x \in [-1, 1]$, and
   \begin{align}
       R(x) \in [0, \varepsilon] \ \ \ \ \text{for} \ x \in [-1, -t - \delta] \cup [t + \delta, 1], \\
       R(x) \in [1 - \varepsilon, 1] \ \ \ \ \text{for} \ x \in [-t + \delta, t - \delta].
   \end{align}
\end{lemma}

\noindent
Finally, we state and prove miscellaneous lemmas, which are used in repeatedly throughout this work.
\begin{lemma}
    \label{lem:binom_half}
    Given $j \in \mathbb{Z}^{+}$,
    \begin{equation}
        \binom{-\frac{1}{2}}{j} = (-1)^{j} \prod_{k = 1}^{j} \left( 1 - \frac{1}{2k} \right) \ \ \ \ \text{and} \ \ \ \ \binom{\frac{1}{2}}{j} = \frac{(-1)^{j}}{2j - 1} \prod_{k = 1}^{j} \left( 1 - \frac{1}{2k} \right).
    \end{equation}
    Thus, both binomial functions are decreasing in magnitude as a function of $j$, and have sign $(-1)^{j}$.
\end{lemma}

\begin{proof}
    For $a, b \in \mathbb{Z}$, note that
      \begin{align}
   \binom{\frac{a}{2}}{b} = \frac{\left( \frac{a}{2} \right)!}{b! \left( \frac{a}{2} - b \right)!} & = \frac{1}{b!} \prod_{k = 1}^{b} \left( \frac{a}{2} + 1 - k \right) = \frac{1}{b!} \left(-\frac{1}{2} \right)^{b} \prod_{k = 1}^{b} \left( 2k - (a + 2) \right) = (-1)^{b} \prod_{k = 1}^{b} \left( 1 - \frac{a + 2}{2k} \right).
   \end{align}
   Thus, when $a = -1$, $b = j$, we have
   \begin{equation}
       \binom{-\frac{1}{2}}{j} = (-1)^{j} \prod_{k = 1}^{j} \left( 1 - \frac{1}{2k}\right),
   \end{equation}
   and when $a = 1$, $b = j$, we have
   \begin{equation}
       \binom{\frac{1}{2}}{j} = \frac{1}{j!} \left(-\frac{1}{2}\right)^{j} \prod_{k = 1}^{j} \left( 2k - 3 \right) = \frac{1}{2j - 1} (-1)^{j} \prod_{k = 1}^{j} \frac{2k - 1}{2k} = \frac{(-1)^{j}}{2j - 1} \prod_{k = 1}^{j} \left( 1 - \frac{1}{2k} \right).
   \end{equation}
\end{proof}
\begin{lemma}
    \label{lem:sqrt}
    Let $f(x) = \sqrt{1 + x}$. Then for $k \geq 1$, $f^{(k)}(x) = \frac{(-1)^{k + 1} (2k - 3)!!}{2^k} \left(1 + x \right)^{\frac{1}{2} - k}$ where the double factorial takes the product of all positive integers less than or equal to, and of equal parity to $2k - 3$ and $(-1)!! = 1$. Here the superscript $f^{(k)}$ indicates the $k$-th derivative of $f$ with respect to its argument.
\end{lemma}
\begin{proof}
    This is simple to check via induction. The base case is immediate. We assume the case of $k$, and note that
    \begin{equation}
                \frac{d}{dx} \frac{(-1)^{k + 1} (2k - 3)!!}{2^k} \left( 1 + x \right)^{\frac{1}{2} - k} = \frac{(-1)^{k + 2} (2k - 1)!!}{2^{k + 1}} \left( 1 + x \right)^{-\frac{1}{2} - (k + 1)},
        \end{equation}
        which proves the claim.
\end{proof}


\section{Examples of composite gadgets} \label{appx:gadget_compositions}
\noindent

\noindent In this section, we provide proofs and further explanation of several of the constructions introduced in Sec.~\ref{sec:examples}. 

\subsection{Proofs of results}

\noindent We begin with a proof demonstrating the existence of the simple bandpass family.

\bandpass*
\begin{proof}
As is stated, we perform a composition of $(1, 1)$ gadgets: the constant shift gadget $\mathfrak{G}_{\text{shift}}(T_4^{-1}(-T_4(a)))$ of Ex.~\ref{ex:shift_gadget} and the step function gadget $\mathfrak{G}^{(\ell)}_{\text{step}}$ of Ex.~\ref{ex:step_fun_gadget}. Since $\mathfrak{G}^{(\ell)}_{\text{step}}$ is atomic, it is not necessary to perform any kind of corrective protocol: the composition $\mathfrak{G}^{(\ell)} \circ \mathfrak{G}_{\text{shift}}(2T_4^{-1}(a^2))$ immediately achieves the function $f\left( \frac{1}{2} T_4(x) - \frac{1}{2} T_4(a) \right)$, where $f(x)$ is an $\varepsilon$-approximation of $\text{sign}(x)$ on the interval $x \in [-1, -\delta] \cup [\delta, 1]$. Suppose $x \in [-a + \delta, a - \delta]$. Then, since $x = \pm 1/\sqrt{2}$ are the local minima of $T_4(x)$ and $a \leq 1/\sqrt{2}$,
\begin{align}
    \frac{1}{2} T_4(x) - \frac{1}{2} T_4(a) \geq \frac{1}{2} (T_4(a - \delta) - T_4(a)) &= 4 \left( (a - \delta)^4 - (a - \delta)^2 - a^4 + a^2  \right)
    \\ & = 4 \left( a^2 - (a - \delta)^2 - \left( (a - \delta)^2 + a^2 \right) \left( a^2 - (a - \delta)^2 \right) \right)
    \\ & = 4 (a^2 - (a - \delta)^2)(1 - (a - \delta)^2 - a^2)
    \\ & = 4 (2a\delta - \delta^2)((1 - 2a^2) + (2a\delta - \delta^2))
    \\ & \geq 4 (2a\delta - \delta^2)^2 \geq 4 a^2 \delta^2.
\end{align}
Thus, we require that $f$ is an $\varepsilon$-approximation of the sign function on $[-1, -4a^2\delta^2] \cup [4a^2\delta^2, 1]$ in order for the $\varepsilon$-condition to hold. The resulting depth $\zeta'$ of this gadget is then given by substituting $\delta \mapsto 4a^2\delta^2$ in Eq.~\eqref{eq:bandpass_scaling}, so
\begin{equation}
\zeta' = \mathcal{O}\left( (2a\delta)^{-4(\nu_{\ell} + 1)} \log^{1 + \nu_{\ell}} \left( \varepsilon^{-1} \right) \right)
\end{equation}
and the proof is complete.
\end{proof}

\noindent
We now provide a proof of Thm.~\ref{thm:inv_cheb}, which we restate for convenience.

\inversecheb*

\begin{proof}
Let $\Phi_{n}$ be the length $\zeta = \mathcal{O}(\delta^{-1} \log(\varepsilon^{-2} \delta^{-1/4}))$ protocol of Cor.~\ref{cor:root} which approximately achieves $e^{i \theta \sigma_z} \mapsto e^{i \theta \sigma_z / 2^n}$ for $\cos(\theta) \in [\sqrt{\delta}, 1]$. Let $\mathfrak{G}^{(n)}$ be the corresponding $(1, 1)$ atomic gadget. Now, consider the $(3, 1)$ atomic gadget with $\Xi \equiv \Phi = \{0, 0, \phi_0, 0, \phi_1, 0, \dots, \phi_{N - 1}\}$, where $\Xi' \equiv \Phi_n = (\phi_0, \dots, \phi_{N - 1})$ is the protocol obtained from $\mathfrak{G}^{(n)}$. Let $S = \{2, 1, 0, 1, 0, \dots, 1, 0\}$. The $(1, 1)$ gadget obtained from $\mathfrak{G}^{(n)}$ by pinning (Def.~\ref{def:aux_gadget_operations}) the index-$1$ input as $U_1 = e^{-i \pi \sigma_x/2}$ and the index-$2$ input as $U_2 = e^{i \pi \sigma_x/2^{n + 1}}$ then achieves, given an embeddable unitary $U_0 = e^{i \arccos(x) \sigma_x}$, the unitary
\begin{align}
    U = \Phi[U_0, U_1, U_2] &= e^{i \pi \sigma_x / 2^{n + 1}} \Phi_n[e^{i (\arccos(x) - \pi/2) \sigma_x}].
\end{align}
Since $x \in [-1 + \delta, 1 - \delta]$, it follows that $\cos(\arccos(x) - \pi/2) = \sqrt{1 - x^2} \in [\sqrt{2\delta - \delta^2}, 1] \subset [\sqrt{\delta}, 1]$. Thus, by Cor.~\ref{cor:root},
\begin{align}
    || U - e^{i \arccos(x) \sigma_x / 2^n} || = || \Phi_n[e^{i (\arccos(x) - \pi/2) \sigma_x}] - e^{i (\arccos(x) - \pi/2) \sigma_x / 2^{n}}|| \leq \varepsilon.
\end{align}
and the proof is complete.
\end{proof}

Let us now discuss the gadgets which enact the transformations of Eq.~\eqref{eq:primitives} and Eq.~\eqref{eq:primitives2}. Indeed, achieving these transformations is simply a matter of pre-composing protocols yielding these formulas with the Chebyshev inversion protocol Thm.~\ref{thm:inv_cheb}. What is particularly convenient about these protocols is that they \textit{require no correction}. The inverse Chebyshev polynomials fall under the umbrella of approximate phase functions, discussed in Sec.~\ref{appx:embeddability_criterion}, so their outputs will automatically be $\varepsilon$-embeddable. We then have the following results.

\arbitrarymultiplication*

\begin{proof}
Denote the $(1, 1)$ gadget achieving and $\varepsilon$-approximation of $T_2^{-1}(x)$ for $x \in [-1 + \delta, 1 - \delta]$ (Thm.~\ref{thm:inv_cheb}) by $\mathfrak{G}$: this gadget will have a cost/depth of $\mathcal{O}(\delta^{-1} \log(\varepsilon^{-2} \delta^{-1/4}))$. Moreover, the output will be $\varepsilon$-embeddable (Rem.~\ref{rem:emb_crit}). Let $\mathfrak{G}'$ denote the $(2, 1)$ gadget which partially composes $\mathfrak{G}$ with the first leg of the multiplication gadget of Ex.~\ref{ex:product_gadget}. The resulting gadget will have the same asymptotic cost, and will $2\varepsilon$-achieve $T_2(f(x_0)) x_1$ for $x_0$ and $x_1$ in the desired ranges, where $f$ is an $\varepsilon$-approximation of $T_2^{-1}(x)$ (the $2\varepsilon$-error comes from the fact that one of the oracles, being used twice, is $\varepsilon$-embeddable, so we apply Lem.~\ref{lem:error}). Then, note that $T_2'(x) = 4x$, so $T_2(x)$ has a Lipschitz constant of $4$ on $[-1, 1]$, so
\begin{align}
    | T_2(f(x_0)) x_1 - x_0 x_1 | \leq 4 |f(x_0) - T_2^{-1}(x_0) | \leq 4\varepsilon.
\end{align}
Thus, the total error is of order $\mathcal{O}(\varepsilon)$, and the proof is complete.
\end{proof}

\arbitraryaddition*

\begin{proof}
We proceed similarly to the previous theorem. Denote the $(1, 1)$ gadget achieving an $\varepsilon$-approximation of the function $T_4^{-1}(x)$ for $x \in [-1 + \delta, 1 - \delta]$, at a depth of $\mathcal{O}(\delta^{-1} \log(\varepsilon^{-2} \delta^{-1/2}))$. As before, the output of this protocol of $\varepsilon$-embeddable. We compose two instances of $\mathfrak{G}$ with the input legs of the addition gadget of Ex.~\ref{ex:sum_gadget}. This will yield a $\mathcal{O}(\varepsilon)$-approximation of $(T_4(f(x_0)) + T_4(f(x_1)))/2$ for $f(x)$ and $\varepsilon$-approximation of $T_4^{-1}(x)$. Note that $T_4'(x) = 32x^3 - 16x$, which has $16$ as its maximum, so the Lipschitz constant on $[-1, 1]$ is $16$, and
\begin{equation}
    \left| \frac{T_4(f(x_0)) + T_4(f(x_1))}{2} - \frac{x_0 + x_1}{2} \right| \leq 8 \left( |f(x_0) - T_4^{-1}(x_0)| + |f(x_1) - T_4^{-1}(x_1)| \right) \leq 16\varepsilon.
\end{equation}
Therefore, the total approximation error is on the order $\mathcal{O}(\varepsilon)$, and the proof is complete.
\end{proof}

To conclude our discussion of addition and multiplication gadgets, we will consider the embeddability properties of the gadgets produced from these protocols.

\begin{definition}[Domain-modified gadget]
   Given an $(a, b)$ gadget $\mathfrak{G}$, define the parameterized family of gadgets $\mathfrak{G}(\phi)$ to be those which achieve the unitaries $e^{i \phi \sigma_z/2} U_k' e^{-i \phi \sigma_z/2}$ for $k \in [b]$. Generally speaking, particular choice of $\phi$ will allow us to vary the domain on which the unitaries produced by the gadget are half-twisted embeddable, and therefore can be corrected with the ancilla-free corrective procedure. We let $\mathfrak{G}_{\text{mult}}(\phi)$ and $\mathfrak{G}_{\text{add}}(\phi)$ denote the domain-modified multiplication and addition gadgets of Ex.~\ref{ex:product_gadget} and Ex.~\ref{ex:sum_gadget}.
\end{definition}

\begin{remark}[Half-twisted embeddability of multiplication and addition]
\label{rem:add_mult_embed}
For the case of multiplication, it is easy to compute the real component of the off-diagonal element of the achieved unitary. In particular, if we let $U'(x_0, x_1)$ be unitary achieved by $\mathfrak{G}_{\text{mult}}(0)$, then
\begin{align}
    \Im\left[\langle 0 | U'(x_0, x_1) |1\rangle \right] = \sqrt{1 - x_1^2},
\end{align}
which will always be greater than or equal to $0$. In fact, given the restriction of $x_1 \in [-\sqrt{1 - \delta^2}, \sqrt{1 - \delta^2}]$, it follows that $\Im\left[\langle 0 | U'(x_0, x_1) |1\rangle \right] \geq \delta$. For the general product gadget of Thm.~\ref{thm:provable_mult}, because the Chebyshev inversion only occurs on the leg encoding variable $x_0$, the same holds true for this case. In summary, \emph{both multiplication gadgets always yield half-twisted embeddable output} (Rem.~\ref{rem:half_twisted_criteria}).

Now, we turn our attention to the addition gadgets. Let $U_{\phi}'(x_0, x_1)$ denote the unitary achieved by $\mathfrak{G}_{\text{add}}(\phi)$. As is easy to verify by hand or via symbolic mathematical software, we have, for example
\begin{align}
\label{eq:embed}
    \Im\left[\langle 0 | U'_{\pi}(x_0, x_1) |1\rangle\right] = (2x_0 - 4x_0^3) \sqrt{1 - x_0^2} + (2x_1 - 4x_1^3) \sqrt{1 - x_1^2}, \\
    \Im [\langle 0 | U'_{\pi/2}(x_0, x_1) |1\rangle]
    = (2 - 4 x_1^2) x_0 \sqrt{1 - x_0^2} + (4x_0^2 - 2) x_1 \sqrt{1 - x_1^2}.
\end{align}
The conditions under which these outputs have fixed sign are more complicated than the case of multiplication, but it is apparent that they \emph{can} achieve fixed sign for significant domains of interest in each of the variables $x_0$ and $x_1$. For example, for $\phi = \pi$, if we are guaranteed that $x_0, x_1 \in [0, 1/\sqrt{2}]$, then Eq.~\eqref{eq:embed} will always be non-negative, and the output will be half-twisted embeddable by Rem.~\ref{rem:half_twisted_criteria}. There exist many such regions in the domain of $x_0, x_1$, which can be utilized to guarantee half-twisted embeddable output on a case-by-case basis. However, half-twisted embeddability \emph{will not hold generically} for the addition gadget.
\end{remark}

\noindent Now, we revisit some of the examples discussed in Sec.~\ref{sec:examples}. We begin with the example of the $2^n$-mean. 

\subsection{Contrasting gadgets with LCU: the  \texorpdfstring{$2^n$}{2\^n}-mean}

\noindent Our goal is to explicitly characterize the fundamental differences between using the gadget method vs. LCU for achieving the $2^n$-mean of variables $x_k$ which are encoded in $\sigma_x$-rotations $e^{i \arccos(x_k) \sigma_x}$. We assume, for the sake of performing corrections, that we have access to controlled variants of the oracles we are summing.

\meangadget*

\begin{remark}[Incomparability for LCU and gadgets]
\label{rem:incomp}
Even in the natural context of taking a sum, the gadget formalism and LCU are still \emph{highly incomparable}, as the circuit-level operations they achieve are fundamentally different. Despite the fact that LCU achieves better scaling in error/domain parameters over the gadget technique (as will be demonstrated), the gadget access model, when defined correctly, is much more difficult to achieve. In particular, the gadgets presented take a collection of oracles $e^{i \arccos(x_k) \sigma_k}$ and output an approximation of the \emph{unitary} $V = e^{i \arccos((x_0 + \cdots + x_{2^n - 1}) / 2^n) \sigma_x}$. On the other hand, LCU takes the same set of oracles and outputs a $\widetilde{\mathcal{O}}(2^n)$-dimensional block-encoding of the operator $V' = \frac{1}{2^{n}} \left(e^{i \arccos(x_0) \sigma_x} + \cdots + e^{i \arccos(x_{2^n - 1})} \right)$, which is generally non-unitary, and therefore can only be applied to a particular, \emph{known} state $|\psi\rangle$ via amplitude amplification. 

We refer to the former access model as \emph{strong} and the latter as \emph{weak}, to emphasize the clear separation between the two operations realized, despite the fact that their functional form in the particular block being considered is the same.
\end{remark}

To achieve the desired output with precision $\varepsilon$ over $x_0, \dots, x_{2^n-1}$ all within $[-1 + \delta, 1 - \delta]$ with gadgets, we require repeated use of the addition and corrective protocols. In particular, to generate an embeddable sum of input leg variables, we require usage of both the addition and correction protocols, which when composed will result in a gadget of circuit depth $\widetilde{\mathcal{O}}(\zeta n \delta^{-1} \log(\varepsilon^{-1}))$ in order to achieve an $\varepsilon/2^n$-approximation of the desired sum, for $\zeta$ one of the choices of Theorem~\ref{thm:qsp_correction} (this choice also dictates number of required ancilla qubits and success probability). Performing the $n$-fold successive pooling of $2^n$ individual oracles will therefore result in a circuit of depth $\widetilde{\mathcal{O}}((Cn)^n \zeta^n \delta^{-n} \log^n(\varepsilon^{-1}))$, for some constant $C$.

We can now outline the individual methods:

\begin{itemize}
\item \textbf{$2^n$-mean gadget with ancillae (and controlled access)}:
In the case that we make use of the correction protocol which utilizes ancillae and controlled access, it is straightforward to see that the nesting depth of the gadget network is $n$, where at each layer, two legs are paired with each other. Thus, the total number of ancillae used throughout the gadget is $\mathcal{O}(n)$: the same complexity as used by LCU to add $2^{n}$ individual block-encodings. In this case, performing an $n$-fold nesting of these gadgets will result in a circuit of depth $\widetilde{\mathcal{O}}((C_1 n)^n \delta^{-2n} \log^{2n}(\varepsilon^{-1}))$ for some constant $C_1$ as $\zeta = \widetilde{\mathcal{O}}(\delta^{-1} \log(\varepsilon^{-1}))$. Such a gadget will achieve the desired sum as an $\varepsilon$-approximation with probability at least $(1 - \varepsilon/2^{n})^{2^n} \geq 1 - \varepsilon$.

\item \textbf{$2^n$-mean gadget without ancillae}:
In the case that we are able to make use of the ancilla-free correction protocol, $\zeta = \widetilde{\mathcal{O}}(\delta^{-1} \gamma^{-2} \log^2(\varepsilon^{-1}))$. Thus, with unit success probability and no ancilla qubits, the resulting gadget has circuit depth $\widetilde{\mathcal{O}}((C_2 n)^n \gamma^{-2n} \delta^{-2n} \log^{3n}(\varepsilon^{-1}))$ to achieve an $\varepsilon$-approximation of the desired sum.

\item \textbf{$2^n$-mean with LCU}:
Finally, when using LCU, we are able to achieve a block-encoding of $2^n$ individual oracles with circuit depth $\mathcal{O}(2^n)$, and $\mathcal{O}(n)$ individual ancilla qubits. However, in general, this will create a block-encoding in a much higher-dimension space than the gadget-based protocol. Thus, if we wish to \emph{access} this block-encoding, and apply it to a particular state, we will need to perform rounds of amplitude amplification. 

In particular, suppose we wish to construct a quantum circuit which applies the sum $V'$ of Rem.~\ref{rem:incomp} to a particular single-qubit state $|\psi\rangle$ with success probability $1 - \varepsilon$ (the same as the gadget case). 
The initial success probability will be $|\langle \psi | (V')^{\dagger} V' |\psi\rangle|^2$. Generally speaking, the overlap can be very small, if the sum of the operators being computed is close to $0$. Since we are assuming that $\arccos(x_k) \in [0, \pi]$ for each $k$, this sum can only be close to $0$ if the $x_k$ lie too close to $-1$ and $1$. Indeed, consider the case where of the $2^n$ values of $x_k$ being summed, $2^{n - 1}$ of them are $1 - \delta$ and $2^{n - 1}$ are $-1 + \delta$: the resulting average of $\sigma_x$-operations will be $\delta$-close to $0$, and the success probability will be $\delta$-small. More formally, we restrict $x_k \in [-1 + \delta, 1 - \delta]$, and in the worst case (the case of the previous sentence), we have
\begin{equation}
    V' = \begin{pmatrix} 0 & i \sqrt{1 - (1 - \delta)^2} \\ i \sqrt{1 - (1 - \delta)^2} & 0 \end{pmatrix} = i \sqrt{2\delta - \delta^2} \sigma_x
\end{equation}
and the resulting success probability will be lower-bounded by $\delta$. By the variant of fixed-point oblivious amplitude amplification of \cite{gslw_qsvt_19}, to boost the success probability to $1 - \varepsilon$, we require $\widetilde{\mathcal{O}}( \delta^{-1} \log(\varepsilon^{-1}))$ executions of the LCU circuit, for a total circuit depth of $\widetilde{\mathcal{O}}(2^n \delta^{-1} \log(\varepsilon^{-1}))$.
\end{itemize}

\begin{table}[htpb]
    \centering
    \begin{tabular}{l | l l l l l}
        Method & Depth \ \ \; & Anc. \ \ \ & Domain \ \ \ & Succ. \ \ \ & Access model \\\hline
        Gadget w/ anc. \ & $\widetilde{\mathcal{O}}((C_1 n)^n \delta^{-2n} \log^{2n}(\varepsilon^{-1}))$ & $\widetilde{\mathcal{O}}(n)$ & $\mathcal{D}$ & $1 - \varepsilon$ & Strong \\
        Gadget w/o anc. \ & $\widetilde{\mathcal{O}}((C_2 n)^n \gamma^{-2n} \delta^{-2n} \log^{3n}(\varepsilon^{-1}))$ & None & $\mathcal{D}'$ & $1$ & Strong
        \\
        LCU & $\widetilde{\mathcal{O}}(2^n \delta^{-1} \log(\varepsilon^{-1}))$ & $\widetilde{\mathcal{O}}(n)$ & $\mathcal{D}$ & $1 - \varepsilon$ & Weak
    \end{tabular}
    \caption{A table summarizing the differences between utilizing gadgets vs. LCU to construct the $2^n$-mean block-encoding, as presented above. Note that $\mathcal{D} = [-1 + \delta, 1 - \delta]$. In addition, $\mathcal{D}' \subset \mathcal{D}$, determined on a case-by-case basis, as the domain is determined by the half-twisted embeddability of the constituent variables being summed. For more information, see Rem.~\ref{rem:add_mult_embed}. As discussed in the main text, the apparent improvement by LCU made here is because the gadget-achieved block-encoding in the gadget setting can uniquely be used coherently as a \emph{unitary} (where the mean has been computed \emph{inside the phase}) by another quantum protocol.}
    \label{tab:mean_lcu_comparison}
\end{table}


\section{Functional programming, QSP/QSVT, and gadgets} \label{appx:functional_programming}

\noindent This paper does not construct a quantum programming language, but it does denote a series of operations within quantum computing in a succinct, abstracted style. It will turn out that these operations are most cleanly interpreted \emph{functionally}; by \emph{functional} we mean that each statement in our notation (either atomic or composite) operates by transforming a set of inputs to a set of outputs \cite{selinger_qpl_04, selinger_higher_order_04}; commonly this is contrasted with \emph{imperative} operations, by which one sequentially updates a collection of global variables. In a functional setting, one builds programs through the application and composition of functions; this is a so-called declarative paradigm, by which one describes the logic of a program rather than its control flow. In our setting, function definitions are trees of expressions mapping values to other values, rather than a sequence of statements updating some program state. In practice, there is some debate and disagreement over the best characterization of functional programming conceits, as well as programming languages which mix and match between paradigms; however, for our purposes, we will truly consider only the construction of composite functions.

A defining feature of most functional programming languages is the treatment of functions as first class objects; they can be passed as arguments, bound to names, and returned by other functions, just as any other data type. In this, small functions can be combined in a modular style to yield composite processes of greater complexity. The claim of functional programming (and especially its strongly typed and pure sub-classes) is that such constraints mitigate \emph{side-effects} (modifications of the program state outside the local environment) and thus aid program correctness and programmer reasoning. Functional programming has longstanding ties to theoretical computer science through the lambda calculus, and underlies a variety of general purpose and academic programming languages, realizing many of its purported benefits.

We have claimed that the constructions in this work allow us to discuss QSP and QSVT in the context of \emph{functional programming}. Interestingly, in contrast to the classical case, where many aspects of programming language design can appear syntactical (or at worst superficial), we show that unique aspects of coherent quantum computation promote a constructive role for techniques of functional programming; the protocols we provide, in satisfying required properties of common objects in functional programming, rooted in methods in category and type theory, can immediately be reasoned about (and optimized over) using a wealth of preëxisting (classical) literature.

Before discussing the connection between functional programming and our setting, we provide a remark on the differences between QSP and QSVT, which are largely cosmetic in the context of our results. In this way, the rest of this section (and the paper as a whole) can concern itself with the simpler, single-qubit setting, and results can be lifted to QSVT if desired, up to the conditions discussed below.

\begin{remark}[On differences between QSP and QSVT] \label{remark:qsp_vs_qsvt}
    In much of this paper we refer to QSP and QSVT jointly; this is no coincidence, as the latter is often viewed as a lifted version of the former \cite{gslw_qsvt_19}, in which QSP occurs within privileged small subspaces spanned by the singular vectors of carefully engineered linear operators (this can be understood as a consequence of Jordan's Lemma \cite{jordan_75, regev_06}, or more simply the cosine-sine decomposition \cite{cs_qsvt_tang_tian_23}). For our purposes, most results valid for QSP are also valid for QSVT within these preserved subspaces, and indeed it is this simultaneous manipulation of (possibly exponentially many) eigenvalues or singular values that often underlies exponential advantage for QSVT-based algorithms, e.g., those for Hamiltonian simulation or matrix inversion \cite{gslw_qsvt_19, mrtc_unification_21}. There are caveats to this statement, many addressed in recent work on dequantizing QSVT under low-rank assumptions \cite{chia_20}, but in the general case we assume that the difficulty of computing a linear operator's SVD underlies the classical hardness.

    Despite technical analogousness, presentations of QSP and QSVT often switch between pictures based on interleaved rotations versus interleaved reflections. Following \cite{gslw_qsvt_19, mf_recursive_23}, we repeat this distinction here, to be used for anyone wishing to convert this work (which uses the rotation convention throughout) to the reflection picture (or vice versa).
        \begin{align}
            \Phi[U] &\equiv e^{i\phi_0\sigma_z} \prod_{k = 1}^{n} U e^{i\phi_k\sigma_z}\quad\text{(Rotation)},\\
            \Phi[V] &\equiv e^{i\phi_0\sigma_z} \prod_{k = 1}^{n} V e^{i\phi_k\sigma_z}\quad\text{(Reflection)},
        \end{align}
    where the rotation $U$ and reflection $V$ are parameterized by $x \in [-1, 1]$ in the usual way:
        \begin{equation}
            U \equiv 
            \begin{bmatrix}
                x & i\sqrt{1 - x^2}\\
                i\sqrt{1 - x^2} & x
            \end{bmatrix},
            \quad\quad
            V \equiv 
            \begin{bmatrix}
                x & \sqrt{1 - x^2}\\
                \sqrt{1 - x^2} & -x
            \end{bmatrix}.
        \end{equation}
    Converting between these two pictures following from the simple identity $V = i e^{-i (3\pi/4)\sigma_z}Ue^{I(\pi/4)\sigma_z}$, from which we see that the phases $\Phi$ for an equivalent rotation-based protocol can be derived, given phases $\Psi$ for a reflection-based protocol, by
        \begin{align}
            \phi_k &= \psi_k - \pi/2, \quad k \in \{1, \cdots, n - 1\},\\
            \phi_0 &= \psi_0 - 3\pi/4,\\
            \phi_n &= \psi_n + \pi/4.
        \end{align}
    where the two unitaries will agree up to an unimportant overall phase depending only on $n$. For all discussions of antisymmetric protocols, we will refer to the reflection picture, though analogous conditions can be given for the reflection picture using the above map. Up to this distinction, and within the rank of a given non-square operator transformed by QSVT, we can treat our results on semantic embedding identically to the QSP case.
\end{remark}

The naturalness of applying functional programming to our setting stems from the character of QSP and QSVT. In these algorithms a succinct, classical data structure (a list of real phases) corresponds to a succinct, classical mathematical object (a polynomial transformation). However, the method by which this transform is enacted (namely, the aspect of QSP and QSVT which requires a quantum computer), means that the QSP phases and the realized transform are only of practical use when wrapped in a unitary operation (i.e., the circuit realizing the QSP or QSVT protocol). The input and output of a QSP and QSVT protocol are thus in practice, depending on the choice of the algorithmist, very different. While we might like to reason purely about the application, manipulation, and composition of polynomial functions, it is not obvious how the pre-image of the quantum circuit realizing such transforms, while respecting the contiguity of subroutines, behaves, or whether such a pre-image can even exist.

It turns out that it is precisely the business of functional programming to specify patches for the combination of subroutines such that disparate inputs and outputs can be automatically wrapped, handled, and verified. In turn this restores high-level algorithmic reasoning in an evidently safe and intuitive way.

The composition of QSP- and QSVT-based subroutines at the level of their achieved functional transforms can be seen in the single-variable setting in \cite{rc_semantic_alg_23, mf_recursive_23}, where it is noted that the obvious method for composing circuits \emph{does not} induce composition of the achieved transform save under special global conditions on the QSP phases. The subtlety of this condition is made more evident in \cite{mf_recursive_23}, where the requirement of real functional transforms is not explicitly mentioned, perhaps under the assumption that the non-real case merely analytically continues the achieved composition, which as shown in \cite{rc_semantic_alg_23} is generally untrue.

Before discussing our case more specifically, we remark that not only does this work not introduce a full quantum programming language, but that there has been extensive research on the design of (especially functional) full quantum programming languages \cite{selinger_qpl_04, selinger_higher_order_04, pagani_higher_order_semantics_14, ag_functional_qpl_05}, for which the reader is directed to consult an excellent high-level survey in \cite{gay_qpls_06}. For the most part these programming languages rely on sophisticated mathematical techniques to address the unique features in quantum computation (e.g., entanglement, superposition, measurement collapse) as they pertain to program compilation, verification, control flow, and general design considerations. Nevertheless, the general difficulty of algorithm design in such expansive languages remains, albeit in a modified, organized form. In our work, we consider only a limited class of (surprisingly expressive and historically useful) quantum computing algorithms, and borrow results from functional programming theory to characterize their combination in a high-level and interpretable way.

\subsection{Natural transformations and the QSP gadget} \label{appx:natural_transformations}

The usefulness of QSP and QSVT as quantum algorithms is founded in reasons so clear that they become difficult to see. Mainly, instead of quantum circuits, these algorithms allow one to consider continuous functions—said functions, in turn, are guaranteed to be applicable to the spectra of large linear operators. This action of moving from a circuit picture to a functional picture hides a series of homomorphisms. In this subsection we give these homomorphisms (of which there are many) names, and discuss their role in constituting (to borrow a category-theoretic term \cite{cat_theory_78}) \emph{natural transformations}. As discussed in previous work \cite{rc_semantic_alg_23}, the difficulty arising from these homomorphisms is that various manipulations in their images may not correspond to sensible manipulations in their pre-images. Understanding when such analogies can be formally restored (to borrow category-theoretic terms again, when certain diagrams \emph{commute}) is a benefit of this work, and one motivation for applying techniques from functional programming, type theory, and category theory to QSP- and QSVT-based quantum algorithms.

\begin{remark}[Pictures for QSP and QSVT] \label{rem:qsp_qsvt_pictures}
    This remark aims to clarify a misconception possible with the diagrams we present. In the depiction of quantum computational programs, one often relies on the circuit depiction \cite{nc_textbook_11} in which quantum gates (unitary operations) represented by boxes are applied to qubits (or some other finite-dimensional assembled quantum systems), represented by lines, in organized diagrams in which time moves from left to right. Such diagrams clearly depict the gate-level evolution of quantum \emph{states}, taking them as input and producing them as output. One can also consider the depiction of (specifically QSP- and QSVT-based) algorithms which take as input oracular unitary processes. In this case, boxes depict gadgets (e.g., a QSP or QSVT superoperator), and lines indicate input and output \emph{quantum processes}. This work only lightly considers the direct depiction of quantum circuits in terms of their gates, and instead primarily focuses on the combination and composition of \emph{higher-order quantum processes} whose inputs and outputs are themselves quantum processes.

    The functional programming character of this work stems from the observation that there are two natural levels from which to consider the combination of higher-order processes. The first is in gluing together the inputs and outputs of these processes in a diagrammatic way, as discussed in Fig.~\ref{fig:qsp_gadget_composition}. The second concerns the underlying scalar transformations achieved by these superoperators through the language of QSP and QSVT; these polynomial functions can also be manipulated and composed, and it is the business of this work to ensure that actions in this \emph{functional space} (i.e., functional composition, which is semantically clear) are reflected by simple actions in \emph{gadget space} (i.e., superoperator composition). In this way, algebraic manipulations are reduced to geometric ones, and both are guaranteed to have easy to compute circuit pre-images.

    At the end of the day, we still desire that these computations are achievable at the gate level, respecting the contiguity of subroutines used in black-box ways. It is precisely the statements of the main theorems of this work that this achieveability is summarized by the valid linking of superoperators in the first picture discussed above. Consequently, if the \emph{functional space} has a valid diagram (a pre-image) in \emph{gadget space}, then our desired functional manipulation is guaranteed to be both realizable by a simple circuit, and have well-characterized runtime and resource requirements.
\end{remark}

Previous work by two of the authors of this work discussed the limited combination and composition of single-variable QSP and QSVT protocols; in this case the major category-theoretic object referenced was the \emph{natural transformation}, the properties of which gave name and structure to the major theorems of that work \cite{rc_semantic_alg_23}. While the section following this one will primarily refer to objects in functional programming proper, the connection between category theory and common functional programming objects is strong enough to be of use (and we will refer back to the diagram given in Fig.~\ref{fig:qsp_as_natural_transformation} when discussing definitions of monads in the next section, which can also themselves be defined in terms of diagrams, per those of Def.~\ref{def:natural_transformation_monad}).

Natural transformations can be thought of as a method to discuss relations between functors (they are \emph{morphisms between functors}), where functors are themselves objects which translate category theoretic diagrams to analogous diagrams (they are \emph{morphisms between categories}). In \cite{rc_semantic_alg_23}, we formalized the two pictures in Rem.~\ref{rem:qsp_qsvt_pictures} in terms of a series of functors over QSP/QSVT-specific mathematical objects (e.g., QSP phase lists, QSP unitaries, polynomial transforms, etc.), by which the \emph{components} of the natural transformation became the various homomorphisms which allow one to move between objects acted on by various functors. This culminates in definition statement (reproduced from \cite{rc_semantic_alg_23} in Def.~\ref{def:semantic_embedding}) which specifies that the embedding of a functional transformation in a QSP/QSVT protocol is \emph{semantic} if there exist special functors over a category of QSP phase-lists which permit the diagram in Fig.~\ref{fig:qsp_as_natural_transformation} to \emph{commute} (i.e., that any traversal along its arrows yields the same result if its start and end points are the same).

\begin{definition}[Semantic embedding for QSP] \label{def:semantic_embedding}
    Taking the diagram in Fig.~\ref{fig:qsp_as_natural_transformation}, we say that the morphisms $f, g, h$ are \emph{semantically embedding} for QSP protocols if $S, T$ as translated by Table~\ref{tab:qsp_categories} are such that $Sf, Sg, Sh$ correspond to circuit composition (respecting contiguity of the inner phase list), $Tf, Tg, Th$ correspond to polynomial composition, the components of the natural transformation $\tau a, \tau b, \tau c$ are simple projection from a unitary operator to its top-left matrix element, and the diagram given commutes. Moreover, the functions $f, g, h$ should be efficiently computable.
\end{definition}

This appendix is not meant to meticulously cover results in category theory (for which we refer a reader to the appendices of \cite{rc_semantic_alg_23} or a common textbook \cite{cat_theory_78}). Rather, we use this space to present that the conditions we care about in combining QSP/QSVT subroutines as modules are of great interest within category theory, which has already developed advanced structures to describe and analyze them.

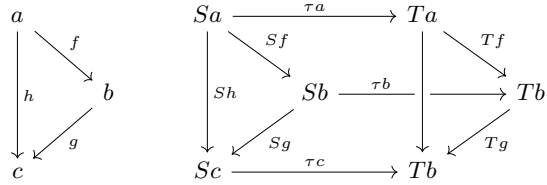
\begin{figure}
    \centering
    \begin{tikzcd}
        a \arrow[rd, "f"] \arrow[dd, "h"] & & Sa \arrow[dd, "Sh"] \arrow[dr, "Sf"] \arrow[rr, "\tau a"] & & Ta \arrow[dr, "Tf"] \\
        & b \arrow[dl, "g"] & & Sb \arrow[dl, "Sg"] \arrow[rr, "\tau b" near start] & & Tb \arrow[dl, "Tg"] \\
        c & & Sc \arrow[rr, "\tau c"] & & Tb \arrow[from=uu, crossing over]           
    \end{tikzcd}
    \caption{A diagram presenting a natural transformation between morphisms $S$ and $T$. Special assignments from quantum algorithms to the objects depicted here provide a concrete definition for semantic embedding in \cite{rc_semantic_alg_23}, from which the assignments in Table~\ref{tab:qsp_categories} are borrowed.}
    \label{fig:qsp_as_natural_transformation}
\end{figure}

\begin{table}[htpb]
    \centering
    \begin{tabular}{l l l}
        Category picture\;\;\; & QSP picture & Description\\[-2.3ex]
        & &
        \\\hline
        \rule{0pt}{0.9\normalbaselineskip}$a, b, c$ & $\Phi$ & \text{QSP phase list}\\
        $f, g, h$ & $\Phi_0 \rightarrow \Phi_1 \circ \Phi_0$ & \text{QSP phase nesting} \\
        $Sa, Sb, Sc$ & $U(\Phi)$ & \text{QSP unitary}\\
        $Sf, Sg, Sh$ & $U(\Phi_0) \rightarrow U(\Phi_1 \circ \Phi_0)$\;\;\;\; & \text{QSP unitary embedding}\\
        $Ta, Tb, Tc$ & $P(x)$ & \text{Polynomial transform}\\
        $Tf, Tg, Th$ & $P_0(x) \rightarrow (P_1\circ P_0)(x)$ & \text{Polynomial composition}\\
        $\tau a, \tau b, \tau c$ & $U(\Phi) \rightarrow P(x)$ & \text{Projection to matrix element}\\
    \end{tabular}
    \caption{A table relating the definition of natural transformations, and the instantiation of this diagram with objects and arrows (functions) in QSP. We suppress multi-labeling of QSP objects in the second column for brevity. A full explanation of the notation given here can be found in the source work \cite{rc_semantic_alg_23}.}
    \label{tab:qsp_categories}
\end{table}

The upshot of this work is that, following the major theorems in the main body, and specifically Thm.~\ref{thm:full_gadget_composition} and its corollaries in Appx.~\ref{appx:main_proofs}, the algorithmist is free to diagrammatically link the QSP gadgets (with possibly many oracular inputs and unitary outputs) in any way permitted by a natural notion of their geometry. In Fig.~\ref{fig:qsp_gadget_composition}, we summarize the two pictures specified above, considering our gadgets as circuits, and then as superoperators, as well as a simple case illustrating the free choices one has to compile simpler gadgets into larger ones. It is the authors' hope that this depiction is simple, and serves to bolster the linear, text-based description of these operations in Appx.~\ref{appx:main_proofs}.

\begin{figure}
    \centering
    \includegraphics[width=0.6\textwidth]{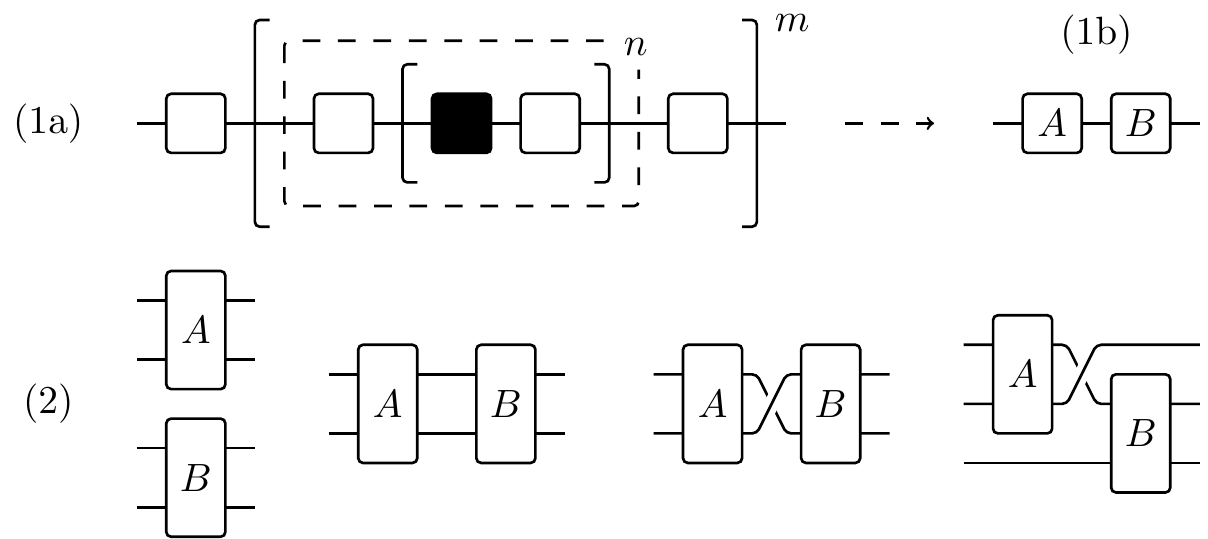}
    \caption{Simplified depiction of the circuit picture (1a) versus the gadget picture (1b) for single-variable recursive embedding of QSP \cite{rc_semantic_alg_23, mf_recursive_23} (equivalently the composition of two $(1,1)$ gadgets). In the multi-input and multi-output setting the gadget picture of (2) offers significant simplification in enumerating partial compositions, each of whose functional actions can be interpreted following similar methods to those of Thm.~\ref{thm:full_gadget_composition}, characterized by \emph{interlinks} (Def.~\ref{def:gadget_interlink}). In this example, $(2,2)$ gadgets $A$ and $B$ linked following a subset of the possible linkages, resulting in composite gadgets of type $(4, 4)$, $(2,2)$, and $(3, 3)$, as shown. For a full enumeration of such partial compositions, see Appx.~\ref{appx:main_proofs}, and for computing their costs, consult Lem.~\ref{lem:cost_matrix_composition}.}
    \label{fig:qsp_gadget_composition}
\end{figure}

\subsection{On the QSP and QSVT types, and their related monadic type} \label{appx:qsp_qsvt_types}

At a low-level, our work is concerned with the following problem: given repeatable access to single-qubit unitaries which are crucially (1) known to look as if they could have (approximately) resulted from QSP protocols with certain mild symmetries, and (2) encode in a specified matrix element a (Laurent) polynomial function of possibly many unknown variables, do there exist techniques to manipulate these matrix elements in well-understood and controllable ways while generating unitaries that respect properties (1) and (2). In other words, we are interested in superoperators that preserve (1)-(2), and more ambitiously (given that such operations are not too difficult to find) in the description of the closure under (1)-(2)-preserving operations. This closure is ideally described in terms of the achieved matrix elements, which are (Laurent) polynomial functions of some fixed set of underlying real-valued inputs.

Our access model is restrictive for an important reason: while our computation lives in SU(2), the underlying unknown parameters are real, and the achieved functions (as matrix elements) are also real and bounded. It is precisely that we project down to fixed matrix elements, and parameterize these same matrix elements by real values from an interval, that allows us to flexibly achieve non-linear transformations of said signals obliviously. However, for most intents and purposes, we are enticed to reason purely in terms of the real functions achieved, and moreover we would like to do this irrespective of the realizing dynamics (under some assurance that they are efficiently implementable, and computable). This scenario is quite common in functional programming, as mentioned, and we take the rest of this section to formalize our desire in terms of the natural object dealing with such scenarios: monads.

Commonly a monad is defined (for programmers) apart from category theoretic concepts \cite{wadler_comprehending_90, wadler_essence_92, wadler_monads_95}; while we will return to a category theoretic definition later, in part to match with the functors and natural transformations defined in \cite{rc_semantic_alg_23}, for the moment we consider a monad as defined in three parts (often called a Kleisli triple): these are a type constructor $M$, a type converter (often called \texttt{unit}), and a combinator (called \texttt{bind} or denoted \texttt{>>=}). The first of these builds up a type, and can be thought of as a specification for how \texttt{unit} wraps its given data (i.e., embeds it with other, surrounding information constituting the monad object), while \texttt{bind} specifies how to unwrap this data, apply a desired function to it, and return the result re-wrapped in the monadic type. These operations have to respect a few basic identities, namely
    \begin{align}
        \texttt{unit}\; x \texttt{ >>= } f &\leftrightarrow f\;x\\
        \texttt{ma} \texttt{ >>= } \texttt{unit} &\leftrightarrow \texttt{ma},\\
        \texttt{ma >>= } \lambda x \rightarrow (f(x) \texttt{ >>= } g) &\leftrightarrow
        (\texttt{ma >>= } f) \texttt{ >>= } g,
    \end{align}
where here $\lambda$ is an anonymous function, and where these conditions formalize that \texttt{unit} is both the left and right identity for \texttt{bind}, and indicate that \texttt{bind} is associative. Ultimately, this construction gives rise to the ubiquitous (if opaque) expression that a monad is a \emph{monoid in the category of endofunctors}, where \texttt{unit} and \texttt{bind} are the identity and binary operation respectively.

In the context of functional programming from a more purely category theoretic perspective, the structure of a monad can also be derived from a functor with two additional natural transformations \cite{moggi_lambda_calc_88, moggi_abstract_89}. We discuss this to make a connection to \cite{rc_semantic_alg_23}, which introduces semantic embedding in category-theoretic terms. In this case, one starts with a higher-order function (or functional/functor) named \texttt{map} with the following type, serving as our base functor:
    \begin{equation}
        \texttt{map}\; \phi : (\texttt{a} \rightarrow \texttt{b}) \rightarrow (\texttt{ma} \rightarrow \texttt{mb}).
    \end{equation}
Consequently, two objects \texttt{join} and \texttt{unit} can be defined in terms of this functor, the latter in satisfying the following identity,
    \begin{equation}
        (\texttt{unit} \circ \phi)\; x \leftrightarrow ((\texttt{map}\;\phi) \circ \texttt{unit})\; x.
    \end{equation}
The \texttt{join} object is a little more involved, and is as promised in category theoretic terms a \emph{natural transformation} \cite{cat_theory_78, lambek_higher_order_88} which allows one to flatten nested applications of the monad. It must satisfy the following three identities:
    \begin{align}
        (\texttt{join} \circ (\texttt{map}\; \texttt{join})) \;\texttt{mmma} &\leftrightarrow (\texttt{join} \circ \texttt{join}) \;\texttt{mmma} \leftrightarrow \texttt{ma},\\
        (\texttt{join} \circ (\texttt{map}\; \texttt{unit})) \;\texttt{ma} &\leftrightarrow (\texttt{join} \circ \texttt{unit}) \;\texttt{ma} \leftrightarrow \texttt{ma},\\
        (\texttt{join} \circ (\texttt{map}\; \texttt{map}\; \phi)) \;\texttt{mma} &\leftrightarrow ((\texttt{map}\;\phi) \circ \texttt{join}) \;\texttt{mma} \leftrightarrow \texttt{mb},
    \end{align}
To convert entirely to the diagrammatic methods of category theory, we can give an equivalent, more diagrammatic definition, breaking away from the identities as presented above (as well as those using \texttt{unit} and \texttt{bind} previously).

\begin{definition}[Monad, following \cite{cat_theory_78}] \label{def:natural_transformation_monad}
    A monad $T = \langle T, \eta, \mu\rangle$ (the aforementioned Kleisli triple) in a category $X$ consists of a functor $T : X \rightarrow X$ and two natural transformations (here denoted by $\dot{\longrightarrow}$
        \begin{equation}
            \eta: I_X \dot{\longrightarrow} T, \quad
            \mu: T^2 \dot{\longrightarrow} T,
        \end{equation}
    which assert that the following diagrams commute
        \begin{equation}
            \begin{tikzcd}
                T \arrow[r, "T_\mu"] \arrow[d, "\mu T"] & T^2 \arrow[d, "\mu"] & IT \arrow[d, equals] \arrow[r, "\eta T"] & T^2 \arrow[d, "\mu"] & \arrow[l, "T\eta"'] TI \arrow[d, equals]\\
                T^2 \arrow[r, "\mu"] & T, & T \arrow[r, equals] & T \arrow[r, equals] & T.
            \end{tikzcd}
        \end{equation}
    Here we can identify $\eta$ with the identity and $\mu$ the multiplication of the monad $T$ built from the endofunctor $T: X \rightarrow X$. In this way, we can recognize the two diagrams as summarizing the conditions previously cited for \texttt{unit} (here $\eta$) and \texttt{join} (here $\mu$), with \texttt{map} (here $T$) acting as the base functor. In short, the left diagram shows associativity for the monad, while the right diagram express the left and right \emph{unit laws}. These diagrams appear often within category theory, unsurprisingly being precisely those of a monoid whose base set (often denoted $M$) is here replaced by the endofunctor $T$.
\end{definition}

It turns out that the functor $T$ defined in previous work \cite{rc_semantic_alg_23} and depicted in Fig.~\ref{fig:qsp_as_natural_transformation} is precisely the endofunctor over which this paper instantiates a monad, as discussed in Rem.~\ref{rem:qsp_monad}, save in the more expansive, multivariable case, in which the realizing dynamics of \texttt{bind} are less simple, and have non-trivial resource cost. In this way, we see that the benefit of a functional programming characterization is not merely syntactic, but more accurately a guide toward a constructive superoperator which, if realized, immediately permits the desired declarative methods of functional programming.

Classically, monads are defined in a number of equivalent ways with differing terms, among the more common being \texttt{unit} (or \texttt{pure}) and \texttt{bind}, as mentioned previously, which ultimately ascribe a method for a programmer to pipeline subroutines together with declarative assurance that the input and output structure of these elements are respected and handled in the expected way (the common terms here are managing \emph{control-flow} and \emph{side-effects}). The common analogy is that a monad allows one to wrap data (here used in the broad, functional programming sense) such that the programmer is freed to reason agnostic to operational details, working purely over types. In our setting, these types are real scalar values and real-argument, real-valued multivariable polynomial functions. We can thus summarize the work of this paper in the following remark, which ties together our constructions and their natural description in terms of monads and their defining properties.

\begin{remark}[Constructing the QSP monad] \label{rem:qsp_monad}
    Given a real-valued polynomial function of possibly multiple variables, this function can be suitably embedded in a single qubit unitary given access to a basic unitary process which encodes the argument(s) for the intended function in the angle of rotation about a known and fixed axis. This encoding is analogous to \texttt{unit}, with the standard type signature:
        \begin{equation}
            \texttt{unit} : T \rightarrow M\;T,
        \end{equation}
    where $T$ is some type (in this case polynomial functions), and where $M T$ is the same function, but now encoded in a unitary process in the standard way (i.e., as the top left element of an antisymmetric QSP superoperator applied to an oracle unitary). In other words, QSP constitutes a monadic type in whose context one can interpret underlying simpler, polynomial functions. The primary goal of this work, then, is to show that even in the multivariable setting, \texttt{bind} is efficiently implementable, where we can write the standard type signature:
        \begin{equation}
            \texttt{bind} : (M\;T, T \rightarrow M\; U) \rightarrow M\;U,
        \end{equation}
    where $U$ is another type, here the same as $T$ in most cases. The primary difficulty before this work was that, although we understood of the nature of the map $T \rightarrow M\; U$ quite well within QSP, i.e., from the unwrapped type $T$ to the wrapped type $M\;U$, the action of wrapping a type through \texttt{unit} was not so obviously reversible in the general case. It just so happens that for single-variable QSP, antisymmetric phases, as discussed in \cite{rc_semantic_alg_23, mf_recursive_23}, make \texttt{bind} nearly trivial to implement. In our setting, applying a broadband-NOT, followed by an approximate square root, permits us to build a relatively efficient subroutine (under some mild bound restrictions on the achieved functions) for recovering \texttt{bind} in such a way as to satisfy the identities for \text{unit}. Moreover, we find that our construction not only permits the repeated self-embedding of multivariable QSP protocols within themselves (while respecting the desired action within $T$), but that it opens a variety of manipulations in $T$ (including manipulations on functors over $T$) to be suitably lifted through $M$ and thus realized within a quantum computation. In this way, the natural, mathematically attractive endofunctors of the space of real-valued polynomials are elevated as first class objects about which the algorithmist can reason unencumbered by circuit considerations.
\end{remark}

We could have also chosen to rewrite this monad in terms of \texttt{unit} and \texttt{join}, following from the natural transformation shown between the two functors over the space of polynomials and unitaries as discussed in \cite{rc_semantic_alg_23}. In either case, the language of monads expediently enumerates the consequences of our constructions. Moreover, given there exist many disparate quantum circuits resembling QSP, which in turn realize disparate (and often useful) transformations in the algorithmically cognizable space of polynomial functions over scalars and operators, similar techniques to those in this work, by the great generalizing power of functional programming, hold promise in almost automatically applying to larger classes of composite quantum programs. Such functional programming methods for the manipulation of higher order quantum processes have been cursorily applied in the black box setting already \cite{oksm_higher_order_23}.

\begin{remark}[Discussions on the QSP and QSVT types] \label{remark:qsp_qsvt_types}
    In the standard definition of a monad, the wrapped value $M\;T$ can be seen as the original data $T$ along with bookkeeping objects designed to keep track of side effects \cite{wadler_comprehending_90, wadler_essence_92, wadler_monads_95}. In our setting, $T$ is a class of well-defined mathematical objects (polynomial functions over scalar variables, possibly with multiple inputs and outputs), while $M\;T$ are quantum circuits (or unitaries) which in some sense realize these transforms over the spectrum of a given oracle unitary. In our case, we are apparently only given the ability to \emph{actually run these programs}, and are thus only ever shown the wrapped monadic type, under the assumption that certain tasks achieved by QSVT do not admit efficient classical algorithms. Therefore, understanding how to apply monadic functions to these \emph{QSP and QSVT types} is the only way to fruitfully manipulate these algorithms as subroutines. It is in this sense that we say functional programming techniques are naturally applied to QSP and QSVT; by specifying a satisfying assignment for the monadic functions, we free the algorithmist to reason about the much simpler multivariable polynomial transforms and their endofunctors. The QSP and QSVT types, in this case specially symmetrized circuits, are effectively identical up to the relation discussed in Rem.~\ref{remark:qsp_vs_qsvt}, and by merit of their restricted structure can be manipulated, applied, and composed without leaving the coherent setting in which much of quantum advantage appears to reside.
\end{remark}

\subsection{The grammar of gadgets} \label{appx:formal_gadget_grammar}

\noindent We have described gadgets as specifying modular computations, themselves possibly comprising simpler, analogous gadgets. It is worthwhile to understand such composite objects in the context of formal techniques from programming language theory; more narrowly, we focus on techniques for the specification of syntax, and methods for coupling said syntax to hierarchical structures within programming languages. While this work does not specify a language for quantum computation generally, it does specify a language for enacting multivariable polynomials on the spectra of linear operators. Even this limited scope has shown great success in quantum algorithms and, in what follows, can be shown to have a substantively complex descriptive theory. In this section we discuss specifications of formal \emph{grammars} and \emph{languages}, connect these to our construction of gadgets, and explore some of the consequences of making this connection.

To keep this section brief, we do not discuss formal grammars from first principles, fronting only relevant objects, and otherwise pointing toward common references for the interested reader. Generally, the study of formal grammars seeks to investigate valid \emph{sentences} built from a specified \emph{alphabet} according to a formal \emph{syntax}. These linguistic terms are no mistake, given the field's roots in the analysis of natural language \cite{chomsky_hierarchy_56}, though such study has ultimately found greater purchase in the theory of programming languages, where one wishes to investigate valid \emph{programs} built from specified \emph{symbols} according to a formal \emph{syntax}. Where possible we try to use terminology from programming language theory, but give common alternative names for objects where useful. We note that such \emph{syntactical} studies do not concern the meaning or \emph{semantics} of programs, only their structural correctness; nevertheless, understanding such structure is vital when writing \emph{parsers} (recognizing and breaking down programs into atomic components) and \emph{optimizers} (simplifying a program according to rigidly permitted rules based in syntax).

\begin{definition}[Context-free grammar (CFG)] \label{def:cfg}
	A \emph{context-free grammar} $G$ is defined by a four-tuple $(V, \Sigma, R, S)$ where
		\begin{enumerate}[label=(\alph*)]
			\item $V$ is a finite set, and each element $v \in V$ is called a \emph{non-terminal character}. Each variable represents a different type of clause (phrase) in the program (sentence).
			\item $\Sigma$ is a finite set of \emph{terminal characters}, disjoint from $V$; this set is the alphabet of the language defined by the grammar of $G$.
			\item $R$ is a finite relation in $V \times (V \cup \Sigma)^*$, where $*$ is the Kleene star operation. The members of $R$ are often called the \emph{productions} of the grammar.
			\item $S$ is the \emph{start symbol} of the program, and is a member of $V$.
		\end{enumerate}
	A production rule $R$ in a CFG is formalized as a pair $(\alpha, \beta)$ where $\alpha \in V$ is a non-terminal character and $\beta \in (V \cup \Sigma)^*$ is a finite string of terminals and non-terminals. Most commonly this production is written as $\alpha \rightarrow \beta$. Note that $\alpha$ can contain \emph{only} non-terminal characters (hence context-free), and that $\beta$ may be $\varepsilon$, the empty string. Commonly productions with the same right-hand-side are written together, with vertical bars separating alternative left-hand-sides, e.g, $\alpha \rightarrow \beta_0 \,|\, \beta_1$.
\end{definition}

\begin{definition}[Context-free language (CFL)] \label{def:cfg_language}
	A \emph{context-free language} of a context free grammar $G = (V, \Sigma, R, S)$ is the set of all terminal symbol strings derivable from the start symbol according to the productions in $R$:
		\begin{equation}
			L(G) = \{w \in \Sigma^* : S \xrightarrow{*} w\},
		\end{equation}
	where $\xrightarrow{*}$ indicates finitely repeated application of production rules (also referred to as a derivation). A language $L$ is said to be a CFL if there exists a CFG $G$ such that $L = L(G)$.
\end{definition}

A problem arises in the use of CFGs in programming language design; as their name suggests, such grammars (and their derived languages) cannot capture agreement or reference, where two different parts of a program must agree with each other in some way. To be sensitive to such distinctions while parsing, additional structure is required which substantially increases the complexity of the resulting languages and their analysis. Given that our gadgets can be of arbitrary size, and their connection involves agreement, e.g., in leg number, across perhaps large regions of the resulting diagram, it will turn our that the generic description of the grammar of gadgets will require such additional structure. We give a minimal presentation of one construction below. Before this, we provide an overarching comment on some common grammars of differing strength.

\begin{remark}[RGs, CFGs, CSGs, and UGs] \label{remark:chomsky_hierarchy}
	In the discussion of grammars, reference is often made to the \emph{Chomsky hierarchy} \cite{chomsky_hierarchy_56}, which specifies four common varieties of grammar in a hierarchy, where the set of languages generated by those grammars higher in the hierarchy strictly contain the set of those generated by lower grammars. These are: regular grammars (RGs), context-free grammars (CFGs), context-sensitive grammars (CSGs), and unlimited grammars (UGs). The defining characteristics of these grammars are the nature of their \emph{productions}, which have increasingly general form as one moves up the hierarchy. Generally, RGs and CSGs achieve the most study in programming language theory, given the efficiency of verifying that a given string exists in the language generated by one of these grammars. Nevertheless, CSGs and UGs have achieved attention historically given documented cases of super-CFG structure in both programming and natural languages. Moreover, there exist a variety of augmented grammars which fall between hierarchical levels or beyond them, as discussed below, again with documented use in programming languages.
\end{remark}

We now introduce an augmented grammar which, while outside of the Chomsky hierarchy, possesses a particularly compact structure that illustrates the sort of bookkeeping necessary in keeping track of the assembly of arbitrary gadgets. We stress that this construction, like all formal grammars, does not capture semantic structures, only syntactic ones. For QSP-derived gadgets, semantic structure concerns the achieved polynomial transforms, while the syntactic structure concerns only how such transforms can be combined. We have already established a diagrammatic formalism for understanding the combination of atomic gadgets, and thus what follows maps directly onto this diagrammatic formalism, allowing one in principle to verify, given a string of symbols from a (in our case possibly infinite) alphabet, whether such a string could have been generated by the specified grammar (and consequently whether a given assemblage of gadgets could have resulted from a valid diagrammatic composition).

\begin{definition}[Van Wijngaarden grammar (W-grammar)] \label{def:vw_grammar}
	W-grammars are two-level grammars, i.e., defined by a pair of grammars, that operate on different levels  \cite{w_grammar_65, w_grammar_12}. The first is a \emph{hypergrammar} (an attribute grammar \cite{knuth_attribute_grammars_68}), which is a CFG whose non-terminal characters may have attributes. The second is a \emph{metagrammar}, which is a CFG whose induced language defines the possibilities for the hypergrammar's attributes (its terminal strings). The language of a W-grammar is defined by a two step process:
		\begin{enumerate}[label=(\alph*)]
			\item Within each hyperrule, for each attribute that occurs, choose a value generated by the metagrammar. The result is a standard CFG rule. Do this for all possible values.
			\item Using the resulting (possibly infinite) CFG, generate strings of terminal symbols in the normal way.
		\end{enumerate}
	W-grammars are known to be Turing complete, and thus membership of a given string in a given W-grammar, or even whether a given W-grammar is empty or not, are undecidable.
\end{definition}

We now specify a W-grammar for gadget assemblages; in this case the metagrammar will define tuples (attributes) limiting the size of the initial gadget, and the size of permitted sub-gadgets into which said initial gadget is allowed to be decomposed. For each attribute, we can then define a CFG whose productions cover the decomposition of one gadget into two coupled gadgets, and the decomposition of an atomic gadget into an assembly of M-QSP protocols. Once more, we capture only the formal structure (syntax) of such protocols, not their achieved functional transforms (semantics).

\begin{example}[A W-grammar for gadgets] \label{ex:gadget_w_grammar} 
    We consider a two-level grammar (equivalently a W-grammar, Def.~\ref{def:vw_grammar}) generating the language of compositions of atomic gadgets. For the hypergrammar we will define a finite series of hyperrules according to the possible attributes of non-terminal characters; in this case, these attributes will be tuples positive integers representing (1) the number of input or output legs on a gadget, and (2) the number of involved legs in an interlink. We give the CFG of this hypergrammar first, and then discuss the CFG generating the attributes present in the hypergrammar. Specifically, consider the following productions:
	\begin{align}
		S &\rightarrow A[a, b, c],&&a, b \leq c,&&\text{(Spawn gadget)}\\
		A[a, b, c] &\rightarrow A[a - e, f, c] B[d, c] A[e + d, b - f + d, c], &&0 \leq e \leq a, 0\leq d \leq f \leq c,&&\text{(Split gadget)}\\
        A[a, b, c] &\rightarrow A[a, d, c]\,|\, A[a, e, c]\,|\, A[f, b, c], &&d \in [b], b \leq e \leq c, f \in [a],&&\text{(Mod. gadget)}\\
        A[a, b, c] &\rightarrow C[a, b, c],&&&&\text{(Atomize)}\\
		C[a, b, c] &\rightarrow \phi_k C[a, b, c]\,|\, \phi_k, &&k \in [b],&&\text{(Add phase)}\\
		C[a, b, c] &\rightarrow \zeta_{j,k} C[a, b, c]\,|\,\zeta_{j,k}, &&j \in [a], k \in [b],&&\text{(Add oracle)}\\
		B[d, c] &\rightarrow B[d, c]\sigma_{j,k} \,|\,\sigma_{j,k}, &&j, k \in [d],&&\text{(Permute links)}\label{eq:prod_permute}
	\end{align}
	where one can think of the non-terminal $A[a, b, c]$ as a gadget with $a$ input legs, $b$ output legs, and which is allowed comprise gadgets with at most $c$ input or output legs internally. Here $B[d]$ represents an interlink involving $d$ legs. $S$ is the start symbol as usual. Note that for fixed positive integers $a, b, c$, there are only finitely many productions, and thus this is a true CFG under for each valid attribute. The terminal symbols, $\phi_k, \zeta_{j,k}, \sigma_{j,k}$ can be thought of as M-QSP phases for the $k$-th output leg, uses of the $j$-th M-QSP oracle for the $k$-th output leg, and a permutation between the $j$-th and $k$-th operating qubits respectively.

    We have named each of the productions of the hypergrammar, which allow one to break down gadgets into smaller gadgets (splitting), apply common gadget superoperators (modification) such as elision, augmentation, and pinning, convert a gadget non-terminal symbol to an atomic gadget non-terminal symbol (atomization), and parameterize a series of M-QSP protocols within an atomic gadget. Note that for brevity we do not include a production enforcing antisymmetrization of the phases (though this is permissible with a CFG), and moreover note that this grammar does not account for atypical gadgets. Again for simplicity, note that permutation is a subset of interlinks covered by the final production.
 
    In sum, this grammar allows derivation of unambiguous prescriptions for \emph{assemblages of atomic gadgets}; these are distinct from the \emph{quantum circuit realizing this assemblage}, and exemplifies the difference between semantic and syntactic aspects of language. Note also that this prescription does not capture the semantic relation between gadget legs (for instance in augmentation, elision, or pinning). In other words, the metagrammar is merely present to capture context-sensitive aspects of the resulting grammar necessary to constrain which gadget geometries can be combined. This mix of certain semantic and syntactic information while parsing is a hallmark of attribute grammars, and their main utility \cite{knuth_attribute_grammars_68}.

	The metagrammar we desire to build is one which generates finite tuples of positive integers. This is pretty simple to do, as we only need pairs and triples of positive integers, and so we consider the following productions:
		\begin{align}
			S &\rightarrow A \,|\, B,&&\text{(Spawn tuple)}\\
			A &\rightarrow 0 \,|\, 1 \,|\, 0A \,|\, 1A,&&\text{(Add to tuple)}\\
			B &\rightarrow 0 \,|\, 1 \,|\, 2 \,|\, 0B \,|\, 1B \,|\, 2B,&&\text{(Add to tuple)}
		\end{align}
	where we have condensed productions from $A$ and $B$ in single lines. The resulting tuple can be determined by counting the number of each symbol in the resulting string, with order within the tuple specified by the natural numerical order of the symbols. Together, this and the attribute-augmented CFG define a W-grammar whose induced W-language unambiguously identifies an assemblage of atomic gadgets.
\end{example}

\begin{remark}[On the notation in Ex.~\ref{ex:gadget_w_grammar}]
    We note that the production rules provided in the hypergrammar of Ex.~\ref{ex:gadget_w_grammar} uses some constants relied on elsewhere in this manuscript. As this subsection is self contained, we intend for these constants to be used only here. Nevertheless, the $a, b$ used in the productions do indeed align with the $(a, b)$ gadget notation. The positive integer $c$, the maximum number of input or output legs (equivalently, \emph{arity/co-arity} or \emph{valence/co-valence}) of a gadget within a given CFL, is not mentioned or relevant to other sections of this work. Likewise $B[d, c]$ refers to interlinks over $d$ legs for gadgets with maximum arity $c$, while $C[a, b, c]$ are antisymmetric atomic gadgets, again with maximum arity $c$. We assume throughout derivations that input and output legs are in a definite, fixed order when instantiated, and all operations act with respect to this order; ultimately this is unimportant given the existence of the permutation production (\ref{eq:prod_permute}).
\end{remark}

A neat corollary to the above example follows from recognizing that the CFG of the hypergrammar can be considered on its own, for a fixed attribute. We discuss this in the example below, noting that most of the examples discussed in the main body of this work are low-degree gadgets.

\begin{example}[A CFG for limited gadgets] \label{ex:limited_gadget_cfg}
    It is worthwhile noting that when one is allowed to consider only gadgets of fixed maximum input and output leg number, the construction of the hypergrammar in Ex.~\ref{ex:gadget_w_grammar} demonstrates that a true CFG results. Consequently, if we restrict to this setting, we can recover many of the desirable properties of CFGs, including algorithms for efficiently verifying the correctness of a gadget assemblage. It is not clear, however, whether there is a concise way of understanding this more limited model in comparison to the general setting.
\end{example}

While general composite objects built from atomic gadgets appear to be describable only by relatively complex grammars outside the Chomsky hierarchy \cite{chomsky_hierarchy_56}, it is worthwhile to consider whether useful sub-grammars for gadgets can be specified. For instance, in the case where one only considers $(1, 1)$ and $(2, 1)$ gadgets, or in fact any finite number of gadget arities, as well as restrictions on the possible assemblages of gadgets to trees or other cycle-free directed graphs, it should be possible to completely capture wide functional expressivity while maintaining efficient algorithms for the verification of membership in the induced language. There exists a wealth of literature on augmented CFGs which nevertheless maintain nice verification properties, for instance attribute grammars \cite{knuth_attribute_grammars_68}. This in turn may permit the applicability of highly developed classical tools for compiler design and program optimization \cite{lsua_compilers_06} to the theory of QSP-based gadgets. 

However, specification of a minimal sufficient grammar for gadgets is not the main focus of this paper, and understanding optimization over such grammars in general may also depend on substantive problems in the theory of polynomial approximation and decomposition, given the relative complexity of the map between a specification of an assemblage of gadgets, and a specification of a quantum circuit resulting from that assemblage. Nevertheless, we believe that the concrete examples above support the validity of considering higher-order abstractions in the design of quantum algorithms for block encoding manipulations.
\end{appendix}

\end{document}